\documentclass[aps,prx,twocolumn,superscriptaddress,floatfix,citeautoscript,preprintnumbers]{revtex4-2}

\usepackage{graphicx}
\usepackage{amssymb,amsmath,amsthm,mathrsfs,amsfonts,mathtools,comment}
\usepackage{hyperref}
\hypersetup{
  colorlinks,
  citecolor=blue,
  filecolor=blue,
  linkcolor=blue,
  urlcolor=blue
}
\usepackage[capitalise]{cleveref}
\usepackage{bm}
\usepackage{physics}
\usepackage{braket}
\usepackage[flushleft]{threeparttable}
\usepackage{booktabs}

\usepackage{makecell}
\usepackage{caption}
\usepackage[labelfont=bf]{subcaption}

\newlength{\fskip}
\newlength{\fgap}
\newcommand{\overlaption}[2][]{
  \centering
  \setbox0=\vbox{#2}\fskip=\ht0%
  #2\vspace{-\fskip}   \caption{}\label{#1}%
  \vspace{\dimexpr\fskip-\baselineskip-\belowcaptionskip-\fgap}}%

\usepackage{orcidlink}
\usepackage{dsfont}
 \usepackage{multirow} 
 \newcommand{\dInt}{{\mathrm{d}}}

\renewcommand{\tr}{{\mathrm{Tr}}}

\newcommand{\lind}{{\mathcal{L}}}
\renewcommand{\norm}[1]{\lVert #1 \rVert}

\newcommand{\varO}{\mathcal{O}}

 \newcommand{\I}{\mathrm{i}}
\newcommand{\mc}{\mathcal}
\newcommand{\wt}{\widetilde}
\newcommand{\wh}{\widehat}
\newcommand{\bv}{\mathbf}

\newcommand{\diag}{\operatorname{diag}}
\renewcommand{\Re}{\operatorname{Re}}
\renewcommand{\Im}{\operatorname{Im}}

\newcommand{\note}[1]{\overset{\text{#1}}}

\newcommand{\mf}[1]{\mathfrak{#1}}

\newcommand{\tnorm}[1]{{\left\vert\kern-0.25ex\left\vert\kern-0.25ex\left\vert#1\right\vert\kern-0.25ex\right\vert\kern-0.25ex\right\vert}}

\newcommand{\RR}{\mathbb{R}}
\newcommand{\CC}{\mathbb{C}}

\newcommand{\ZZ}{\mathbb{Z}}

\newtheorem{thm}{\protect\theoremname}
\newtheorem{lem}[thm]{\protect\lemmaname}
\newtheorem{rem}[thm]{\protect\remarkname}
\newtheorem{prop}[thm]{\protect\propositionname}

\providecommand{\definitionname}{Definition}
\providecommand{\assumptionname}{Assumption}
\providecommand{\corollaryname}{Corollary}
\providecommand{\lemmaname}{Lemma}
\providecommand{\propositionname}{Proposition}
\providecommand{\remarkname}{Remark}
\providecommand{\theoremname}{Theorem}

\makeatletter
\renewcommand*\env@matrix[1][\arraystretch]{%
        \edef\arraystretch{#1}%
        \hskip -\arraycolsep
        \let\@ifnextchar\new@ifnextchar
        \array{*\c@MaxMatrixCols c}}
\makeatother
\renewcommand{\hat}{\wh}
\DeclareMathOperator\arctanh{arctanh}

\begin{document}

\title{Dissipative phase decision without ground-state preparation}

\begin{abstract}
 We propose a dynamical approach to identifying ground-state quantum phases through short-time dissipative cooling. Rather than determining the phase by preparing highly accurate approximations to ground states, we prepare a representative state of a candidate phase and monitor the early-time response of phase-sensitive observables under cooling dynamics tailored to the target Hamiltonian. For a class of phase-decision problems in which the relevant observables can be inferred from the low-energy manifold, and with jump operators implementable using only short-time Hamiltonian simulation, the dissipative evolution rapidly suppresses high-energy components and drives the system into a low-energy manifold whose observables already reveal the underlying ground-state phase, well before mixing to the steady state. We demonstrate this strategy for the frustrated $J_1$--$J_2$ Heisenberg chain, the Kitaev honeycomb model, and the XXZ chain, including Berezinskii--Kosterlitz--Thouless and topological phase transitions. In particular, coarse filter resolutions and short evolution times suffice to recover phase-sensitive quantities such as the Luttinger parameter and topological diagnostics. We further provide theoretical justification that cooling dynamics with such jump operators can rigorously prepare low-energy manifolds for free-fermionic and free-bosonic systems, and investigate this mechanism for interacting fermionic systems. Our results suggest that phase decision is a plausible target for future utility-scale studies on early fault-tolerant quantum devices.
\end{abstract}

\author{Hao-En Li}
\thanks{These authors contributed equally to this work.}
\affiliation{Department of Mathematics, University of California, Berkeley, California 94720, USA}
\author{Yilun Yang}
\thanks{These authors contributed equally to this work.}
\affiliation{Department of Mathematics, University of California, Berkeley, California 94720, USA}
\author{Lin Lin}
\email{linlin@math.berkeley.edu}
\affiliation{Department of Mathematics, University of California, Berkeley, California 94720, USA}
\affiliation{Applied Mathematics and Computational Research Division, Lawrence Berkeley National Laboratory, California 94720, USA}

\date{\today}							
\maketitle

\section{Introduction}

Quantum many-body systems often exhibit a wide range of competing phases and phase transitions even at zero temperature~\cite{Sachdev_2011,Carr2010}. Understanding such a ground-state phase diagram involves two distinct tasks. The first is to identify the relevant competing phases, and the second is to decide, for a given parameter point, which phase is realized. In many settings one already has plausible {candidate phases}, suggested for instance by effective theories, numerical evidence, or experiments, but deciding between them remains costly. We refer to the second task as the \emph{phase-decision} problem, whose precise definition will be given in \cref{sec:models}. 

Our approach is motivated by the following classical example. To decide which of ice and liquid water is thermodynamically preferred at a given temperature, one could compare static properties such as the free energies of the two phases, as illustrated in \cref{fig:phase}a. Another approach is to start from the ice phase and ask whether it begins to melt. Once melting sets in, one need not wait for full equilibration to infer that the liquid phase is thermodynamically preferred. This dynamical perspective is illustrated in \cref{fig:phase}b. Computationally, one can monitor qualitative changes in microscopic observables during melting, such as the pair correlation function~\cite{Widom2002}.

To solve the quantum phase-decision problem, one can again take either a static or a dynamical perspective, as illustrated in \cref{fig:phase}c,d. The static perspective prepares representative states for the candidate phases, estimates their energies accurately, and compares them. This static perspective has been widely used in classical quantum many-body computation~\cite{LeBlancAntipovBeccaEtAl2015,QinChungShiEtAl2020,XuChungQinEtAl2024,GongZhuShengEtAl2014,HaghshenasLanGongEtAl2018,NomuraImada2021,LiuGongLiEtAl2022,HohenadlerLangAssaad2011,SatoAssaadGrover2018,ViterittiRendeSachdevCarleo2026}. On a quantum computer, {quantum phase estimation (QPE) can estimate energies with high precision,} but it requires an initial state with sufficient overlap with the relevant eigenstate and circuit depth long enough to resolve the target energy splitting, which is demanding for small-gap phase decisions~\cite{AbramsLloyd1999,LeeLeeZhaiEtAl2023,DingLin2023}. Adiabatic state preparation provides another route, but it can also become costly near phase boundaries where small gaps control the required evolution time~\cite{farhi2000quantumcomputationadiabaticevolution,Albash2018Adiabatic}.

This work instead explores the \emph{dissipative} analogue of the dynamical perspective. Rather than resolving small energy differences between carefully prepared approximations to competing ground states, we ask whether short-time dissipative cooling of a representative state already contains enough information to decide the phase. In this picture, one starts from a representative state of one candidate phase, evolves it with dissipative cooling dynamics tailored to the target Hamiltonian~\cite{Diehl2008driven,Kraus2008Markov,Verstraete2009Dissipative,Zhou2021DissipativeMPS,cubitt2023dissipativegroundstatepreparation,Langbehn2024Cooling,DingChenLin2024,MotlaghZiniArrazolaEtAl2024,LambertCirioLinEtAl2024,MiMichailidisShabaniEtAl2024,EderFinzgarBraunEtAl2025,LloydMichailidisMiEtAl2025,LiZhanLin2025,ZhanDingHuhnEtAl2025,ding2025endtoendefficientquantumthermal,Molpeceres2025cooling,feldmeier2026digitaldissipativestatepreparation,molpeceres2026benchmarkquantumalgorithmsground}, and monitors phase-sensitive observables. If these observables already enter the phase region associated with the ground state, the phase decision can be made before the dynamics reaches the steady state. This viewpoint is close in spirit to earlier works that use quench dynamics to probe phases~\cite{HaldarMallayyaHeylEtAl2021,KobayashiMotome2025,BhattacharyyaDasguptaDas2015,TarnowskiUnalFlaschnerEtAl2019,WangZhangChenEtAl2017}, which aim at using short-time nonequilibrium response to distinguish phases.
Dissipative protocols differ from quench probes because the evolution is biased toward lower energies, and become particularly attractive when the phase decision depends on low-energy information, but not on preparing a nearly exact ground state.

\begin{figure*} 
    \centering
    \includegraphics[width=0.8\textwidth]{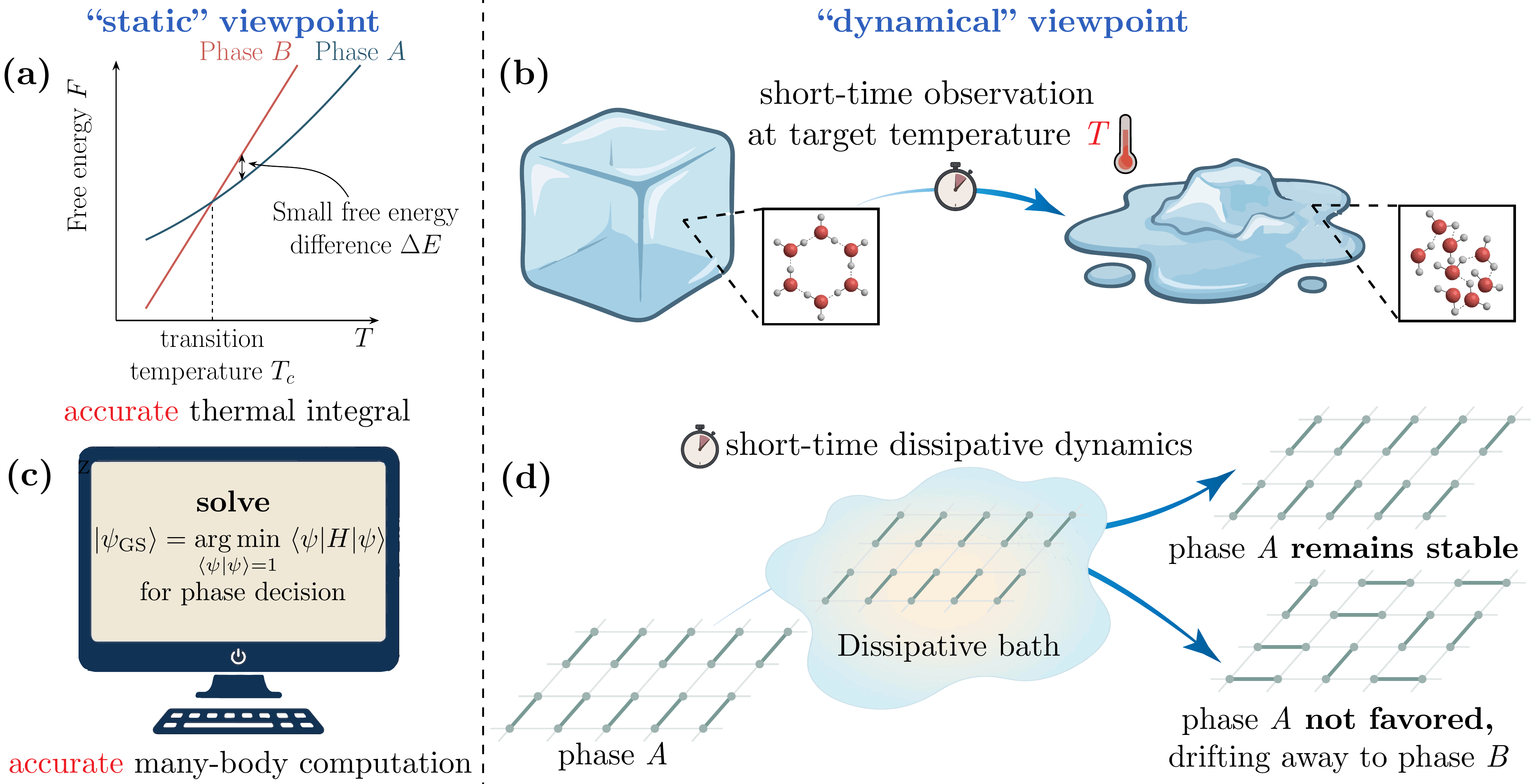}
    \caption{Conceptual illustration of the static and dynamical viewpoints for solving the phase-decision problem. {\bf (a)} Static perspective in the classical setting based on free-energy comparison. {\bf (b)} Dynamical perspective in the classical setting based on short-time response. By observing the onset of melting over a short time window, or computationally by monitoring qualitative changes in microscopic observables such as pair-correlation functions, this approach infers the thermodynamically preferred phase from early-time nonequilibrium dynamics. {\bf (c)} Static perspective for the quantum many-body phase-decision problem. Preparing representative states for competing phases and evaluating their energies with high precision can become computationally expensive near phase boundaries. {\bf (d)} Dissipative dynamical perspective for the quantum phase-decision problem. Dissipative cooling tailored to the target Hamiltonian drives a representative state toward the low-energy sector, while the resulting short-time evolution already reveals qualitatively correct phase-sensitive observables well before the system reaches its steady state. }\label{fig:phase}
\end{figure*}

We emphasize that this dynamical perspective is not intended to solve arbitrary ground-state phase-diagram problems. Classical dynamics may be slowed by metastable behavior such as nucleation barriers~\cite{Langer1969,vanKimuraRuszelSpitoni2019}, and quantum many-body dynamics can likewise exhibit metastability~\cite{YinSuraceLucas2025}. Some phase decisions also require resolving closely competing low-energy states with high accuracy, and in the worst case even determining whether a many-body system is gapped or gapless is undecidable~\cite{Cubitt2015Gap,Bausch2021Uncomputability}. We thus focus on a narrower setting, in which dissipative cooling rapidly suppresses high-energy components and reaches a low-energy sector whose observables already distinguish the competing phases. Because the phase-decision task asks only for a binary answer, it may be more robust to noise than full state preparation, and is therefore a plausible target for future utility-scale studies on early fault-tolerant devices.

More technically, our dissipative protocol is based on algorithmically engineered Lindblad dynamics~\cite{Davies1974,Lindblad1976,GoriniKossakowskiSudarshan1976}. Recent progress has produced efficient dissipative protocols for Gibbs state preparation~\cite{Temme2011Metropolis,MozgunovLidar2020,ChenBrandao2021,shtanko2021preparing,ChenKastoryanoBrandaoEtAl2025,RallWangWocjan2023,gilyen2024quantum,JiangIrani2024,DingLiLin2025,SmidMeisterBertaMario2025,ChenDingZhang2026,ScandiAlhambra2026} and for ground-state preparation beyond frustration-free settings~\cite{DingChenLin2024,LiZhanLin2025,ZhanDingHuhnEtAl2025,ding2025endtoendefficientquantumthermal}. These works typically aim at accurate state preparation. To this end, they often require precise dissipative implementations, including jump operators with enough energy resolution to distinguish low-lying excitations, together with evolution up to the mixing time, namely, the time required to drive arbitrary initial states close to the stationary state~\cite{TemmeKastoryanoRuskaiEtAl2010,kastoryano2013quantum,Bardet2023rapid,gamarnik2024slowmixingquantumgibbs,ding2025polynomialtimepreparationlowtemperaturegibbs,Kochanowski2025rapid,Rouze2025efficient,TongZhan2025,Lin2025,Rouze2026optimal,Di2025mixing,WattsSarkarCollinsEtAl2026,Bao2026typicalRelaxation}. If phase discrimination is already possible from low-energy states, one can instead use coarser filters and stop the dissipative evolution once the diagnostic has stabilized. The resource can thus be significantly reduced compared to the resolution and mixing requirements of faithful ground-state preparation.

Our work uses the dissipative constructions developed in~\cite{DingChenLin2024,ZhanDingHuhnEtAl2025}, but for a different objective and in a different regime. We employ jump operators implementable using only short-time Hamiltonian simulation, and study the resulting dissipative dynamics well before it mixes to the steady state. This pre-mixing regime has received much less analysis in the literature, and it raises two questions. First, does the dynamics already produce informative low-energy states in that regime? We answer this affirmatively for free fermions and free bosons, and argue that interacting fermionic systems generically combine aspects of both behaviors. Second, can such low-energy states decide phases in concrete settings? We show that using suitable diagnostics, phase decision problems can be solved in representative examples including the frustrated $J_1$--$J_2$ Heisenberg chain, the Kitaev honeycomb model, and the XXZ spin chain. We emphasize that the present demonstrations are proof-of-principle tests in classically accessible regimes, which provide controlled settings for validation, before moving to more challenging interacting systems. Throughout, dissipative cooling is used as an algorithmic tool for accessing the low-energy physics of the target Hamiltonian, rather than for engineering exotic mixed-state phases~\cite{Coser2019classificationof,Sang2024mixedphase,Lessa2025mixed,Ellison2025mixed,Sohal2025mixed,Guo2025mixed,sun2025anomalousmatrixproductoperator,liu2026establishingmixedstatephaseequivalence,Lang2015Phase,Kuwahara2021Gibbs,RakovszkyGopalakrishnanKeyserlingk2024,SoaresBrunelliSchiro2026}.

The rest of the paper is organized as follows. In \cref{sec:filter} we introduce dissipative cooling with jump operators that can be implemented using only short-time Hamiltonian simulation.  In \cref{sec:models}, we formulate the phase-decision problem and demonstrate the protocol on three representative examples. Next, in \cref{sec:theory} we analyze the resulting cooling dynamics for free fermionic and bosonic systems, followed by their implications for interacting fermionic systems. Finally, in \cref{sec:discussion} we discuss the scope and limitations of the approach, and whether solving phase-decision problems can lead to a practical quantum advantage. Technical proofs and additional numerical results are deferred to the appendices.

\section{Dissipative cooling with realistic filters}
\label{sec:filter}

The recently developed dissipative cooling dynamics \cite{DingChenLin2024,ZhanDingHuhnEtAl2025} are governed by the Lindblad equation \cite{Lindblad1976,GoriniKossakowskiSudarshan1976}:
\begin{equation}\label{eq:lindblad}
    \begin{aligned}
\frac{\dInt \rho}{\dInt t} &= \lind[\rho] = \underbrace{-\I[H,\rho]}_{\mc L_H[\rho]} + \underbrace{\sum_{a }  \left(K_a \rho K_a^{\dagger} - \frac{1}{2}\left\{K_a^{\dagger}K_a,\rho \right\}\right)}_{\mc D[\rho] },
    \end{aligned}
\end{equation}
where the Lindblad superoperator $\lind$, or the generator of the quantum Markov semigroup, consists of a coherent part $\mc L_H$ and a dissipative part $\mc D$. Here, $\rho(t)$ denotes the many-body density matrix, and $H = \sum_k \lambda_k \ketbra{\psi_k}$ represents the system Hamiltonian with its spectral decomposition. The jump operators $K_a$ are constructed from a set of coupling operators $A_a$ using a filter function $f(s)$ 
 as follows:
\begin{equation}
    \begin{aligned}
    K_a & = \int_{\RR} f(s) e^{\I Hs} A_a e^{-\I Hs} \dInt s\\
    & = \sum_{k,l} \hat{f}(\lambda_k - \lambda_l) \braket{\psi_k | A_a | \psi_l} \ket{\psi_k}\bra{\psi_l}.
    \end{aligned}
    \label{eq:jump_op}
\end{equation}
In this expression, $\hat{f}(\omega)$ is the Fourier transform of the time-domain filter function $f(s)$, following the convention in \cite{ChenKastoryanoBrandaoEtAl2025,DingChenLin2024}:
\begin{equation}
    \hat{f}(\omega) = \int_{\RR} f(s)e^{\I \omega s} \dInt s.
    \label{eq:ft_filter}
\end{equation}
The choice of filter function determines which transitions between energy eigenstates are allowed, depending on the values of the frequency-domain filter $\wh f(\omega)$ at the energy differences (Bohr frequencies) $\lambda_k - \lambda_l$. For ground-state preparation, {we would like to allow such a transition from $\ket{\psi_l}$ to $\ket{\psi_k}$ only when it lowers the energy, namely when $\lambda_k < \lambda_l$}. This can be ensured by requiring $\hat{f}(\omega)=0$ for all $\omega \ge  0$. We also need the first excited state to transition to the ground state, so we require $\hat{f}(-\Delta_H) > c$ for some constant $c>0$, where $\Delta_H$ is the energy gap between the ground state and the first excited state. To distinguish this from the protocol used in the rest of the manuscript, we refer to this setting as the \emph{ideal} filter, see {\cref{fig:filter_illustration_freq}}.

For the ground-state simulation task, the total Lindblad simulation time is upper bounded by the mixing time, which is the time required for the system to evolve sufficiently close to its steady state regardless of the initial state. In other words, it measures how long the system takes to ``forget'' its starting point and settle into thermal or ground-state equilibrium. The mixing time is usually highly system dependent. Under the ideal filter, rigorous upper bounds on the mixing time of the dissipative cooling algorithm have been established for several classes of non-commuting Hamiltonians, including systems satisfying certain forms of the eigenstate thermalization hypothesis (ETH)~\cite{DingChenLin2024}, free Majorana fermions, and weakly interacting spin and fermionic systems~\cite{ZhanDingHuhnEtAl2025}. Numerical evidence indicates that such protocols can be much more powerful in practice, and can prepare ground states of spin systems far beyond the weakly interacting regime, as well as ground, excited, and transition states in quantum chemistry settings~\cite{ZhanDingHuhnEtAl2025,LiZhanLin2025,LiLin2025,WattsSarkarCollinsEtAl2026}.

The end-to-end resource cost of dissipative cooling depends on several factors, including the choice of parameters in the filter function, the set of coupling operators $A_a$, the algorithm used to simulate the Lindblad dynamics, and the mixing time. The goal of this paper is not to provide concrete resource estimates, which will be left for future work. Instead, we focus on \emph{design principles} for efficient dissipative cooling algorithms and their utility in solving phase-decision problems. 

According to \cref{eq:jump_op}, 
implementing the jump operators requires Heisenberg evolution of the coupling operators $A_a$ under the system Hamiltonian, which uses Hamiltonian simulation as a subroutine. As a rule of thumb, the cost is therefore usually dominated by the total Hamiltonian simulation time.
In the ideal filter setting, constructing $K_a$ requires truncating the integral in \cref{eq:jump_op} at a time $s_{\max}$ that is inversely proportional to the energy gap $\Delta_H$ of the Hamiltonian \cite{DingChenLin2024,ZhanDingHuhnEtAl2025}. This can become large near phase boundaries, where the gap closes, forcing a large $s_{\max}$ and hence deep quantum circuits. On early fault-tolerant quantum devices, it is therefore desirable to consider a more realistic setup that resolves energies only up to some $\Delta = \Omega(1)$, with a corresponding family of filters satisfying
\begin{equation}\label{eq:finite_resolution_filter}
    \hat{f}(\omega) = \begin{cases}
         = 1, & \omega <  -\Delta,\\
        \in [0,1], &\omega \in [-\Delta,  \Delta],\\
         = 0, & \omega > \Delta,
    \end{cases}~ \Delta = \mc O(1/s_{\max}).
\end{equation}
Within the frequency window $[-\Delta, \Delta]$, the function $\hat{f}(\omega)$ is monotonically non-increasing. 
A detailed construction of these \emph{realistic filters} ({\cref{fig:filter_illustration_freq}}) with finite energy resolution can be found in Appendix~\ref{appendix:filter}. 
When $\Delta_H < \Delta$, if $\hat{f}(\omega_0) =0$ for some Bohr frequency $\omega_0\in (-\Delta,0)$, the transitions between the corresponding states become effectively dark. Conversely, when $\hat{f}(\omega_0) > 0$ for some Bohr frequency $\omega_0 \in (0, \Delta)$, the ``heating leakage'' 
would allow these energy-raising transitions, which will compete with cooling processes.

\begin{figure}
    \centering
    \begin{subfigure}{0.71\linewidth}
     \overlaption[fig:filter_illustration_freq]{\includegraphics[width=\linewidth]{{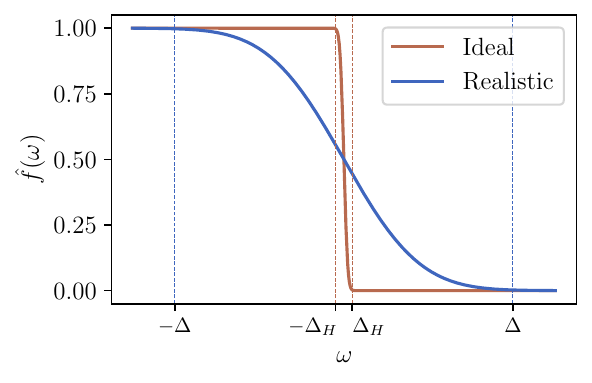}}}
    \end{subfigure}
    
      ~~~\begin{subfigure}{0.67\linewidth}
      \overlaption[fig:filter_illustration_time]{\includegraphics[width=\linewidth]{{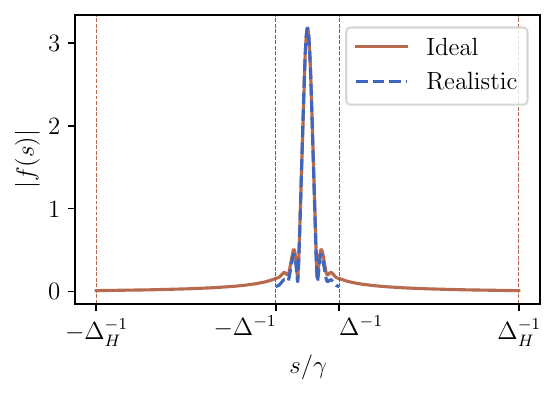}}}
    \end{subfigure}
\caption{Comparison between ideal and realistic filters in {\bf (a)} frequency domain and {\bf (b)} time domain. \label{fig:filter_illustration}}
\end{figure}

Using realistic filters shortens the Hamiltonian simulation time needed to implement the jump operators, see {\cref{fig:filter_illustration_time}}. Because of leakage, however, it cannot prepare the ground state exactly, even in the infinite-time limit of the Lindblad evolution. The first question is whether it can still drive the system efficiently into a low-energy manifold. The answer is yes, and we postpone the theoretical analysis of this question to \cref{sec:theory}. We next show that, when realistic filters drive the system into a low-energy manifold, the Lindblad simulation time needed to estimate low-energy local observables can be much shorter than the mixing time to solve phase-decision problems. Since the total Hamiltonian simulation cost is roughly the product of the time needed to implement the jump operators and the Lindblad simulation time, our protocol can use considerably fewer resources than faithful state preparation with ideal filters.

\section{Probing Quantum Phases with Dissipative Dynamics}\label{sec:models}

\subsection{Overview and problem setup}
In this section, we give concrete examples of dissipative protocols for deciding between competing ground-state phases. Before turning to these examples, we formulate the phase-decision problem studied in this work.

Consider a family of Hamiltonians $H(\kappa)$ parameterized by $\kappa \in \mathcal I \subset \mathbb{R}$. For the parameter region of interest, assume that one has already identified a finite set of candidate phases $\{\mathscr{P}_i\}$. We then choose a phase-sensitive diagnostic $\mathscr{O}: \mathscr{D} \to \varOmega$, where $\mathscr{D}$ is the set of quantum states and $\varOmega$ is the diagnostic space. In practice, $\mathscr{O}$ is obtained by measuring suitable observables and post-processing the resulting data, and need not be an order parameter in the Landau sense. For instance, in the examples below, $\mathscr{O}$ may be a fitted Luttinger parameter, or a topological invariant such as a Chern number. We assume that $\varOmega$ contains disjoint regions $\{\varOmega_i\}$ that can be used to label candidate phases $\{\mathscr P_i\}$.

For a target parameter $\kappa$, the phase-decision task is to determine which candidate phase contains the ground state of $H(\kappa)$. Given a representative initial state $\rho_0$ and dissipative dynamics tailored to $H(\kappa)$, let $\rho_\kappa(t)$ denote the evolved state. We say that the dissipative protocol succeeds at $\kappa$ if there exists a time $t_0$ such that for all $t\ge t_0$, 
\begin{equation}
    \mathscr{O}[\rho_\kappa(t)],  \mathscr{O}[\rho_{\kappa, \mathrm{GS}} ] \in \varOmega_i,
\end{equation}
where $\rho_{\kappa, \mathrm{GS}}$ is the ground state of $H(\kappa)$. {In other words, the transient dissipative state already enters the same diagnostic region as the ground state, so the ground-state phase at this $\kappa$ can be identified without waiting for full convergence to the steady state.} If the steady state of the Lindblad dynamics with ideal filters coincides with the ground state, then necessarily $t_0 \le t_{\mathrm{mix}}$, where $t_{\mathrm{mix}}$ denotes the mixing time. The main question in this section is whether one can have $t_0 \ll t_{\mathrm{mix}}$, even with realistic filters.

We examine three test cases. The first is the frustrated $J_1$--$J_2$ chain~\cite{Majumdar1969J1J2,Majumdar1969J1J22} in \cref{sec:J1J2}, which exhibits BKT-type transitions with extended critical phases~\cite{Haldane1982J1J2BKT}.
We show that even with realistic filters, short-time dissipative dynamics can still resolve these phases.
In particular, we probe the phase diagram through spin-correlation functions and extract the Luttinger parameter $\mathcal{K}$~\cite{Giamarchi2003,SongRachelLeHur2010,RachelLaflorencieSongEtAl2012} to carry out the phase decision. We also benchmark the protocol against quench dynamics in Appendix~\ref{appendix:J1J2_supplementary}.

The second example is the two-dimensional Kitaev honeycomb model~\cite{Kitaev2006}, an exactly solvable system exhibiting a topologically ordered phase.
Dissipation has long been regarded as a powerful resource for preparing topological states~\cite{DiehlRicoBaranovEtAl2011,BardynBaranovRicoEtAl2012,Bardyn2013TopoDiss,BudichZollerDiehl2015,LiuBergholtzBudich2021,PokartKoningDiehlBudich2026}, and in particular for preparing the Chern insulators \cite{Goldstein2019}. 
In \cref{sec:Kitaev}, we demonstrate that by applying dissipative dynamics to the original gapless Hamiltonian without opening the spectral gap, one can extract a mixed-state Chern diagnostic with the expected finite-size phase label 
using finite-resolution filters and short evolution time. This approach leverages topological diagnostics generalized to mixed states~\cite{Bardyn2013TopoDiss,BudichZollerDiehl2015,BudichDiehl2015,Evered2025Honeycomb}.

Finally, in \cref{sec:XXZ} we consider the anisotropic XXZ spin chain~\cite{Franchini2017}, {which also exhibits a BKT-type phase transition.} {Its nearest-neighbor structure allows efficient tensor network simulations of the Lindblad dynamics.} For a relatively large system ($L=40$), {we find from tensor network simulations that a coarse filter resolution and a short evolution time suffice to recover the  qualitatively correct  correlation-function behavior associated with the respective phases.} Adiabatic evolution~\cite{farhi2000quantumcomputationadiabaticevolution,Albash2018Adiabatic} provides an alternative route to preparing ground states or low-energy states, but it is often bottlenecked by the closing of the spectral gap at critical points. This is especially restrictive for BKT transitions, where the gap opens exponentially slowly as the coupling deviates from the critical value. 
We benchmark our dissipative protocol against adiabatic evolution for the XXZ spin chain.
Within our maximum simulated evolution time, we find that {the adiabatic dynamics fails to resolve the power-law decay of the correlation functions characteristic of the gapless phase.}

\subsection{Frustrated $J_1$--$J_2$ Heisenberg chain}\label{sec:J1J2}

We first consider the antiferromagnetic $J_1$--$J_2$ Heisenberg chain. This model incorporates frustrated next-nearest-neighbor interactions and is described by the Hamiltonian:
\begin{equation}
    H = \sum_{i=1}^{L} \left( J_1 \mathbf{S}_i \cdot \mathbf{S}_{i+1} + J_2 \mathbf{S}_i \cdot \mathbf{S}_{i+2} \right),\quad J_2/J_1\ge 0.
\end{equation}
Here $\bv S_i = (S_i^x, S_i^y, S_i^z)$ are spin-$\frac12$ operators at site $i$. We adopt the periodic boundary condition (PBC) and the subscript $i$ is defined modulo $L$.  
When the frustration ratio $\kappa := J_2/J_1\le \kappa_{\rm crit} \approx 0.241$ \cite{OkamotoNomura1992,Giamarchi2003}, the system is in the critical \emph{Tomonaga--Luttinger liquid} (TLL) phase, while for $\kappa >\kappa_{\rm crit} $ the system undergoes a BKT phase transition into a gapped dimerized phase called \emph{valence bond solid} (VBS){, see \cref{fig:phase_j1j2}}. 
For a given $\kappa$, we would like to decide whether the system is in the TLL or VBS phase based on measured observables. In principle, one could distinguish these phases by examining the decay properties of the spin correlation function $\left( \langle S_i^z S_j^z \rangle - \langle S_i^z \rangle \langle S_j^z \rangle \right)$, which decays algebraically in the TLL phase and exponentially in the gapped VBS phase. Near a BKT transition, however, the correlation length on the VBS side becomes exponentially large in $(\kappa-\kappa_{\rm crit})^{-1/2}$, and higher corrections on the TLL side can complicate finite-size scaling analysis~\cite{BursillGehringFarnellEtAl1995}.

\begin{figure} 
    \centering
     \begin{subfigure}{\linewidth}
    \overlaption[fig:phase_j1j2]{\includegraphics[width=0.92\linewidth]{{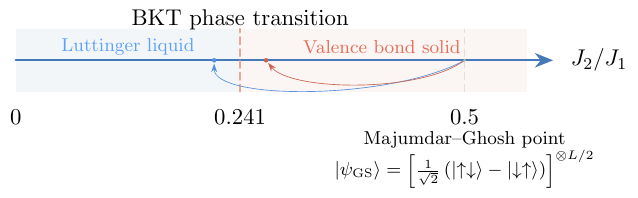}}}
    \end{subfigure}
      \begin{subfigure}{\linewidth}
  \overlaption[fig:J1J2_fitting]{\includegraphics[width=0.96\linewidth]{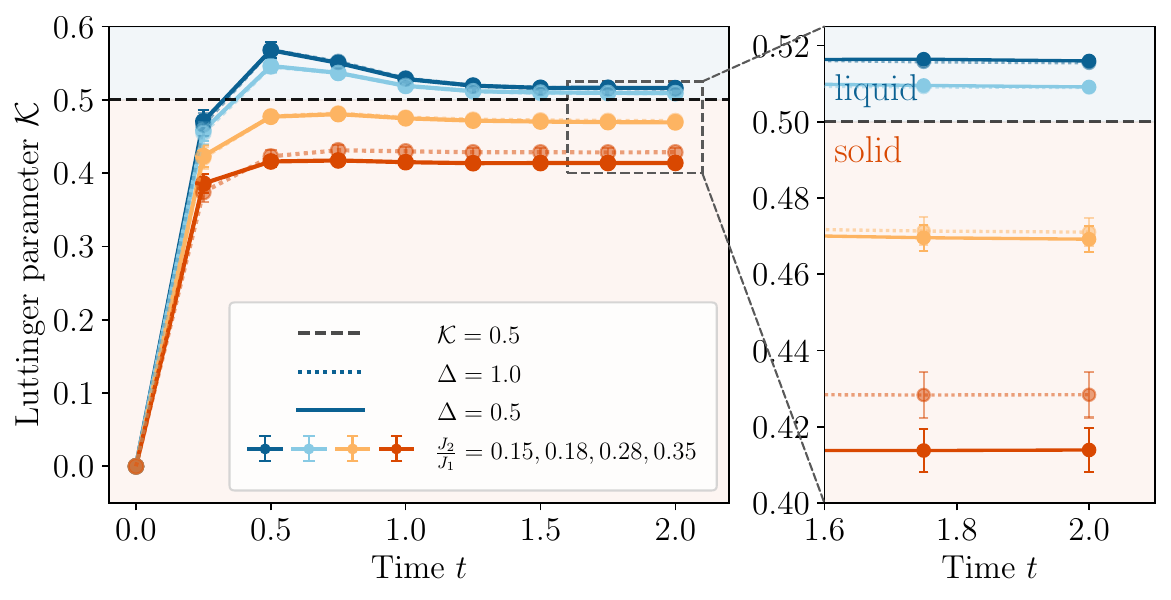}}
      \end{subfigure}
\caption{{\textbf{(a)} The phase diagram of the $J_1$--$J_2$ chain. The critical point is at $J_2/J_1\approx 0.241.$ \textbf{(b)}} Fits of the subsystem bipartite $S^z$-fluctuation   $\mathcal{F}(\ell)$ for the $J_1$--$J_2$ model for different frustration ratios $J_2/J_1$. The fitting is performed along the dissipative dynamics. See Appendix~\ref{appendix:analysis} for the fitting model. The error bars represent the standard deviation computed from the covariance matrix. Here the first two values of $\mc F$ are excluded for the fitting purpose.  }
\end{figure}

This issue can be overcome by leveraging the fact that, within the TLL regime, the effective low-energy theory is described by a free bosonic conformal field theory (CFT) with a continuously varying Luttinger parameter $\mc K$~\cite{Giamarchi2003,SongRachelLeHur2010,RachelLaflorencieSongEtAl2012}. This parameter dictates the scaling behavior of the entanglement entropy even in interacting systems \cite{SongRachelFlindtEtAl2012}, and can be extracted numerically from bipartite $S^z$ fluctuations, which are defined as the variance of the total $S^z$ in a subsystem $\mathcal{A}_\ell$ of size $\ell$:\begin{equation}
    \mathcal{F}(\ell) = \sum_{i,j \in \mathcal{A}_\ell} \left( \langle S_i^z S_j^z \rangle - \langle S_i^z \rangle \langle S_j^z \rangle \right), \quad \mathcal{A}_\ell = \{1, 2, \dots, \ell\}.
\end{equation}
Within the TLL phase, $\mc K$ appears as the coefficient of the logarithmic term in $\ell$ (see Appendix~\ref{appendix:analysis} for the detailed fitting procedure). By analyzing $\mathcal{F}(\ell)$, the collective behavior of the correlation functions is effectively amplified, allowing the phase to be identified more reliably near the critical point. The Luttinger parameter $\mc K > 1/2$ in the TLL phase, and $\mc K = 1/2$ at the BKT transition. In the gapped VBS phase, $\mathcal{F}(\ell)$ saturates rather than growing logarithmically in the thermodynamic limit, so fitting the TLL form there should be viewed only as an effective finite-size diagnostic (and should satisfy $\mc K < 1/2$) rather than interpreted literally as a Luttinger parameter.

We set the coupling operators as the valence bond operators $\{\mathbf{S}_i \cdot \mathbf{S}_{i+1}\}_{i=1}^{L}$ and consider filter energy resolutions of $\Delta = 1$ and $\Delta = 0.5$. We perform exact diagonalization simulations of Lindbladian dynamics with system size $L= 14$. The initial state is chosen to be the ground state of the Majumdar--Ghosh point $J_2/J_1=1/2$ in the VBS phase, which is one of the exact dimer-product ground states, $\ket{\psi_{\rm MG}} = \left[\frac{1}{\sqrt{2}}\left(\ket{\uparrow\downarrow} - \ket{\downarrow\uparrow} \right) \right]^{\otimes L/2}$.
As illustrated in \cref{fig:J1J2_fitting}, we find that the combination of a relatively \emph{coarse} filter resolution and a \emph{short} dissipative evolution is sufficient to accurately distinguish the different phases, even within the critical region.
On the TLL side, the extracted Luttinger parameter consistently satisfies $\mc K > 0.5$ as expected. For target couplings in the VBS phase, the same fit yields an effective logarithmic coefficient below $0.5$. While the quality of the fit is high within the TLL regime, it degrades in the VBS phase due to the fact that the logarithmic scaling form is a property only of the critical TLL phase. In the thermodynamic limit, the bipartite fluctuations ${\mc F}(\ell)$ shift from logarithmic divergence in the TLL phase to a constant saturation in the gapped VBS phase.

We also benchmark our protocol against quench dynamics in Appendix~\ref{appendix:J1J2_supplementary}. We find that the quench approach fails to detect signatures of the quantum phase transition within the simulated time scales, {highlighting the comparative advantage of our dissipative method near a BKT transition.}

 \begin{figure*}
    \begin{subfigure}{0.345\linewidth}
    \overlaption[fig:honeycomb]{\includegraphics[width=0.99\linewidth]{{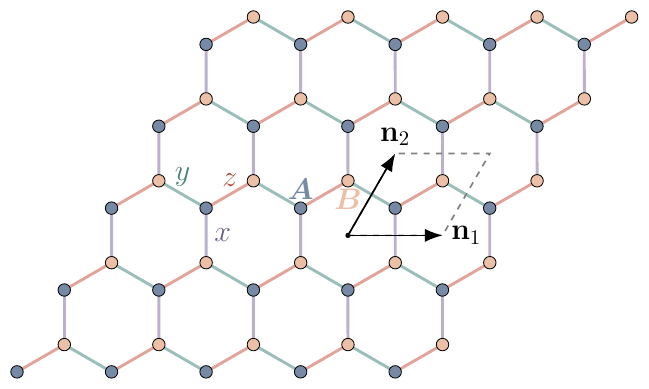}}}
    \end{subfigure}\hfill
    \begin{subfigure} {0.305\linewidth}
        \overlaption[fig:kitaev_phase]{\includegraphics[width=0.99\linewidth]{{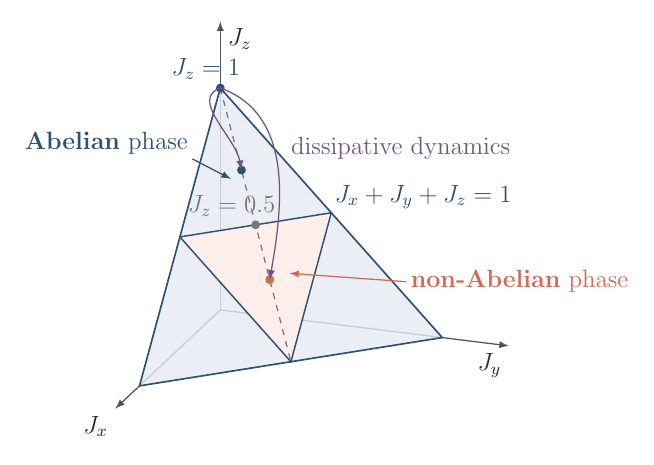}}}  
    \end{subfigure}
    \begin{subfigure}{0.34\linewidth}
        \overlaption[fig:kitaev_dissipative]{\includegraphics[width=0.99\linewidth]{{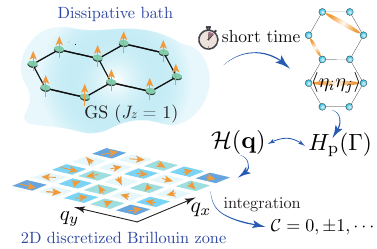}}}
    \end{subfigure}
\caption{{\bf (a)} The honeycomb lattice. The three types of bonds are labeled by $x$, $y$, and $z$. $\mathbf n_1$ and $\mathbf n_2$ are two primitive vectors. The two sublattices are labeled by $A$ and $B$. {\bf (b)} The phase diagram of the Kitaev honeycomb model on the plane $J_x+J_y+J_z = 1$. The blue region is the Abelian phase and the orange region is the non-Abelian phase. The dashed line is the $J_x = J_y =\frac{1-J_z}{2}$ line. The critical point is at $J_z = 0.5$. {\bf(c)} An illustration of the dissipative protocol for probing the topological phase transition. The system is initialized in the Ising state and then evolved under the short-time dissipative dynamics. The discretized Berry curvature is computed from the Bloch Hamiltonian $\mc H(\bv q)$ and then integrated over the first Brillouin zone to extract the Chern number $\mc C$. }
\end{figure*}

\subsection{2D Kitaev honeycomb model}

\label{sec:Kitaev}
The dissipative protocol naturally extends to two or higher dimensional systems. As an illustrative example, let us consider the Kitaev honeycomb model~\cite{Kitaev2006}, which is described by a spin-$\frac12$ Hamiltonian with anisotropic exchange on a honeycomb lattice:
\begin{equation}
    \label{eq:kitaev_spin}
    H = -4\sum_{\alpha \in \{x,y,z\}} J_\alpha \sum_{\langle i,j \rangle_\alpha} S_i^\alpha S_j^\alpha,
 \end{equation}
 where $\langle i,j \rangle_\alpha$ denotes the nearest neighbor sites connected by an $\alpha$-type bond. $\alpha$ can be $x$, $y$, or $z$ depending on the direction of the bond, see \cref{fig:honeycomb}. The model exhibits a rich phase diagram featuring Abelian and non-Abelian spin liquid phases depending on the relative strengths of the coupling constants $J_x$, $J_y$, and $J_z$~\cite{Kitaev2006}, as shown in \cref{fig:kitaev_phase}. In recent years, theoretical and experimental studies of this model have focused extensively on topological state preparation on quantum processors \cite{Evered2025Honeycomb,WillCochranRosenbergEtAl2025}, as well as topological signatures in quench dynamics \cite{WangZhangChenEtAl2017,TarnowskiUnalFlaschnerEtAl2019}.

The Kitaev honeycomb model can be mapped to Majorana fermions coupled to a static $\mathbb{Z}_2$ gauge field. When restricted to the vortex-free sector and a fixed gauge, the Hamiltonian reduces to a quadratic Majorana fermion model
\begin{equation}\label{eq:kitaev_majorana_0}
   H =  \I  \sum_{\bv s, \mu, \bv t, \nu} J_{\alpha(\bv s, \mu; \bv t, \nu)}  \eta_{\bv s\mu} \eta_{\bv t\nu}  
\end{equation}
where $\eta_{\bv{s}\mu}$ are Majorana fermions satisfying the anticommutation relations $\{\eta_{\bv s\mu},\eta_{\bv t\nu}\} = \delta_{(\bv s,\mu),(\bv t,\nu)}$, and $\mu=A, B$ are the two sites contained in the unit cell of the honeycomb lattice, as illustrated in \cref{fig:honeycomb}. 
Applying a Fourier transform followed by a Bogoliubov transformation diagonalizes the system into a free-fermion model in momentum space, in which the independent quasiparticle degrees of freedom can be represented by the $B$-sublattice fermionic operators
\begin{equation}
      H = \sum_{\bv q}\abs{g(\bv q)}  (1-2c_{\bv qB}^\dag c_{\bv qB}),
\end{equation}
with the dispersion relation given by
\begin{equation}
    g(\bv q) := J_x \cos q_x + J_y \cos q_y + J_z + \I (J_x \sin q_x+ J_y \sin q_y),
\end{equation}
and $c_{\bv qB}$, $c_{\bv qB}^\dagger $ represent fermionic annihilation and creation operators of the Bogoliubov quasiparticles. For simplicity, we take the side length  $L$ of the supercell to be odd and $q_x,q_y = \frac{2\pi n}{L}$ with $n=-\frac{L-1}{2},\cdots,\frac{L-1}{2}$. The ground state is given by
\begin{equation}
    \ket{\psi_\text{GS}} = \prod_{\bv q} c_{\bv qB}^\dag \ket{0}.
\end{equation}
Here $\ket{0}$ is the vacuum state of $c_{\bv qB}$ operators. The energy gap is given by $\min_{\bv q} 2\abs{g(\bv q)}$. We refer to Appendix \ref{app:kitaev_diagonalization} for the details of the model and its diagonalization.

{Here we are interested in the line $J_x= J_y = ({1-J_z}) / {2}$ in the parameter space
 (the dashed line in the phase diagram \cref{fig:kitaev_phase}).} 
At the critical point $J_z = 0.5$, the model changes from the gapped Abelian phase for $J_z>0.5$ to the gapless phase for $J_z<0.5$ with time-reversal symmetry preserved.
Along this line, the gapped phase has the trivial parent-Hamiltonian Chern number $\mathcal{C} = 0$. For $J_z<0.5$, the thermodynamic system is gapless; after adding a small time-reversal-breaking perturbation that opens a gap, one obtains a chiral non-Abelian spin liquid with $\mathcal{C} = \pm 1$ \cite{Kitaev2006}. {In the present work we do not explicitly add such a perturbation. Instead, we use the Chern number of the mixed-state parent Hamiltonian constructed from the covariance matrix as a finite-size diagnostic that distinguishes the two regimes in the dissipative dynamics.}

\begin{figure}
\centering 
    \begin{subfigure}{0.495\linewidth}
    \overlaption[fig:dissipative_2D_topological_chern]{\includegraphics[width=0.95\linewidth]{{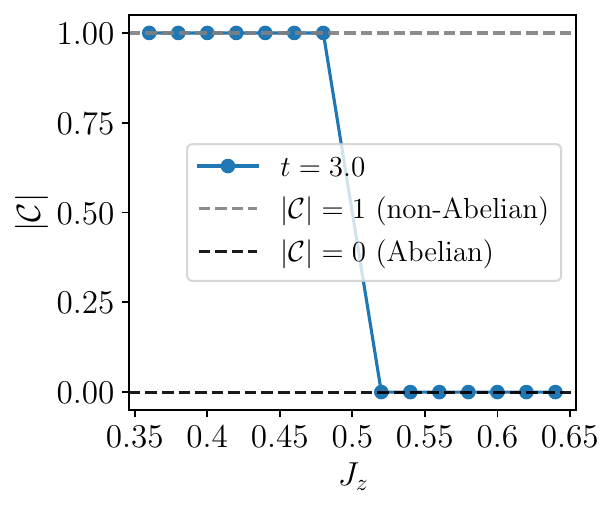}}}
    \end{subfigure}
\hfill
    \begin{subfigure}{0.485\linewidth}
 \overlaption[fig:dissipative_2D_topological_purity]{\includegraphics[width=0.95\linewidth]{{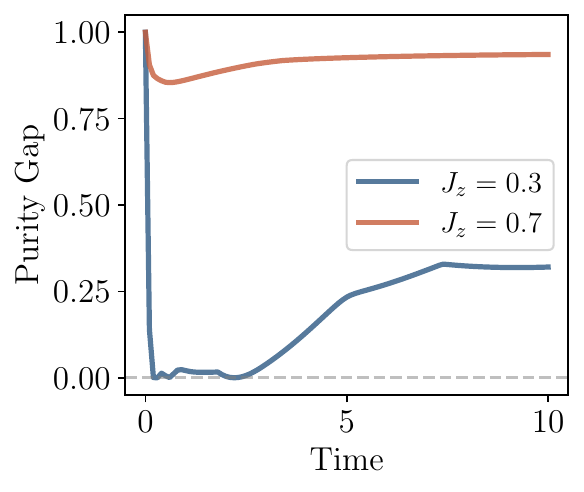}}}
    \end{subfigure}
    \begin{subfigure}{0.325\linewidth}
   \overlaption[fig:dissipative_2D_topological_band_t0]{\includegraphics[width=\linewidth]{{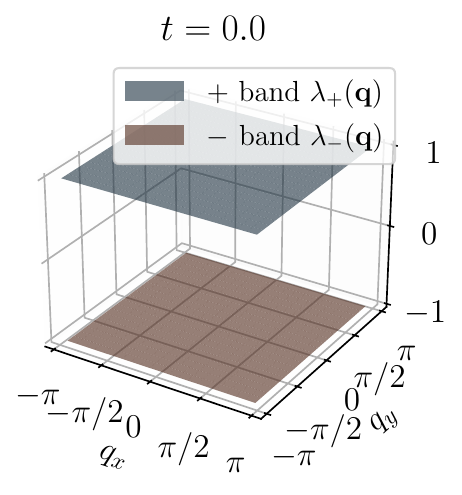}}}
    \end{subfigure}
        \begin{subfigure}{0.325\linewidth}
   \overlaption[fig:dissipative_2D_topological_band_t3]{\includegraphics[width=\linewidth]{{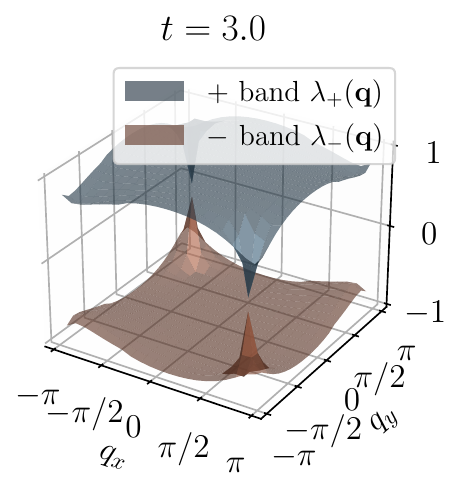}}}
    \end{subfigure}    
    \begin{subfigure}{0.325\linewidth}
   \overlaption[fig:dissipative_2D_topological_band_t8]{\includegraphics[width=\linewidth]{{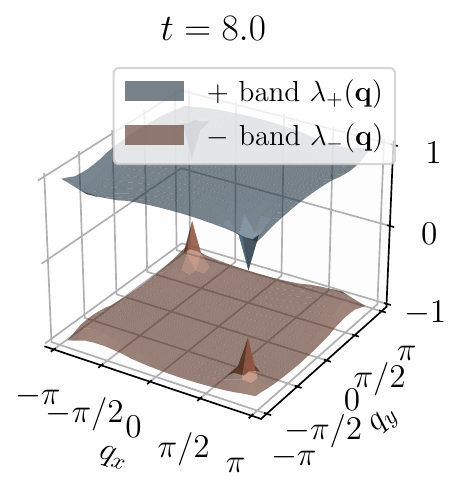}}}
    \end{subfigure}
    \caption{Dissipative dynamics for probing the topological phase transition in the Kitaev honeycomb model. The system size is $L=21$ and the filter resolution is $\Delta = 0.5$. {\bf (a)} The Chern number computed from the covariance matrix at the short simulation time $t=3$. {\bf (b)} The evolution of the purity gap along the dissipative dynamics. We start from $J_z=1$ in the Abelian phase, and for targets with $J_z<0.5$ the purity gap closes and then reopens during the dynamics. {\bf (c-e)} The band structure of the Bloch Hamiltonian $\mc H(\bv q)$ constructed from the dissipative dynamics at $t = 0, 3$ and $8$ for the target parameter $J_z = 0.4$.}
    \label{fig:dissipative_2D_topological}
\end{figure}

We choose local Majorana operators $\{\sqrt 2 \eta_{\bv s A}, \sqrt 2\eta_{\bv s B}\}_{\bv s}$ as Lindblad coupling operators and set the energy resolution to be $\Delta = 0.5$ in the realistic filter. {This choice of coupling operators is commonly referred to as \emph{bulk dissipation} and will be revisited later in \cref{sec:XXZ,sec:free_fermion}.}  
For each $J_z$, the initial state is the all-spins-up configuration, i.e., the ground state in the Ising limit $J_x = J_y = 0$ and $J_z = 1$. 
The two phases can be diagnosed by analyzing the fermionic covariance matrix $\Gamma$ of the dissipatively evolved state, which encodes the relevant topological information~\cite{DiehlRicoBaranovEtAl2011,Bardyn2013TopoDiss,BudichDiehl2015,BudichZollerDiehl2015,BardynWawerAltlandEtAl2018}. {Specifically}, at each time during the evolution, we construct the corresponding parent Hamiltonian from $\Gamma$ as
\begin{equation}
H_{\rm p } = {\I} \sum_{\bv s, \mu, \bv t, \nu} \Gamma_{\bv s\mu, \bv t\nu} \eta_{\bv s\mu} \eta_{\bv t\nu},\quad
\Gamma_{\bv s\mu, \bv t\nu} = \frac{\I}{2} \langle [\eta_{\bv s\mu}, \eta_{\bv t\nu}] \rangle.
\end{equation}
Applying the Fourier transform to the parent Hamiltonian, we obtain the $2\times 2$ Bloch Hamiltonian $\mc H(\bv q)$ for each momentum vector $\bv q$ in the discretized first Brillouin zone. For mixed states with an open purity gap  $\Delta_{P}(\Gamma) >0 $ which is the smallest eigenvalue of $(2\I \Gamma)^2$, one first flattens this parent Hamiltonian \cite{Bardyn2013TopoDiss}, and the Chern number $\mc C$ is then computed from the resulting two-band model using the specialized Fukui--Hatsugai--Suzuki (FHS) algorithm~\cite{FukuiHatsugaiSuzuki2005}; see Appendix~\ref{appendix:topological} for details.
The schematic illustration of the protocol is shown in \cref{fig:kitaev_dissipative}.

Readers may wonder whether our results contradict the no-go theorem for rapid preparation of topologically ordered states by quasi-local dissipative dynamics~\cite{Konig2014}. In fact, 
different operational notions of topological order lead to different preparation obstructions, for example those based on finite-depth local unitary circuits \cite{ChenGuWen2011}, finite-depth circuits with measurements and feedforward \cite{TantivasadakarnVishwanathVerresen2023}, or fast dissipative evolution \cite{Coser2019classificationof}.  Chern insulators are invertible free-fermionic phases, rather than intrinsically topologically ordered phases considered in \cite{Konig2014}; see Appendix~\ref{app:CI}. While the Kitaev honeycomb model can host toric-code-like anyon excitations~\cite{Kitaev2006}, our fermionic coupling operators are nonlocal in the spin representation of \cref{eq:kitaev_spin}, so the protocol is outside the quasi-local spin setting addressed there. Moreover, our goal is not to prepare the exact pure ground state, but to extract topological information from the mixed low-energy manifold encoded in $\Gamma$. In the gapless regime $J_z<0.5$, the minimum mode energy for a finite system of linear size $L$ scales as $\varO(1/L)$ rather than staying open by a constant. Even so, as long as the finite-size parent Hamiltonian has an open purity gap at the evaluation time, the parent-Hamiltonian construction together with the FHS algorithm yields a well-defined Chern diagnostic, and we will show below that the protocol recovers the expected phase label even with a realistic filter of constant width.

We perform numerical tests on a honeycomb lattice with system size $21 \times 21$, and the results are presented in \cref{fig:dissipative_2D_topological}. As shown in \cref{fig:dissipative_2D_topological_chern}, with a short total dissipative evolution time $t = 3$, the mixed-state Chern number extracted from the parent Hamiltonian already agrees with the expected phase label of the finite-size system. It is also informative to compute the purity gap. As discussed in \cite{Bardyn2013TopoDiss}, a closing of the purity gap signals that the Gaussian state passes through a point where its topological class can change, and thus it serves as an indicator of a dynamical topological transition during the dissipative evolution. For $J_z < 0.5$, the purity gap closes at an early stage and then reopens, which is consistent with the stability of the probed Chern number after the system enters the relevant low-energy manifold (see \cref{fig:dissipative_2D_topological_purity}).

To better convey the transition process, we also plot the band structure of the Bloch Hamiltonian $\mc H (\bv q)$ constructed from the parent Hamiltonian $H_{\rm p}$ (\cref{fig:dissipative_2D_topological_band_t0,fig:dissipative_2D_topological_band_t3,fig:dissipative_2D_topological_band_t8}). We see that since the short-time dynamics data from $t=3$ (\cref{fig:dissipative_2D_topological_band_t3}) corresponds to a mixed state that is still far from the true ground state, the unflattened parent-Hamiltonian spectrum is not yet close to $\pm1$. Nevertheless, the FHS-based Chern number has already reached the expected phase label, because the flattened eigenvectors are already in the correct topological class once the purity gap has reopened. At $t=8$ (\cref{fig:dissipative_2D_topological_band_t8}), the state is much closer to the true ground state, and the parent-Hamiltonian spectrum is closer to the flat-band limit. This suggests that even with a realistic filter and short-time dynamics, the dissipative evolution produces only coarse data while still capturing the qualitative topological features.
Such mixed-state diagnostics can in fact be adapted to experimental settings where only noisy local observables are accessible \cite{Evered2025Honeycomb}: one may employ Hamiltonian learning schemes to construct an effective parent Hamiltonian and compute the corresponding Chern number from it \cite{Qi2019determininglocal,Huang2023Learning,Olsacher2025}.

The performance of this protocol can be justified theoretically within the fixed-gauge quadratic Majorana description used in our numerics. In this setting, the model has a rapid-mixing property, which means that
the mixing time of the dissipative dynamics scales as $\varO(\log L)$, as proved in Ref.~\cite{ZhanDingHuhnEtAl2025} for ideal filter settings and generalized to realistic filter settings later in \cref{thm:free_fermion} (see also \cref{thm:free_fermion_convergence} in Appendix~\ref{appendix:theory}). Further technical details are provided in the appendices. We review topological diagnostics in Appendix~\ref{appendix:topological}, illustrate a simple mechanism showing how quasi-local dissipation can evade unitary obstructions in Appendix~\ref{app:CI}, and analyze steady-state properties with realistic filters in Appendix~\ref{app:steady_state_topological}. Additional supporting numerical results are provided in Appendix~\ref{appendix:kitaev_supplementary}, 
 and the explicit form of the quasi-free Lindbladian underpinning our numerics is given in Appendix~\ref{appendix:kitaev_quasifree}.

\subsection{XXZ model}\label{sec:XXZ}

\begin{figure*}
 \includegraphics[width=0.88\textwidth]{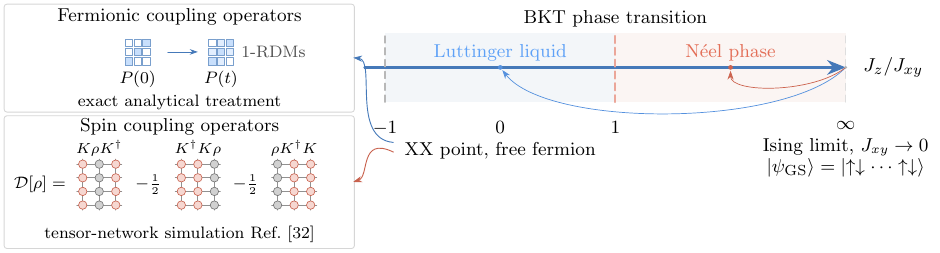}
    \caption{Part of the phase diagram of the one-dimensional XXZ model. The blue region denotes the critical, gapless Luttinger-liquid phase and the orange region denotes the gapped antiferromagnetic N\'eel phase. The transition occurs at $J_z/J_{xy}=1$. At the XX point ($J_z=0$), the model can be mapped to a free-fermion system, which allows us to evaluate two 
    different types of coupling operators: the  spin coupling operators used in \cref{sec:J1J2} and the fermionic coupling operators used in \cref{sec:Kitaev}.  The fermionic coupling operators can be treated analytically using quasi-free methods, whereas the spin coupling operators are simulated numerically using tensor-network methods introduced in \cite{ZhanDingHuhnEtAl2025}.}\label{fig:xxz_phasediagram}
\end{figure*}

\begin{figure} 
    \centering
    \begin{subfigure}{1.00\linewidth}
    \overlaption[fig:dissipative_XX_snapshots]{\includegraphics[width=\linewidth]{{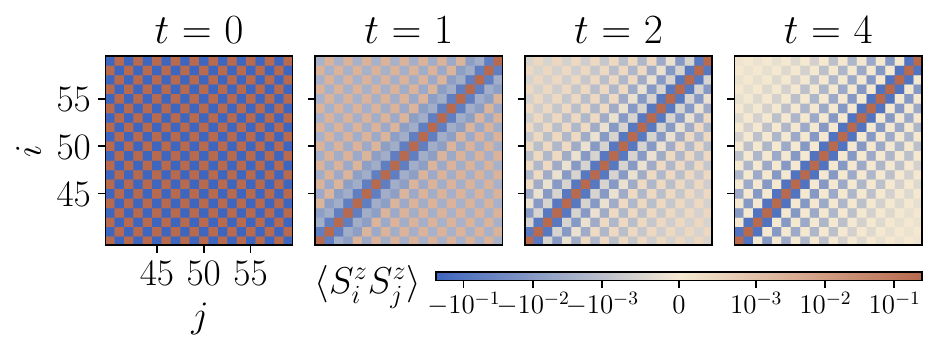}}}
    \end{subfigure}
    \hfill
    \begin{subfigure}{0.48\linewidth}
    \overlaption[fig:dissipative_XX_fit]{\includegraphics[width=\linewidth]{{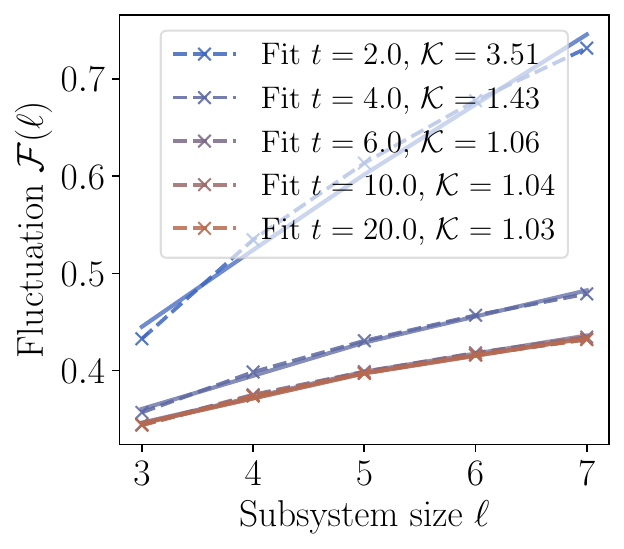}}}
    \end{subfigure}
    \hfill
    \begin{subfigure}{0.48\linewidth}
    \overlaption[fig:dissipative_XX_trace]{\includegraphics[width=\linewidth]{{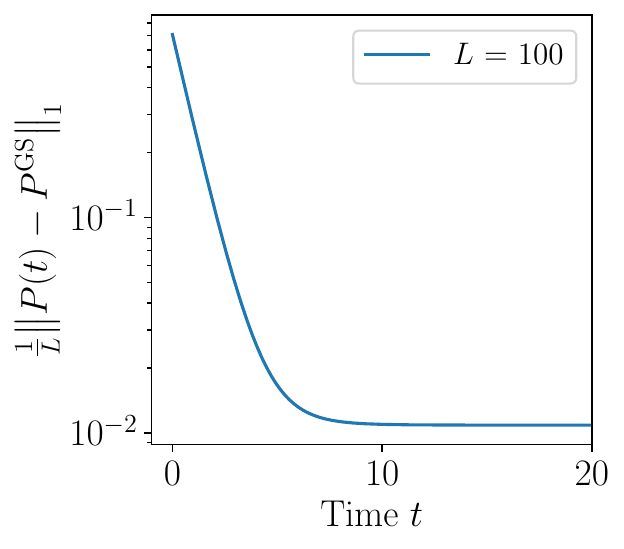}}}
    \end{subfigure}
    \caption{The dissipative cooling dynamics of the XX model with $(J_{xy}, J_z)= (1, 0)$  in \cref{eq:XXZ}, using fermionic coupling operators. The results are obtained by tracking the one-particle reduced density matrix (1-RDM) in quasi-free dynamics with $L=100$. {\bf (a)}  Snapshots of the spin correlation function $\langle S^z_iS^z_j\rangle$ at different times along the dissipative cooling dynamics. We show the data for $i,j\in \{40,\cdots,59\}$. {\bf (b)} Fitted Luttinger parameter $\mc K$ from the subsystem bipartite $S^z$-fluctuation  $\mathcal{F}(\ell)$ along the dissipative cooling dynamics. {\bf (c)} The discrepancy in the trace norm between the 1-RDMs of the evolved state and the ground state, rescaled by the system size.}
    \label{fig:dissipative_XX}
\end{figure}

\begin{figure*}
    \centering
      \begin{subfigure}{0.31\linewidth}
      \overlaption[fig:XXZ_spin_coupling_fidelity]{\includegraphics[width=\linewidth]{{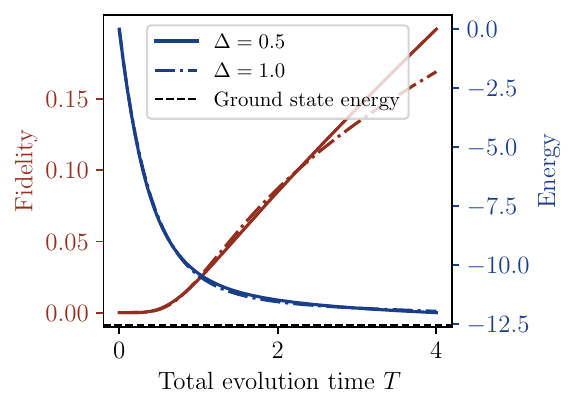}}}
    \end{subfigure}
      \begin{subfigure}{0.415\linewidth}
    \overlaption[fig:XXZ_spin_coupling_SzSz]{\includegraphics[width=\linewidth]{{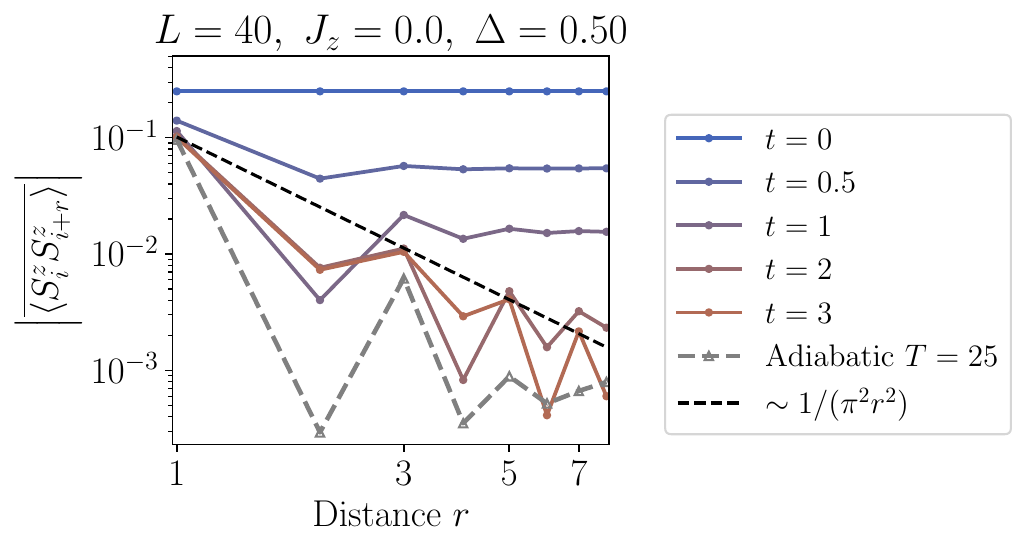}}}
        \end{subfigure}
    \hfill \begin{subfigure}{0.26\linewidth}
    \overlaption[fig:XXZ_spin_coupling_fit]{\includegraphics[width=\linewidth]{{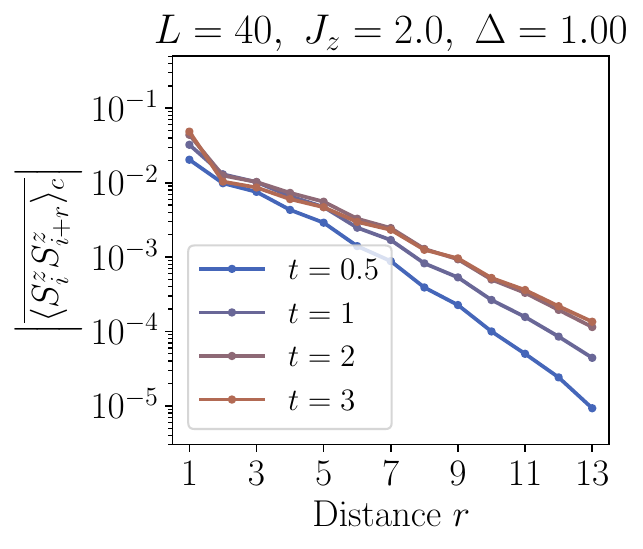}}}
    \end{subfigure}
    \caption{
    The dissipative cooling dynamics of {\bf (a, b)} the XX model  $(J_{xy}, J_z) = (1,0)$ and {\bf (c)} an antiferromagnetic gapped point $(J_{xy}, J_z) = (1,2)$ using spin coupling operators. The system size $L=40$. The results are obtained by tensor network simulations. {\bf (a)} Fidelity between the evolved state and the ground state, and energy expectation value along the dissipative cooling dynamics. {\bf (b)} Average spin correlation function $\langle S^z_iS^z_{i+r}\rangle$ at different times along the dissipative cooling dynamics for $J_z=0$. The gray dashed line stands for adiabatic-evolution results with total evolution time $T=25$. The exact result in the thermodynamic limit is shown with black dashed line for odd $r$, and is 0 for even $r$. 
    {\bf (c)} Average connected spin correlation function $\langle S^z_iS^z_{i+r}\rangle_c $ at different times along the dissipative cooling dynamics for $J_z = 2$. Note that only the $y$-axis is on a logarithmic  scale in {\bf (c)} in contrast to the log--log plot in {\bf (b)}. In this subplot we use the sharp-cutoff filter, as defined in Appendix~\ref{appendix:filter}.
    }
    \label{fig:XXZ_spin_coupling}
\end{figure*}

Our last example is the XXZ Hamiltonian  
\begin{equation}\label{eq:XXZ}
    H = \sum_{i }  \left[ J_{xy} \left(S_i^x S_{i+1}^x + S_i^y S_{i+1}^y \right)+ J_zS_i^z S_{i+1}^z \right],
\end{equation}
where $J_{xy} > 0$. The ground state of the XXZ model in one dimension is solvable via the Bethe ansatz~\cite{Yang1966XXZBethe} and possesses a rich phase diagram~\cite{Langari1998,Kitanine2002,Rakov2016}. In this study, we focus on the transition between two distinct regimes. When $J_z / J_{xy} > 1$, the system resides in a gapped antiferromagnetic N\'eel phase. In the Ising limit $J_{xy} \to 0$, one of the two degenerate ground states reduces to the product state $|\psi_{\text{ini}}\rangle = \ket{\uparrow\downarrow\cdots\uparrow\downarrow}$. It is easy to prepare experimentally and serves as our initial state. At $J_{z}/ J_{xy} = 1$, the system undergoes a BKT liquid--solid phase transition.
When $-1 < J_z / J_{xy} < 1$, the system lies in the critical, gapless XX phase, which is a gapless Luttinger liquid; {see \cref{fig:xxz_phasediagram}.}

We choose the $J_z=0$ point as our target Hamiltonian, since it allows for an exact mapping to a free-fermion system. This specific point provides a unique platform to evaluate two distinct approaches: the spin coupling operators and the fermionic coupling operators, which we have used in \cref{sec:J1J2} and \cref{sec:Kitaev}, respectively. 
Although the fermionic operators are nonlocal {in the spin representation because of} the Jordan--Wigner strings, they allow for exact analytical treatment, whose performance will later be rigorously proved in \cref{sec:free_fermion}.  
The spin coupling operators are strictly local in the spin basis, making them suitable for implementation without non-local operations.

The low-energy physics of the XX phase is also effectively captured by the TLL
theory. Within this regime, the Luttinger parameter ${\mc K}$ varies continuously as a function of the ratio $J_z / J_{xy}$, which can be analytically obtained from the Bethe ansatz solution as \cite{Giamarchi2003}
\begin{equation}\label{eq:bethe_solution}
   \mc K = \frac12 \left[1-\frac{\arccos \left(J_z / J_{xy}\right)}{\pi}\right]^{-1}.
\end{equation}

Again, the Luttinger parameter can be fitted from the subsystem bipartite $S^z$-fluctuation  $\mathcal{F}(\ell)$. In the critical regime near the phase transition point $J_z/ J_{xy}=1$, it is in principle still possible to determine the Luttinger parameter from fitting $\mc F(\ell)$ with subleading correction terms in order to identify the phases more precisely. However, as discussed in Refs.~\cite{SongRachelFlindtEtAl2012,SongRachelLeHur2010,RachelLaflorencieSongEtAl2012}, {the treatment of the subleading terms becomes considerably more subtle under OBCs in this regime.} We therefore focus here on cases deep inside the phases, where the phase structure can be identified more directly from the scaling behavior of the correlation functions in the tensor-network based simulations where OBCs are employed.

We first analyze the performance of dissipative dynamics for both choices of coupling operators, and then benchmark the performance against adiabatic evolution.

\subsubsection{Fermionic coupling operators}

Setting $J_z=0$ in \cref{eq:XXZ} gives the XX chain, a special point deep in the gapless XX phase, which can be mapped to a free-fermion Hamiltonian through the Jordan--Wigner transformation and Fourier transform:
\begin{equation}
    H = \sum_{k=1}^L \varepsilon_k \wt c_k^\dag \wt c_k,\  \varepsilon_k =
     \begin{cases}
    <0, & k \le L/2;\\
    >0, & k > L/2.
    \end{cases}
    \label{eq:diag_xxhamil}
\end{equation}
When imposing the periodic boundary condition (PBC) for the spin system \cref{eq:XXZ}, we have $\varepsilon_k = {- J_{xy}\cos\left( {2k \pi } / {L} \right)}$. 
We set the fermionic coupling operators to $\{c_i^\dagger, c_i\}_{i=1}^{L}$. {The performance of bulk dissipation with an ideal filter has been studied extensively in prior work~\cite{ZhanDingHuhnEtAl2025,LiZhanLin2025,LiLin2025}.} Later in \cref{sec:free_fermion} and Appendix \ref{appendix:free_fermion}, we will also provide a rigorous analysis of these operators in the presence of a realistic filter with constant resolution.

 Thanks to the fermionization, the quasi-free dynamics can be solved by tracking the evolution of the one-particle reduced density matrix (1-RDM), which allows simulations up to $L=100$. 
Under this setup, the Lindblad dynamics remains quasi-free and can therefore be computed analytically (see Appendix~\ref{appendix:xx_bulk} for details).  
As shown in \cref{fig:dissipative_XX_fit}, with an energy resolution of $\Delta  = 0.25$, the fitted Luttinger parameter approaches $K=1$ (\cref{eq:bethe_solution} at $J_z=0$) by $t=6$. Notably, during this evolution, the trace norm of the difference between the 1-RDMs of the exact ground state and the evolved state remains on the order of $10^{-1}$ (see \cref{fig:dissipative_XX_trace}). This again demonstrates the key advantage of our approach: identifying quantum phases only requires driving the system into a low-energy manifold where local observables converge (see \cref{sec:free_fermion} for further discussion), rather than achieving exact ground-state preparation. Consequently, the requirements for dissipative evolution time and filter resolution are significantly reduced.

\subsubsection{Spin coupling operators} 
Although fermionic operators $\{c_i^\dagger, c_i\}_{i=1}^{L}$ allow analytical treatment, they are inherently non-local in the spin representation due to Jordan--Wigner strings. For dissipative cooling to be practical on spin-based hardware, the coupling operators should instead be local spin observables, such as the site-wise {Pauli} operators $\{{\sigma_i^z= 2S_i^z}\}_{i=1}^{L}$. We impose open boundary conditions (OBC) for \cref{eq:XXZ}. 
In this setting, the resulting Lindblad dynamics is no longer quasi-free. Therefore, we rely on tensor network simulations to analyze the performance of dissipative cooling.

In \cref{fig:XXZ_spin_coupling} we present the dissipative cooling dynamics with realistic filters generated by {$\{\sigma_i^z \}_{i=1}^L$}. {We choose energy resolutions $\Delta = 1$ and $0.5$ for the realistic filters.} 
The simulations are performed using the tensor network scheme introduced in \cite{ZhanDingHuhnEtAl2025}, up to a total evolution time $t=4$. The system size $L=40$.
In these cases, the initial state still undergoes a rapid cooling phase reaching a relatively low-energy manifold by a total evolution time of $t \approx 2$, followed by a slower saturation towards a higher fidelity with the ground state, see \cref{fig:XXZ_spin_coupling_fidelity}.

{Although an evolution time of $t=4$ is insufficient for the system to enter the long-distance scaling regime needed to reliably extract the Luttinger parameter under OBC, the phase can still be characterized using local observables.}
In \cref{fig:XXZ_spin_coupling_SzSz}, we present snapshots of the average two-point correlation function $C(r) = {\overline{\langle S_i^z S_{i+r}^z \rangle}}$ for $J_z = 0$ at different times along the dissipative cooling dynamics. These results show that the initial long-range order of the N\'eel state is rapidly suppressed, while the correlations move toward the thermodynamic-limit XX-chain behavior~\cite{Lieb1961}
\begin{equation}
    C(r) = \begin{cases}
    0, & r \text{ even};\\
    -1 / {\pi^2 r^2}, & r \text{ odd}.
\end{cases}
\end{equation}
{The rapid emergence of this algebraically decaying behavior as well as the 
alternating pattern provides direct evidence that dissipative cooling can identify the XX phase.}
We see that the behavior of the correlation function can be captured by the dissipative dynamics even at relatively early stages of the evolution with coarse energy filters.
For the results of $(J_{xy}, J_z) = (1, 2)$ in \cref{fig:XXZ_spin_coupling_fit}, this target point lies in the same gapped antiferromagnetic phase as the initial product N\'eel state. In this case, the exponential decay of the connected correlation functions,
\begin{equation}
\langle S_i^z S_{i+r}^z \rangle_c = \langle S_i^z S_{i+r}^z \rangle - \langle S_i^z \rangle \langle S_{i+r}^z \rangle,
\end{equation}
persists throughout the dissipative evolution, indicating the absence of a phase transition during the dynamics.

\begin{figure}
    \centering
    \includegraphics[width=0.85\linewidth]{{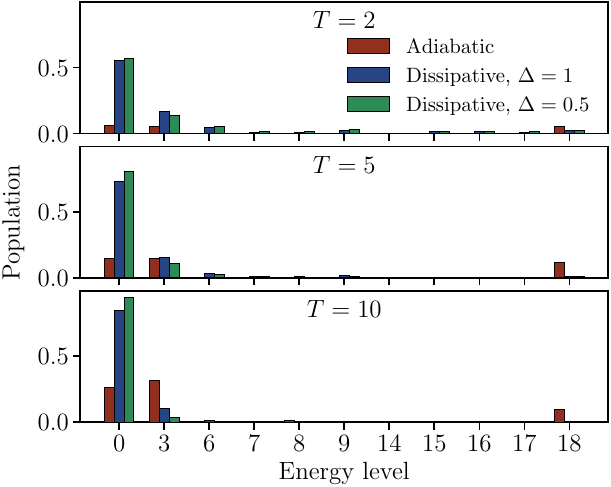}}
    \caption{
    Populations of 20 lowest energy levels for adiabatic and dissipative protocols with different total times $T$. The simulation is performed for $L=10$ using exact diagonalization. Note that we omit the levels with population smaller than $10^{-4}$ during both dynamics.}
    \label{fig:adiabatic_population}
\end{figure}

\subsubsection{Benchmark against adiabatic evolution}
We benchmark the performance of the dissipative cooling protocol using spin coupling operators against adiabatic evolution. To simulate a scenario where the system properties are not known a priori, we employ a simple linear interpolation between the initial Ising Hamiltonian and the target XX Hamiltonian 
for the adiabatic path:
\begin{equation}
    H^{\rm adiabatic}(t) = \left(1- \frac{t}{T}\right) \left(H_{0,1} - 0.1 S_1^z\right) + \frac{t}{T} H_{1,0}.
\end{equation}
Here $H_{J_{xy},J_z}$ denotes the corresponding XXZ Hamiltonian, and $t \in [0, T]$ represents the time variable across a total evolution time $T$. A magnetic field is applied to the initial Hamiltonian to lift degeneracies and ensure a unique initial ground state.

Since this simple adiabatic path crosses a quantum critical point and ends in a gapless phase, one should not expect it to reproduce the target correlations on the short timescales accessible here. As shown in \cref{fig:XXZ_spin_coupling_SzSz}, while our dissipative dynamics correctly capture the scaling of the correlation functions, adiabatic evolution fails to reproduce the polynomially decaying low-energy behavior for total evolution times up to $T=25$.

We further compare the two protocols by examining the final population distribution across the Hamiltonian eigenstates. In \cref{fig:adiabatic_population}, we present the populations of the 20 lowest-energy levels for a system of size $L=10$, obtained via exact diagonalization. Under dissipative evolution, high-energy populations are exponentially suppressed, enabling extraction of ground-state properties at moderate evolution times.
{This suppression can be systematically enhanced by using a finer energy  resolution $\Delta$ at longer total evolution times}, whereas the short-time behaviors of the two resolutions are interestingly close. We will provide a qualitative explanation of this behavior later in \cref{sec:free_fermion}.
Notably, for both cases, the ground-state population exceeds $0.5$ as early as $T=2$.
In contrast, adiabatic evolution retains significant populations in higher energy states even at $T=10$, with the ground-state population remaining below $0.2$.

\section{Low-energy state preparation with realistic filters}\label{sec:theory}

The success of phase probing via dissipative dynamics with realistic filters depends on whether the dynamics prepares low-energy states. In this section, we analyze this question using a drift inequality and three representative regimes. The drift inequality is a Lyapunov-type operator bound that controls the steady-state population at high energies. We first apply this method to free fermions, where the mode decomposition yields a sharp cutoff at the filter resolution and exponential convergence into the corresponding low-energy subspace. We then study free bosons, where the ladder structure of each mode prevents such a strict cutoff and instead yields an exponentially decaying high-energy tail. Proofs are given in Appendix~\ref{appendix:theory}, and additional numerical results appear in Appendix~\ref{appendix:free_and_interacting_MoreNumerical}. Finally, we discuss interacting systems, where local couplings no longer preserve independent mode occupations and numerical results indicate a population profile resembling the bosonic leakage mechanism.

\begin{figure} 
    \centering
    \begin{subfigure}[b]{0.5\linewidth}
        \overlaption[fig:imperfect_ff]{\includegraphics[width=\linewidth]{{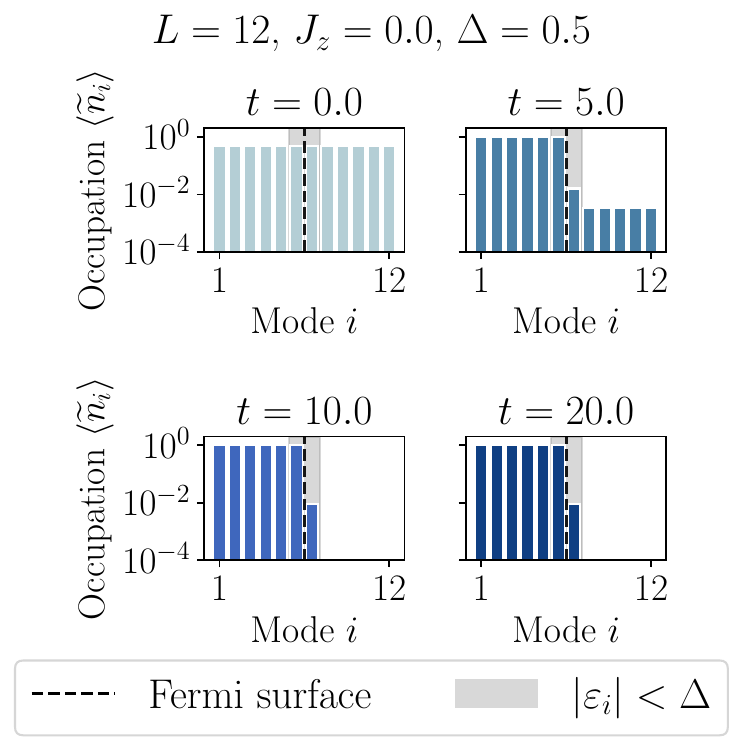}}}
    \end{subfigure}
    \hfill
    \begin{subfigure}[b]{0.47\linewidth}
        \overlaption[fig:imperfect_fb]{\includegraphics[width=0.95\linewidth]{{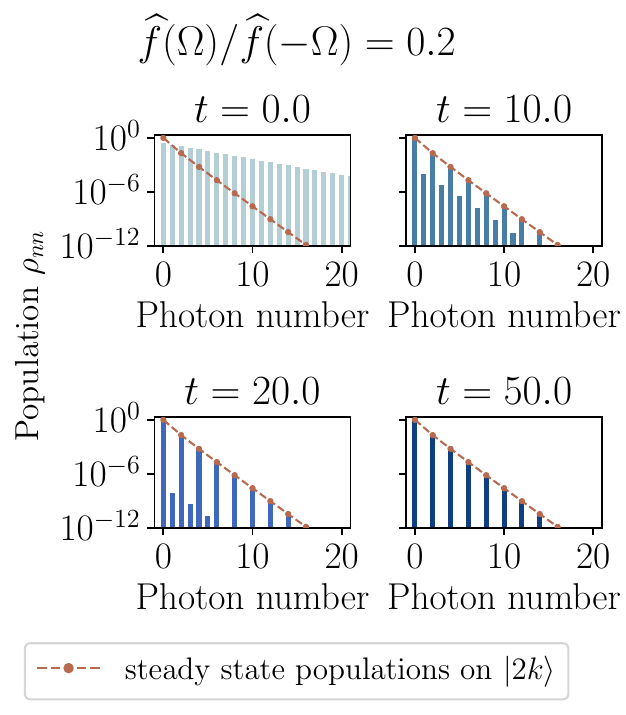}}}
    \end{subfigure}
    \caption{\label{fig:imperfect_ff_fb} Representative evolution of \textbf{(a)} canonical occupation numbers for free fermions, and of \textbf{(b)} state populations for free bosons.}
\end{figure}

\subsection{Drift inequality}

We use a drift-inequality framework to control low-energy state preparation and convergence for dissipative dynamics with realistic filters. The key step is to construct a positive semidefinite drift operator $\mathsf G = \sum_k \mathfrak g_k \dyad{\psi_k}$ that is diagonal in the energy eigenbasis $\{\ket{\psi_k}\}$ of the Hamiltonian $H$, and to establish an operator inequality of the form
\begin{equation}
  \mc L^\dag[\mathsf G] \le -a \mathsf G + b I, 
\end{equation}
for constants $a>0$ and $b\ge 0$. This inequality gives a bound for $\Tr[\mathsf G\rho(t)]$, and the operator Markov's inequality then bounds the steady-state population above an energy threshold $E$ by
\begin{equation}
    \Tr[\mc P_{>E} \rho_{\rm ss}] \le \frac{b}{a \mathfrak g_E},\quad \text{if}\quad \mathfrak g_E = \inf_{k:\lambda_k > E} \mathfrak g_k > 0.
\end{equation}
Here, $\mc P_{>E} := \mathbf 1_{(E,\infty)}(H)$ denotes the spectral projector onto the subspace of states with energy above $E$, and $\lambda_k$ is the energy of the eigenstate $\ket{\psi_k}$. 
The choice of $\mathsf G$ is system-dependent. 
For free fermionic systems, we use the canonical number operator $N_\Delta$ that counts quasiparticle modes with excitation energies above the filter resolution $\Delta$. For single-mode free bosonic systems, we instead use an exponential moment of the Hamiltonian, $e^{\beta H}$, with a parameter $\beta>0$ chosen below the stability threshold. Lemma~\ref{lem:drift} in Appendix~\ref{appendix:theory} gives the finite-dimensional version of the argument, while the bosonic tail bound uses the corresponding unbounded drift estimate proved in Appendix~\ref{appendix:boson}.

\subsection{Free fermions}
\label{sec:free_fermion}

We first consider a generic quadratic fermionic Hamiltonian
\begin{equation}
    H = \sum_{i,j=1}^L F_{ij} c_i^\dag c_j
    \label{eq:hamil_fermion}
\end{equation}
where $F$ is a Hermitian coefficient matrix. We choose the coupling operators to be 
\begin{equation}
    \{c_i^\dag,c_i\}_{i=1}^{L}.
    \label{eq:coup:ferimon}
\end{equation} We assume that the coefficient matrix can be unitarily diagonalized as $F =  U \text{diag}(\varepsilon_1, \ldots, \varepsilon_L) U^\dagger$. 
The Hamiltonian can then 
be diagonalized as 
\begin{equation}
    H = \sum_{k = 1}^{L} \varepsilon_k \wt{c}^{\dagger}_k \wt{c}_k,\quad \wt{c}_k = \sum_{j=1}^L U_{jk} c_j,
\end{equation}
and the ground state energy is given by $E_0 = \sum_{\varepsilon_k<0} \varepsilon_k$.

\begin{thm}[Low-energy state preparation in free-fermion dissipative dynamics, informal]
    \label{thm:free_fermion}
  Consider the free-fermion Hamiltonian \cref{eq:hamil_fermion} and coupling operators \cref{eq:coup:ferimon}. Given a realistic filter with resolution $\Delta$ and any initial state, the total population $\Tr(\mc P_{>E_\Delta}\rho(t))$ in the high-energy subspace above $E_\Delta := E_0 + \sum_{|\varepsilon_k| \le \Delta} |\varepsilon_k|$ decays exponentially in time. Consequently, the state becomes close, in trace norm, to its projection onto the subspace with energy at most $E_\Delta$. Moreover, the convergence rates of the energy and occupation numbers are exponential and independent of the system size.
\end{thm}
The proof of \cref{thm:free_fermion} is deferred to Appendix~\ref{appendix:free_fermion}. Since $\Tr(\mc P_{>E_\Delta} \rho(t))$ decays exponentially, we have $\Tr(\mc P_{>E_\Delta} \rho_{\rm ss})=0$ {in the infinite-time limit}. {Thus the steady state is strictly supported on the subspace with energy at most $E_\Delta$}. The derivation uses the fact that the chosen coupling operators decouple the dynamics of the canonical fermionic modes $\wt{c}_k$. In this representation, the occupation of each canonical mode, $\langle  {n}_k \rangle = \langle \wt{c}_k^\dagger \wt{c}_k \rangle$, converges exponentially toward the steady-state value
\begin{equation}
\braket{ {n}_k}_{\text{ss}} = \frac{\hat{f}(\varepsilon_k)^2}{\hat{f}(-\varepsilon_k)^2 + \hat{f}(\varepsilon_k)^2},
\label{eq:nkss}
\end{equation}
at a relaxation rate given by $\hat{f}(-\varepsilon_k)^2 + \hat{f}(\varepsilon_k)^2$. Based on our construction of the filter, modes with $|\varepsilon_k| > \Delta$ resolve rapidly to the ground state configuration at a constant rate, while the imperfect nature of the realistic filter is confined to the low-energy window defined by the resolution $\Delta$.
For a gapped free-fermion system whose smallest quasiparticle excitation energy is $\Delta_H$, any initial state rapidly relaxes toward the ground state provided that the filter resolution satisfies $\Delta < \Delta_H$. In contrast, for a gapless system, a fixed resolution $\Delta$ guarantees cooling into a low-energy manifold rather than the exact ground state.

We numerically illustrate \cref{thm:free_fermion} in \cref{fig:imperfect_ff}, where we consider the tight-binding model with OBC
\begin{equation}\label{eq:XX_fermionized}
     H = \sum_{i=1}^{L-1} \left( c_{i+1}^\dag c_i +  c_i^\dag c_{i+1} \right),\quad \text{$L$ is even.}
\end{equation}
\cref{eq:XX_fermionized} can also be viewed as the XX Hamiltonian up to a Jordan--Wigner transformation, which is a special case of the XXZ Hamiltonian discussed in \cref{sec:XXZ}.
Its ground state is a half-filled Fermi sea. The Hamiltonian can be diagonalized with $\varepsilon_k = -2\cos\left[{k\pi}/(L+1) \right]$, which is negative when $k \le L / 2$ and positive when $k > L / 2$, and thus
the ground state is given by $\ket{\psi_0} = \prod_{k=1}^{L/2} \wt c_k^\dag \ket{0}_{\wt c}$. 
Since the resulting dissipative dynamics is quasi-free and particle-number preserving, its {1-RDM} evolves autonomously and can be solved using the techniques described in \cite{BarthelZhang2022}. In particular, the population of the $k$-th canonical mode $\langle \wt n_k\rangle$ can be directly read out from the diagonal entries of the unitary-transformed 1-RDM. 

We choose the energy resolution of the realistic filter to be $\Delta = 0.5$, which is larger than the spectral gap. The resulting steady state population closely approximates the ideal half-filling distribution, with deviations confined to a small subset of modes near the Fermi surface. These modes, whose energies $|\varepsilon_k|$ fall within the filter's resolution window $\Delta$, are highlighted by the gray shaded regions in \cref{fig:imperfect_ff}.

The total leakage energy $E_{\rm exc}$ of this gapless free fermion Hamiltonian can be bounded by
\begin{equation}
    \begin{aligned}
        E_{\rm exc} & \le \sum_{|\varepsilon_k| \le \Delta} |\varepsilon_k|  \approx \frac{L+1}{\pi} \int^{\arccos (-\frac{\Delta}{2})}_{\arccos \frac{\Delta}{2}} 2 |\cos x| \dInt x\\
        & = \frac{4(L+1)}{\pi} \left[ 1 - \sqrt{1 - (\Delta/2)^2} \right]\\
        & = \frac{L+1}{2\pi}\Delta^2 + \varO((L+1)\Delta^4).
    \end{aligned}
\end{equation}
It is well established that an extensive excitation energy $E_{\rm exc}$ corresponds to an effective inverse temperature $\beta \propto \sqrt{L / E_{\rm exc}} = \mathcal{O}(1 / \Delta)$, as derived from the Sommerfeld expansion~\cite{ashcroft1976solid}. Generally, approximating a ground state with a thermal state at inverse temperature $\beta$ requires $\beta = \Omega(1/\Delta_H)$ to resolve the spectral gap $\Delta_H$, which in this case closes as $\Delta_H=\varO(1/L)$. However, if the objective is restricted to estimating ground-state \emph{local} observables, these requirements are significantly relaxed; in such cases, $\beta$ need only scale polynomially with the subsystem size.

Let us consider the two-point correlation function, or the 1-RDM element $P_{ji} = \langle c_i^{\dagger} c_j \rangle$. For small $\Delta$, evaluating the discrepancy between the steady-state and ground-state expectation values yields $|\langle c_i^{\dagger} c_{i+r} \rangle_{\rm ss} - \langle c_i^{\dagger} c_{i+r} \rangle_{\rm GS}| = \varO(\Delta^2  r)$ for sites deep in the bulk of the system (see Appendix~\ref{appendix:1rdm_error} for {detailed analysis}).
Thus, to bound the error by a constant, it suffices to choose $\Delta = \varO(1 / \sqrt{r})$, or equivalently $\beta = \varO(\sqrt{r})$.
For our phase-probing application, this implies that a cooling protocol with a realistic filter can effectively estimate observables within a subsystem up to a spatial scale $r = \varO(1 / \Delta^2)$. This demonstrates that accessing an approximate low-energy subspace is sufficient for extracting the bulk signatures of underlying zero-temperature phase transitions in practical settings.
Furthermore, local rapid mixing is stable against suitable local noise perturbations~\cite{Cubitt2015localdissiaptive,Lucia2015localdissiaptive}, which supports the robustness of this mechanism in experimental implementations.

\subsection{Free bosons}\label{sec:free_bosons}
The reduction to independent two-level occupation dynamics is a special property of free fermions and does not hold in general. To investigate a more complex and representative scenario, we consider bosonic systems, where each bosonic mode contains an infinite ladder of coupled energy levels. Specifically, let us consider the following $L$-mode free bosonic Hamiltonian:
\begin{equation} \label{eq:hamil_boson}
    H = \sum_{i,j=1}^L h_{ij} a_i^\dag a_j, 
\end{equation}
where $h$ is a real-symmetric positive definite matrix whose eigenvalues $\varepsilon_i>0$ are the canonical mode frequencies. We choose the coupling operators to be
\begin{equation} \label{eq:coup_boson}
    A_i = a_i + a_i^{\dagger},\quad i = 1,2,\ldots,L.
\end{equation}
In order to discuss the effect of realistic filters, 
we assume that every relevant mode frequency satisfies $\varepsilon_i<\Delta$ and that the filter has stronger cooling than heating at these frequencies, namely $0 \le \hat{f}(\varepsilon_i) < \hat{f}(-\varepsilon_i) \le 1$.
This simple model can describe, for example, optical cavity modes coupled to external modes with both cooling and heating drives.
Since the cooling is stronger than heating, the corresponding Gaussian steady state is stable and takes the form of a squeezed thermal state when the coherent Hamiltonian contribution is included \cite{GrynbergAspectFabre2010,MarianMarian1993}. Here we use the drift inequality to control the low-energy character of the steady state. For simplicity, we first analyze a single-mode bosonic system with Hamiltonian $H= \Omega a^\dag a$ with {$\Omega>0$}.
\begin{thm}[Low-energy state preparation in free-boson dissipative dynamics, informal]\label{thm:boson_population}
    Consider the single-mode free bosonic Hamiltonian and coupling operator $A = a + a^{\dagger}$, and assume $0\le \wh f(\Omega)<\wh f(-\Omega)$. For any parameter $0<\beta < \frac{1}{\Omega}\log \frac{\wh f(-\Omega)}{\wh f(\Omega)}$, there exists a constant $C_\beta > 0$ such that for any $E>0$, the tail population of the steady state above $E$ can be bounded as
\begin{equation}
    \Tr[\mc  P_{>E} \rho_{\rm ss}] \le C_\beta e^{-\beta E}.
\end{equation}
\end{thm}
The proof of the above theorem is deferred to \cref{thm:boson-drift} in Appendix~\ref{appendix:boson}.
In contrast to the free-fermion case, the free-boson steady state does not exhibit a sharp energy cutoff; instead, the total population above some energy threshold $E$ decays exponentially in $E$. The reason is that the jump operator $K=\wh f(-\Omega)a+\wh f(\Omega)a^\dag$ couples adjacent Fock levels of the same bosonic mode, so filter leakage propagates through the full Fock ladder rather than remaining confined to a finite set of two-level occupations.

We also obtain trace-distance convergence bounds to the steady state.
In the purely dissipative case, i.e. $\mc L = \mc D$ in \cref{eq:lindblad}, the steady state is a pure state known as the \emph{squeezed vacuum state}. For Gaussian initial states, the covariance-matrix analysis can include the coherent contribution $\mc L_H$, in which case the steady state becomes a \emph{squeezed thermal state}.

\begin{prop}
    \label{prop:free_boson}
{Consider the free-boson Hamiltonian in \cref{eq:hamil_boson} with the coupling operators defined in \cref{eq:coup_boson}. Under the same realistic filter and stability assumptions as above, the following results hold.}
\begin{enumerate}
    \item If $\mc L_H = 0$ so that $\mathcal L = \mathcal D$, then any initial state with finite initial squeezed-mode occupation converges exponentially fast in trace distance to a tensor product of squeezed vacuum states. The bound has an initial-state-dependent prefactor and decay factor $e^{-\gamma t/2}$, where
\begin{equation}
    \gamma :=
    \min_{i=1,\cdots,L}
    \left[
        \widehat f(-\varepsilon_i)^2
        -
        \widehat f(\varepsilon_i)^2
    \right].
\end{equation}

    \item In the single-mode case $L=1$, the full Lindblad dynamics converges exponentially fast in trace distance from any initial Gaussian state to the Gaussian squeezed thermal steady state. The convergence occurs with decay rate $\left[\widehat f(-\Omega)^2-\widehat f(\Omega)^2\right]/2$, up to an initial-state-dependent prefactor.
\end{enumerate}
\end{prop}

The proofs of Proposition~\ref{prop:free_boson} are deferred to Appendix~\ref{app:additional_boson}. Quadratic dissipative bosonic systems with linear jump operators have been extensively studied~\cite{Prosen2010,BarthelZhang2022}, yet sharp results on convergence rates in trace distance remain limited. 
Our analysis uses two ingredients. The proof of the first part employs a Lyapunov-function approach for purely dissipative dynamics, while the proof of the second part uses recent trace-distance bounds for Gaussian bosonic states~\cite{ZhanDingHuhnEtAl2025,BittelMeleTironeLami2025,BittelMeleEisertLeone2025}.

To further illustrate the effect of realistic filters and validate the theoretical results, we carry out an illustrative numerical simulation. 
For simplicity, we numerically demonstrate only the purely dissipative case for a single-mode bosonic system with Hamiltonian $H = \Omega a^\dag a$ and $\Omega>0$
in \cref{fig:imperfect_fb}, where we choose the leakage strength to be $s = \wh f(\Omega) / \wh f(-\Omega) = 0.2$. 
We initialize the system in a mixed Gaussian state with a mean photon number of $\langle a^\dag a\rangle_0 = 2$, a two-photon coherence of $\langle a^2\rangle_0 = 0.01$, and a vanishing first moment $\langle a\rangle_0=0$. 
We plot the individual Fock state populations, $\rho_{n,n} = \mel{n}{\rho}{n}$, which converge to those of a squeezed pure state (red dashed lines). We observe that the free-boson steady state does not exhibit a sharp energy cutoff; instead, we observe exponentially suppressed populations in arbitrarily high-energy Fock states $\ket{2k}$, which is consistent with the theoretical discussions above.

\begin{figure}
    \centering
    \includegraphics[width=0.99\linewidth]{{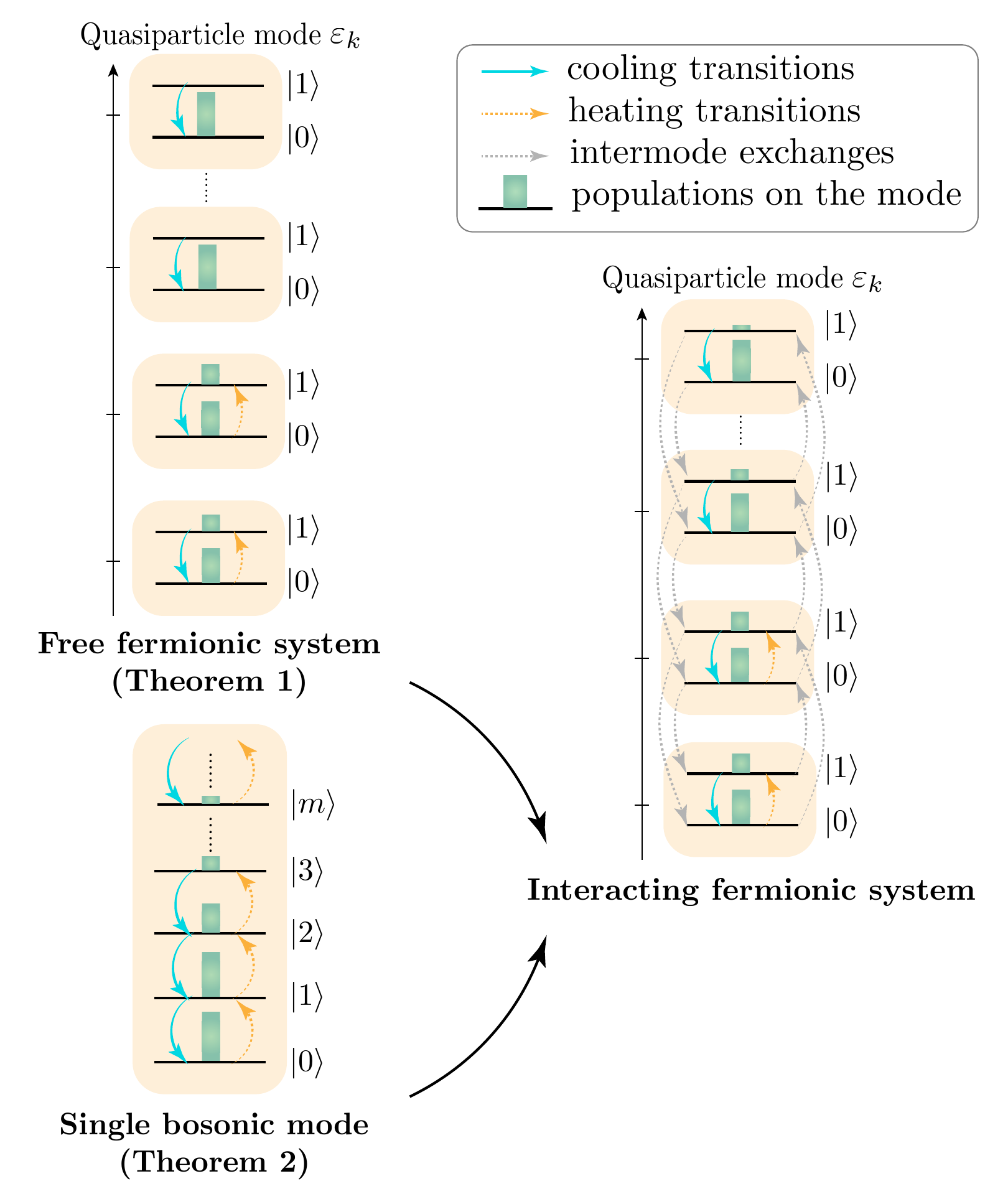}}
    \caption{Schematic illustration of the mechanism behind the steady-state population distribution under a realistic filter for free fermions, free bosons, and interacting fermions. The interacting fermionic system is expected to combine characteristics of both the free-fermion and free-boson cases. It exhibits a spillover effect slightly above the Fermi surface due to the finite filter resolution, similar to the free-fermion case, as well as intermode exchanges that couple different modes in the system. In terms of canonical occupations, these exchanges can populate modes beyond the nominal resolution window with a decaying tail, similar to the free-boson case. For clarity, only a subset of these exchanges is shown in the figure.}
    \label{fig:ff_fb_int_mechanism}
\end{figure}

\begin{figure} 
    \centering
    \begin{subfigure}[b]{0.543\linewidth}
        \overlaption[fig:imperfect_int_mode]{\includegraphics[width=\linewidth]{{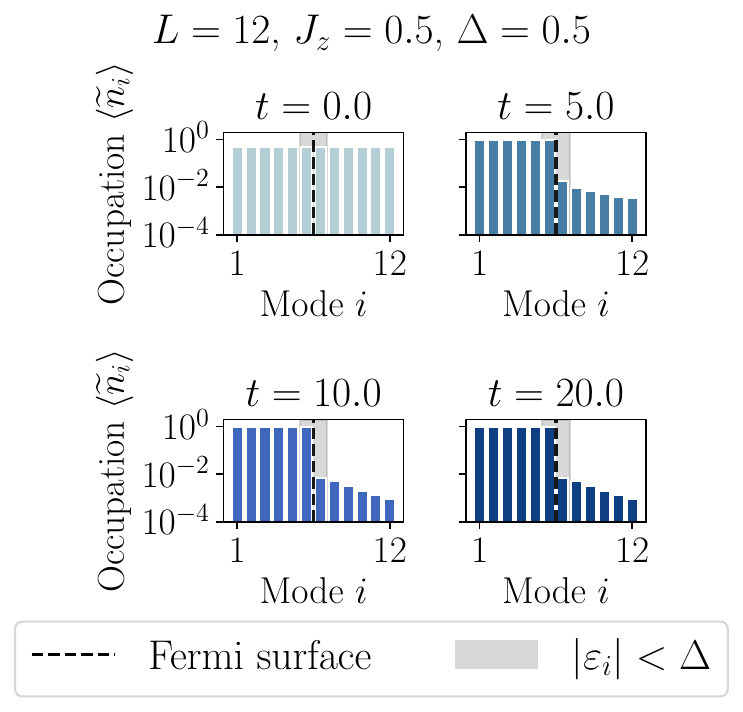}}}
    \end{subfigure}
    \hfill
    \begin{subfigure}[b]{0.437\linewidth}
        \overlaption[fig:imperfect_int_level]{\includegraphics[width=\linewidth]{{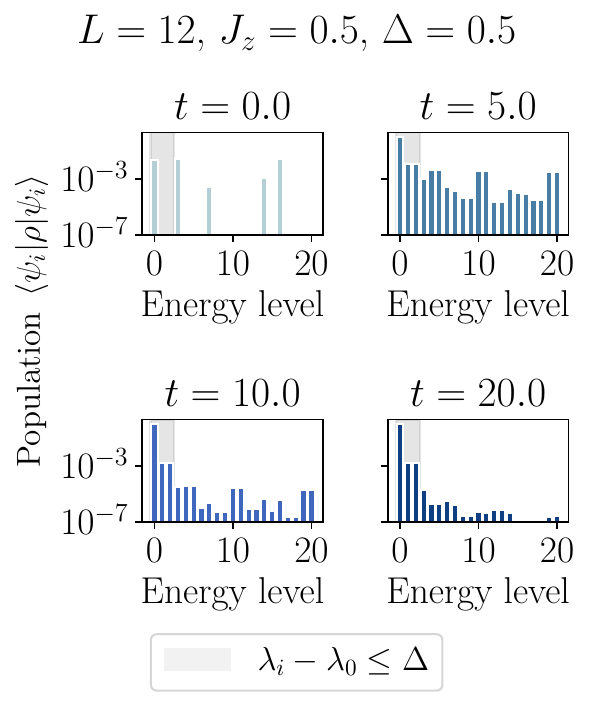}}}
    \end{subfigure}
    \caption{\label{fig:imperfect_int} 
    Representative evolution of \textbf{(a)} canonical occupation numbers and \textbf{(b)} state populations under a realistic filter for interacting fermions.}
\end{figure}

\subsection{Interacting systems}

So far, our analysis has focused on free systems. For interacting models, local coupling operators generally do not preserve the mode-by-mode decomposition that made the free-fermion dynamics exactly solvable. Instead, the same local observable can connect many-body eigenstates that differ in several single-particle occupations.
Consequently, the steady-state behavior under dissipative dynamics generated by realistic filters is expected to show leakage across the many-body spectrum. The resulting population profile can resemble the free-boson case, with a decaying tail, as schematically illustrated in \cref{fig:ff_fb_int_mechanism}.

In the drift inequality framework, we would like to choose $E_0$, and find  $\beta>0$, $a_\beta>0$, and $b_\beta<\infty$ such that
\begin{equation}\label{eq:interacting-drift-target}
    \mathsf G_\beta := e^{\beta(H-E_0)},\qquad
    \mc L^\dag[\mathsf G_\beta] \le -a_\beta \mathsf G_\beta + b_\beta I.
\end{equation}
By the drift argument above, such an estimate would imply an exponential steady-state tail,
\begin{equation}
    \Tr[\mc P_{>E}\rho_{\rm ss}] \le C_\beta e^{-\beta(E-E_0)}
\end{equation}
for energies $E>E_0$ and a constant $C_\beta$ depending on the drift constants. Intuitively, the downward transitions should still dominate upward leakage outside the filter-resolution window. But proving \cref{eq:interacting-drift-target} requires quantitative control of the filtered transition rates in the many-body energy basis and we leave this problem for future work. Below, we provide numerical evidence for this exponential-tail picture.

Let us take a simple illustrative Hamiltonian $H=H_0+J_z V$ with
\begin{equation}\label{eq:XXZ_fermionized}
    \begin{aligned}
          H_0   &= \sum_{i=1}^{L-1} \left( c_{i+1}^\dag c_i +  c_i^\dag c_{i+1} \right),\\
               V &= \sum_{i=1}^{L-1}\left( n_{i+1} - \frac{1}{2} \right) \left( n_i - \frac{1}{2} \right) .
    \end{aligned}
\end{equation}
Here $n_i = c^{\dagger}_i c_i$. This is the spinless-fermion form of the open-boundary XXZ chain, up to the standard Jordan--Wigner transformation and an additive constant.
In \cref{fig:imperfect_int_mode}, we demonstrate the dynamics of the canonical population operators $\langle \wt{n}_i \rangle$, defined with respect to the diagonal modes of the noninteracting Hamiltonian $H_0$. We set the interaction strength $J_z = 0.5$, while maintaining a filter resolution of $\Delta = 0.5$, which remains larger than the free fermion spectral gap. The data show that the noninteracting canonical occupations are no longer sharply cut off at $|\varepsilon_i|=\Delta$; instead, modes above the nominal resolution window acquire a decaying tail. This behavior is consistent with the bosonic case, where the dissipative dynamics couples an infinite ladder of levels rather than independent two-level modes.
These numerics are outside the direct scope of existing perturbative analyses for Lindblad-based thermal-state and ground-state preparation in weakly interacting fermionic systems, for
example in \cite{ZhanDingHuhnEtAl2025,TongZhan2025,ChenRouzeChenEtAl2025,SmidMeisterBertaMario2025}.

For interacting systems, populations in the many-body energy eigenbasis are a more direct diagnostic than occupations of the noninteracting canonical modes. The data in \cref{fig:imperfect_int_level} show a rapidly decaying high-energy tail whose width is controlled by the filter resolution, which is qualitatively similar to the bosonic case.

\section{Discussion}
\label{sec:discussion}

We have shown that short-time dissipative cooling can probe ground-state quantum phases without full ground state preparation. The main cost can be measured by the total Hamiltonian simulation time, which is determined by the filter resolution used to construct the jump operators and by the dissipative evolution time. In our examples, finite-resolution filters already produce informative phase-sensitive diagnostics well before global mixing, including in BKT-type and covariance-matrix based topological phase decisions. Combined with recent progress on dissipative thermalization and ground-state preparation~\cite{lloyd2025quantumthermal,langbehn2025universal,HahnSwekeDeshpandeShtanko2026,ZhanDingHuhnEtAl2025,WangDing2026,slezak2026polynomial,ChenDingZhang2026}, we argue that dissipative phase identification is a plausible target for early fault-tolerant quantum devices.

Metastability arises in many settings, including long-range interacting and periodically driven (Floquet) quantum many-body systems~\cite{YinSuraceLucas2025,Defenu2021,LiuLundgrenTitumEtAl2019,ColluraDeLucaRossiniEtAl2022}, and it has also received increasing attention in open quantum many-body systems~\cite{MacieszczakGutaLesanovskyGarrahan2016,LandaSchiroMisguich2020,MacieszczakRoseLesanovskyGarrahan2021,LiRoseGarrahanLuitz2022,SchutzJagerMorigi2016,RakovszkyGopalakrishnanKeyserlingk2024,ChenHuangPreskillZhou2024,BergamaschiChenVazirani2025}. In such regimes, short-time dissipative dynamics may still carry phase information, but long-lived nonequilibrium sectors {or local energy minima} can obstruct reliable phase identification on practical timescales. Extending the method to these settings may require quantum analogues of enhanced-sampling methods, designed to move between metastable sectors, such as parallel tempering or replica exchange~\cite{BelloRivasElber2015,LaydenMazzolaMishmashEtAl2023,LengDingChenEtAl2026,ChenBassoDingEtAl2025}.

The examples studied in this work are benchmark applications and are classically accessible at the simulated sizes and in the exactly solvable or quasi-one-dimensional regimes considered here. A natural next step is to push toward more challenging regimes, such as frustrated two-dimensional $J_1$--$J_2$ models or perturbed Kitaev honeycomb models, and to move from deep within a phase toward phase boundaries. 
It remains an open question whether classical methods for these harder regimes retain polynomial scaling in the relevant length scales, or develop prohibitive dependence on the correlation length as the phase boundary is approached. For such applications, we do not expect that a quantum advantage can be established from \emph{a priori} complexity-theoretic arguments. Instead, quantum advantages will likely be empirical and depend on the physical details of the model and the phase transition. A practical strategy is to begin with benchmark cases, then move toward phase boundaries and identify where such approaches can offer a genuine computational advantage.

\begin{acknowledgments}
    This material is based upon work supported by the U.S. Department of Energy, Office of Science, Accelerated Research in Quantum Computing Centers, Quantum Utility through Advanced Computational Quantum Algorithms, grant no. DE-SC0025572 (H.L., Y.Y., L.L.). Additional support is from the U.S. Department of Energy, Office of Science, National Quantum Information Science Research Centers, Quantum Systems Accelerator (H.L., L.L.). L.L. is a Simons Investigator in Mathematics. The authors thank Garnet Chan, Yu-Jie Liu, Kevin Stubbs and Chao Yin for insightful discussions.
\end{acknowledgments}

\bibliography{ref}

\clearpage
\newpage

\widetext
\appendix
\setcounter{secnumdepth}{3}

\section{Choice of filter functions}\label{appendix:filter}

To construct the \emph{realistic} filters used in \cref{sec:filter}, we specify the Fourier transform $\hat{f}(\omega)$ of the time-domain filter function $f(s)$ as in~\cite{DingChenLin2024}:
\begin{equation}
    \hat{f}(\omega) = \frac{1}{2} \left[\mathrm{erf}\left( \frac{\omega + a}{\Delta_a} \right) - \mathrm{erf}\left( \frac{\omega + b}{\Delta_b} \right)   \right], \quad 0 \le b < a,
\end{equation}
where
\begin{itemize}
    \item ${\rm erf}(x) := \frac{2}{\sqrt{\pi}}\int_{0}^{x} e^{-t^2} \, \dInt t$ is the error function.
    \item The parameters $a$ and $\Delta_a$ control the turn-on on the negative-frequency side. If $k > 0$ is chosen so that ${\rm erf}(k) \approx 1$, then $\hat{f}(\omega)$ starts to become nonzero near $\omega = -a-k\Delta_a$ and is already close to $1$ by $\omega = -a+k\Delta_a$. In practice we take $k=2$.
    \item The parameters $b$ and $\Delta_b$ control the turn-off near $\omega = 0$. More precisely, the second error-function factor changes from approximately $-1$ to approximately $1$ across the window $[-b-k\Delta_b,-b+k\Delta_b]$, so the cutoff is centered at $\omega=-b$ with width $\varO(\Delta_b)$.
    \item Consequently, the plateau where $\hat{f}(\omega) \approx 1$ is approximately the interval $[-a+k\Delta_a,-b-k\Delta_b]$. One should therefore choose $a-b$ larger than $k(\Delta_a+\Delta_b)$.
\end{itemize}

In the physically relevant cutoff window near $\omega=0$, $\hat{f}(\omega)$ is monotonically non-increasing. We choose $a$ sufficiently large that the left turn-on does not affect the frequencies in $[-\Delta,\Delta]$, and then use $b$ and $\Delta_b$ to set the desired cutoff. If we only require $\hat{f}(-\Delta) \approx 1$ and $\hat{f}(\Delta) \approx 0$, then the filter allows \emph{spillover}, namely a small amount of heating leakage for frequencies in $(0,\Delta)$. A convenient choice is $b = 0$ and $\Delta_b = \Delta/k$, for which $\hat{f}(\Delta) \approx 0$ while $\hat{f}(0)$ remains of order one; with $k=2$, this gives $\Delta_b = \Delta/2$, as used in \cref{fig:filter}. If instead we want a \emph{sharp cutoff} with $\hat{f}(0) \approx 0$, then we set $b = \Delta/2$ and $\Delta_b = \Delta/(2k)$, which places $\omega=0$ at the right edge of the transition window. Unless stated otherwise, we use the spillover filter in this work.

\begin{figure}[b]
    \centering
    \includegraphics[width=0.4\linewidth]{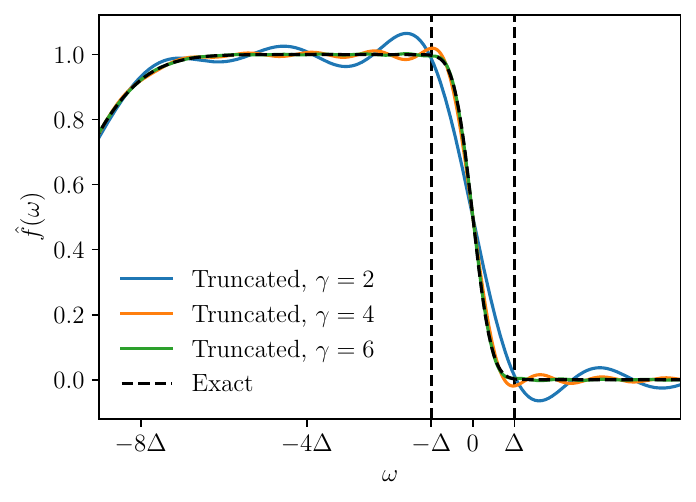}
    \caption{The ``spillover'' type filter function in frequency domain. Here $b = 0$, $\Delta_b = \Delta / 2$, $a = 10 \Delta$, and $\Delta_a = a / 5$.}
    \label{fig:filter}
\end{figure}

In the time domain, the filter is given by
\begin{equation}
    f(s) = \frac{1}{2\pi} \int_{\RR} \hat{f}(\omega) e^{-\I \omega s} \dInt \omega = \frac{e^{-\Delta_a^2 s^2/4}e^{\I a s} - e^{-\Delta_b^2 s^2/4}e^{\I b s }}{2\pi \I s}.
    \label{eq:filter_exact}
\end{equation}
The singularity at $s=0$ is removable because $\lim_{s\to 0} f(s) = (a - b)/2\pi$, hence $f(s)$ is continuous.
In practice, the jump operator is implemented by truncating the integral and subsequently applying the trapezoidal rule,
\begin{equation}
    K_a \approx \sum_{s} f(s)  e^{\I Hs} A_a e^{-\I Hs} \tau,
    \label{eq:filter_trunc}
\end{equation}
where $\tau = \varO(1 / a)$ is the time step. To resolve the structure of $\hat{f}(\omega)$ in the crossover window $[-b-k\Delta_b, -b+k\Delta_b]$, the maximal evolution time should be chosen inversely proportional to $k\Delta_b$:
\begin{equation}
    s_{\max} = \gamma /( k \Delta_b).    
\end{equation}
\Cref{fig:filter} compares the exact filter in \cref{eq:filter_exact} with the truncated filter in \cref{eq:filter_trunc} for several values of the truncation parameter $\gamma$. To suppress oscillations caused by truncation, we take $\gamma$ moderately large. For the tensor-network simulations in \cref{sec:XXZ}, we choose $\gamma = 4$.

\section{Theoretical results for dissipative dynamics with realistic filters}\label{appendix:theory}

In this appendix, we present theoretical results on dissipative dynamics with realistic filters for free fermionic and bosonic systems. Our proof relies on a drift inequality that bounds high-energy populations and the convergence rate of the system, as stated in the following lemma.
\begin{lem}\label{lem:drift}
  Let $H$ be a finite-dimensional Hamiltonian with spectral decomposition $H = \sum_k \lambda_k \dyad{\psi_k}$. 
  Let $\mathsf G :=  \sum_k \mathfrak g_k \dyad{\psi_k}$ be diagonal in the Hamiltonian eigenbasis with non-negative eigenvalues $\mathfrak g_k \ge 0$ for all $k$.
  Assume that there exist constants $a>0$ and $b\ge 0$ such that $\mathsf G$ satisfies the drift inequality
  \begin{equation}\label{eq:drift-G}
    \mathcal{L}^\dagger[\mathsf G] \le -a \mathsf G + b I.
  \end{equation}
  For any energy threshold $E$, define the projector onto the high-energy subspace as $\mc P_{> E} := \mathbf{1}_{(E,\infty)}(H) = \sum_{k:\lambda_k > E} \dyad{\psi_k}$, and $\mathfrak g_E := \inf_{k:\lambda_k > E} \mathfrak g_k$.
  If $\mathfrak g_E > 0$, then for any initial state $\rho(0)$, we have the following bound on the high-energy population at time $t$:
  \begin{equation}\label{eq:drift_markov_bound}
    \Tr[\mc P_{> E} \rho(t)] \le \frac{1}{\mathfrak g_E} \left( e^{-a t} \Tr[\mathsf G \rho(0)] + (1-e^{-a t})\frac{b}{a} \right).
  \end{equation}
  In particular, for the stationary state $\rho_{\rm ss}$, we have by taking the limit $t \to \infty$ from \cref{eq:drift_markov_bound} that
    \begin{equation}\label{eq:drift_markov_bound_ss}
        \Tr[\mc P_{> E} \rho_{\rm ss}] \le \frac{b}{a \mathfrak g_E}.
    \end{equation}
\end{lem}
\begin{proof}
 Set $
    \mathsf m(t):=  \Tr(\mathsf G \rho(t)).$ Since $\rho(t)$ solves the Lindblad equation, the Hilbert--Schmidt adjoint gives $
        \dv{\mathsf m(t)}{t} = \Tr(\mathcal{L}^\dagger[\mathsf G] \rho(t)). $   Since $\rho(t)$ is positive semidefinite, \cref{eq:drift-G} gives
    \begin{equation}
        \dv{\mathsf m(t)}{t} \le -a \mathsf m(t) + b.
    \end{equation}
    By Gr\"onwall's inequality, we have
    \begin{equation}
      \mathsf m(t) \le e^{-a t} \mathsf m(0) + (1-e^{-a t})\frac{b}{a}.
    \end{equation}
    By the definition of $\mathfrak g_E$, we have $\mathfrak g_E \mc P_{> E} \le \mathsf G$. This can be viewed as the operator form of Markov's inequality, which converts the moment bound to a tail bound by taking the trace against $\rho(t)$:
    \begin{equation}
        \Tr[\mc P_{> E} \rho(t)] \le \frac{1}{\mathfrak g_E} \Tr[\mathsf G \rho(t)] \le \frac{1}{\mathfrak g_E} \left( e^{-a t} \Tr[\mathsf G \rho(0)] + (1-e^{-a t})\frac{b}{a} \right).
    \end{equation}
    Taking the limit $t \to \infty$ gives the bound for the steady state \cref{eq:drift_markov_bound_ss}.
\end{proof}

\subsection{Free fermionic systems}
\label{appendix:free_fermion}

Let us consider a general quadratic fermionic Hamiltonian 
\begin{equation}
    H = \sum_{i,j=1}^L F_{ij} c_i^\dag c_j,
\end{equation}
and bulk coupling operators $\{c_i^\dag,c_i\}_{i=1}^{L}$ to construct the Lindbladian jump operators for dissipative cooling dynamics. 
The term ``bulk'' means that we apply the coupling on all sites of the system. 
We first simplify the free-fermionic Hamiltonian using the following two steps:
\begin{enumerate}
    \item Diagonalize the Hermitian matrix $F = U \diag(\varepsilon_k)_{k=1}^L U^\dag$, and choose $\wt{c}_k^\dag = \sum_{j=1}^L U_{jk} c_j^\dag $ such that
    \begin{equation}
        H = \sum_{k=1}^L \varepsilon_k \wt{c}_k^\dag \wt{c}_k.
    \end{equation}
      {We call $\wt{c}_k$ the canonical modes and assume  $\varepsilon_k \neq 0$ for all $k$. }
    \item Apply the particle--hole transformation
    \begin{equation}
     b_k := \begin{cases}
        \wt{c}_k, & \text{if } \varepsilon_k > 0,\\
        \wt{c}_k^\dag, & \text{if } \varepsilon_k < 0,
     \end{cases}
    \end{equation}
    so that each mode has positive excitation energy $\lambda_k$:
    \begin{equation}
        H  = E_0 + \sum_{k=1}^L |\varepsilon_k| b_k^\dag b_k =: E_0 + \sum_{k=1}^L \lambda_k b_k^\dag b_k = E_0 + \sum_{k=1}^L \lambda_k n_k,
    \end{equation}
    A ground state is the quasiparticle vacuum state $\ket{\psi_0} = \ket{0}_b$ satisfying $b_k \ket{\psi_0} = 0$ for all $k$, and the ground state energy is $E_0 = \sum\limits_{k:\varepsilon_k < 0} \varepsilon_k$. Here we use $n_k$ to denote the quasiparticle number operator on the $k$-th canonical mode, i.e.
    \begin{equation}
        n_k:= b_k^\dag b_k = \begin{cases}
            \wt{c}_k^\dag \wt{c}_k, & \varepsilon_k > 0,\\
            1 - \wt{c}_k^\dag \wt{c}_k, & \varepsilon_k < 0.
        \end{cases}
    \end{equation}
\end{enumerate}
   Let us denote the dissipation generator 
   with respect to the jump operator $K$ as 
  \begin{equation}
     \mc D_{K}[\rho] := K \rho K^{\dagger} - \frac12\left\{ K^{\dagger} K , \rho \right\}.
  \end{equation}

\begin{prop}\label{prop:free_fermionic_lindblad}
 For the free-fermionic system defined above, with the coupling operators set $\{c_i^\dag,c_i\}_{i=1}^L$ and a filter satisfying \cref{eq:finite_resolution_filter}, the generator can be decomposed into a sum of independent generators acting on each canonical mode
    \begin{equation}
        \mc L[\rho] = \sum_{k=1}^L \mc L_k[\rho], \quad \mc L_k[\rho] = -\I \lambda_k [n_k, \rho] + \gamma_{c,k}^2   \mathcal{D}_{b_k}[\rho] +  \gamma_{h,k}^2  \mathcal{D}_{b_k^\dag}[\rho],
    \end{equation}
    where \begin{equation}
        \gamma_{c,k}:=\hat{f}(-\lambda_k),\quad \gamma_{h, k} : = \hat{f}(\lambda_k)
    \end{equation}
represent the cooling and heating strengths for the $k$-th mode. The steady state of the dynamics has the following tensor product structure:
\begin{equation}
    \rho_{\rm ss} = \bigotimes_{k: \lambda_k > \Delta} \ketbra{0}_k \bigotimes_{k: \lambda_k \le \Delta} \left( p_k \ketbra{0}_k + (1-p_k) \ketbra{1}_k \right),\quad p_k = \frac{\gamma_{c,k}^2 }{\gamma_{c,k}^2 +\gamma_{h,k}^2 }.
\end{equation} 
Moreover, the total energy converges to a value below $E_0 + \sum_{\lambda_k \le \Delta} \lambda_k$, and $\braket{n_k}$ converges to $0$ for all modes with $\lambda_k > \Delta$, both at the rate $e^{-t}$.
\end{prop}

\begin{proof}
In this canonical basis, the coupling operators $\{c_i^\dag,c_i\}_{i=1}^{L}$ can be expressed as
\begin{equation}
    c_i = \sum_{k=1}^L U_{ik} \wt{c}_k = \sum_{\varepsilon_k > 0} U_{ik} b_k + \sum_{\varepsilon_k < 0} U_{ik} b_k^\dag, \quad c_i^\dag = \sum_{k=1}^L U_{ik}^* \wt{c}_k^\dag = \sum_{\varepsilon_k > 0} U_{ik}^* b_k^\dag + \sum_{\varepsilon_k < 0} U_{ik}^* b_k.
\end{equation}
According to the Thouless theorem, the time evolution of the fermionic operators is given by~\cite{Thouless1960} 
\begin{equation}\label{eq:thouless_fermion}
    e^{\I Hs} b_k e^{-\I Hs} = e^{-\I \lambda_k s} b_k,\quad e^{\I Hs} b_k^\dag e^{-\I Hs} = e^{\I \lambda_k s} b_k^\dag,
\end{equation}
according to the construction in \cref{eq:jump_op}, the jump operators corresponding to the coupling operators are
\begin{equation}\label{eq:Ki+-}
    \begin{aligned}
        K_{i,+}&= \int_{-\infty}^{\infty} f(s) e^{\I Hs} c_i^\dag e^{-\I Hs} \dd s = \sum_{\varepsilon_k > 0}U_{ik}^* {\gamma_{h,k}} b_k^\dag + \sum_{\varepsilon_k < 0} U_{ik}^* {\gamma_{c,k}} b_k;\\
        K_{i,-}& = \int_{-\infty}^{\infty} f(s) e^{\I Hs} c_i e^{-\I Hs} \dd s = \sum_{\varepsilon_k > 0} U_{ik} {\gamma_{c,k}} b_k + \sum_{\varepsilon_k < 0} U_{ik} {\gamma_{h,k}} b_k^\dag,
    \end{aligned}
\end{equation}
where we denote the cooling and heating strength by
\begin{equation}
    \gamma_{c,k}:=\hat{f}(-\lambda_k),\quad \gamma_{h, k} : = \hat{f}(\lambda_k).
\end{equation}
 The {coherent} part is diagonal in the canonical basis, and the dissipative part simplifies by the orthonormality of the columns of $U$, which cancels all off-diagonal mode terms:
\begin{equation}
    \begin{aligned}
        \mc L_H[\rho] = & -\I [H, \rho] = -\I \sum_{k=1}^L \lambda_k [n_k, \rho],\\
      \sum_i  \mathcal{D}_{K_{i,+}}[\rho] 
      = & \sum_{\varepsilon_k<0} \gamma_{c,k}^2   \left(b_k \rho b_k^\dag -\frac12\{b_k^\dag b_k,\rho\}\right) + \sum_{\varepsilon_k>0} \gamma_{h,k}^2  \left( b_k^\dag \rho b_k -\frac12\{b_k b_k^\dag, \rho\} \right)\\
     = & \sum_{\varepsilon_k<0} \gamma_{c,k}^2  \mathcal{D}_{b_k}[\rho] + \sum_{\varepsilon_k>0} \gamma_{h,k}^2  \mathcal{D}_{b_k^\dag}[\rho],\\
    \sum_i \mathcal{D}_{K_{i,-}}[\rho] 
    = & \sum_{\varepsilon_k<0} \gamma_{h,k}^2  \mathcal{D}_{b_k^\dag}[\rho] + \sum_{\varepsilon_k>0} \gamma_{c,k}^2  \mathcal{D}_{b_k}[\rho].
    \end{aligned}
\end{equation}
Putting these together, the total Lindbladian can be expressed as
\begin{equation}\label{eq:lindblad_generator_free_fermion}
\mc L[\rho] = -\I \sum_{k=1}^L \lambda_k [n_k, \rho] + \sum_{k=1}^L   \gamma_{c,k}^2   \mathcal{D}_{b_k}[\rho] + \sum_{k=1}^L  \gamma_{h,k}^2    \mathcal{D}_{b_k^\dag}[\rho] =  \sum_{k=1}^L \mc L_k{[\rho]},
\end{equation}
where
\begin{equation}
    \mc L_k[\rho] := -\I \lambda_k [n_k, \rho] + \gamma_{c,k}^2   \mathcal{D}_{b_k}[\rho] +  \gamma_{h,k}^2  \mathcal{D}_{b_k^\dag}[\rho].
\end{equation}
Since each $\mc L_k$ acts only on mode $k$, the dynamics of any observable $O_k$ supported on mode $k$ is closed in the Heisenberg picture:
\begin{equation}
    \frac{\dInt O_k(t)}{\dInt t} =\mc L^\dag [O_k] = \mc L_k^{\dagger}[O_k] = \I \lambda_k [n_k, O_k] + \gamma_{c,k}^2   \mc D_{b_k}^{\dagger}[O_k] + \gamma_{h,k}^2  \mc D_{b_k^{\dagger}}^{\dagger}[O_k],
\end{equation}
where $\mc L_k^\dag$ is the Hilbert--Schmidt adjoint of the superoperator $\mc L_k$, and $\mc D_{K}^{\dagger}[O] $ is given by $K^{\dag} O K - \frac{1}{2}\left\{ K^{\dagger}K,O\right\}$.
Choosing $O_k = n_k$, we obtain
\begin{equation}\label{eq:Ddag_nk}
  \mc D^\dag_{b_k}[n_k] =  b_k^\dag n_k b_k -\frac12\{n_k,n_k\} \note{$n_k^2=n_k$}= -n_k,\quad \mc D^\dag_{b_k^\dag}[n_k] = b_k n_k b_k^\dag -\frac12\{b_k b_k^\dag, n_k\} = 1-n_k.
\end{equation}
Therefore 
\begin{equation}\label{eq:Ldag_nk}
   \mc L^\dag [n_k] =   \mc L_k^{\dag} [n_k] = \begin{cases}
        -n_k, & \lambda_k > \Delta;\\
        -n_k \left(\gamma_{c,k}^2   +  \gamma_{h,k}^2  \right) + \gamma_{h,k}^2 , &\lambda_k \le \Delta. \\
    \end{cases}
\end{equation}
If {$\gamma_{c,k}^2+\gamma_{h,k}^2>0$}, the stationary population for modes with excitation energies below $\Delta$ is
\begin{equation}
    \braket{n_k} = \frac{ \gamma_{h,k}^2  }{\gamma_{c,k}^2  + \gamma_{h,k}^2 }.
\end{equation}
For the same one-mode generator, each off-diagonal matrix element in the occupation basis decays at rate $(\gamma_{c,k}^2+\gamma_{h,k}^2)/2$ up to the Hamiltonian phase. Thus each $\mc L_k$ has a unique one-mode stationary state under the non-dark-mode condition above. Since the one-mode semigroups commute, the full steady state is the tensor product of these one-mode stationary states:
\begin{equation}\label{eq:steady_state_free_fermion}
    \rho_{\rm ss} = \bigotimes_{k: \lambda_k > \Delta} \ketbra{0}_k \bigotimes_{k: \lambda_k \le \Delta} \left( p_k \ketbra{0}_k + (1-p_k) \ketbra{1}_k \right),\quad p_k = \frac{\gamma_{c,k}^2 }{\gamma_{c,k}^2 +\gamma_{h,k}^2 }.
\end{equation}

    Finally, for modes with $\lambda_k > \Delta$, we have by \cref{eq:Ldag_nk} that 
    \begin{equation}
        \braket{n_k}_t = \Tr(n_k \rho(t)) \le e^{-t} \Tr(n_k \rho(0)),
    \end{equation}
   which converges to $0$ with an exponential rate $e^{-t}$. For the total energy
    \begin{equation}
        \Tr(H \rho(t)) = E_0 + \sum_{k}\lambda_k 
        \langle n_k \rangle_t
        \le E_0 + \sum_{k:\lambda_k \le \Delta}\lambda_k + \sum_{k:\lambda_k > \Delta}\lambda_k 
        \langle n_k \rangle_0
        e^{-t} , 
    \end{equation}
    which will converge to a value not exceeding $E_0 + \sum_{k:\lambda_k \le \Delta} \lambda_k$ with an exponential rate $e^{-t}$.

\end{proof}

With Proposition \ref{prop:free_fermionic_lindblad}, we
 are now ready to prove \cref{thm:free_fermion}. Indeed,  \cref{eq:Ldag_nk} is a drift inequality for the number operator $n_k$ on each canonical mode. Since the energy eigenstates in free fermionic systems are Fock states in the canonical basis, and since the eigenvalues are the occupation-weighted sums of the excitation energies, it is natural to construct the drift inequality for the resolved-mode
number operator $N_{\Delta} := \sum_{k: \lambda_k > \Delta} n_k$ that counts the total occupation on modes with excitation energies above $\Delta$.

  \begin{thm}[Rigorous 
    version of \cref{thm:free_fermion}] \label{thm:free_fermion_convergence} In the presence of finite-resolution imperfection \cref{eq:finite_resolution_filter} in the filter, for any initial state $\rho$, its high-energy population $\tr(\mc P_{>E_\Delta} \rho(t))$ decays with an exponential rate $e^{-t}$, where $E_\Delta : = E_0+\sum_{\abs{\varepsilon_k}\le\Delta} \abs{\varepsilon_k} =  E_0 +\sum_{k: \lambda_k \le \Delta} \lambda_k$. In particular, every stationary state $\rho_{\rm ss}$ satisfies $\Tr(\mc P_{>E_\Delta} \rho_{\rm ss}) = 0$.

    Moreover, whenever the normalized projection $\rho_{\le E_\Delta}(t)$ is defined, the state converges to its projection onto the low-energy subspace in trace distance exponentially fast.  Specifically,
$\norm{\rho(t)-\rho_{\le E_\Delta}(t)}_1
    \le
    2\sqrt{L}e^{-t/2}.$ 
\end{thm}
\begin{proof}
We first show the convergence of high-energy population.
    By \cref{eq:Ldag_nk}, we have 
    \begin{equation}
        \mc L^\dag[N_\Delta] = \sum_{k: \lambda_k > \Delta}   \mc L^\dag[n_k] = -\sum_{k: \lambda_k > \Delta}   n_k = -N_\Delta.
    \end{equation}
The canonical Fock states $\ket{\bv n} =  \prod_{k: n_k=1} b^\dagger_{k} \ket{0}_b$, with configurations $\bv n = (n_1, n_2, \cdots, n_L) \in \{0,1\}^L$, form an eigenbasis of $H$, and the corresponding eigenvalue is $E_{\bv n} = E_0+\sum_k \lambda_k n_k$. For any canonical Fock basis vector $\ket{\bv n}$ in the range of $\mc P_{>E}$ where $E\ge E_{\Delta}$, we have
 $E_{\bv n} > E_{\Delta}$ and thus there must be at least one mode $k$ such that $\lambda_k > \Delta$ and $n_k = 1$. Since both $\mc P_{>E}$ and $N_\Delta$ are diagonal in the canonical Fock basis, this gives the operator inequality
\begin{equation}\label{eq:markov_inequality_N}
     \mc P_{>E}\le N_\Delta,
    \quad \forall E\ge E_\Delta.
\end{equation}
By 
Lemma~\ref{lem:drift}, we have
\begin{equation}\label{eq:high_energy_convergence}
    \Tr(\mc P_{>E}\rho(t))
    \le
    \Tr(N_\Delta\rho(t))
    \le 
    e^{-t}\Tr(N_\Delta\rho(0)) \le Le^{-t},\quad \forall E\ge E_\Delta.
\end{equation}
Hence $\Tr[\mc P_{>E_\Delta}\rho_{\rm ss}] = 0$ for every stationary state $\rho_{\rm ss}$, which means that stationary states have no support in the high-energy subspace above $E_\Delta$.

Next we show convergence of the state to its low-energy projection in trace distance. We denote
$\mc P_{\le E_\Delta} = \mathbf 1_{(-\infty, E_\Delta]}(H) = 1-\mc P_{>E_\Delta}$ as the projection onto the low-energy subspace. 
Whenever $p_\Delta(t):=\Tr(\mc P_{\le E_\Delta}\rho(t))>0$, we define the projection of $\rho(t)$ onto the low-energy subspace as
\begin{equation}
    \rho_{\le E_\Delta}(t)
    :=
    \frac{\mc P_{\le E_\Delta}\rho(t)\mc P_{\le E_\Delta}}
    {p_\Delta(t)}.
\end{equation}
Since $
    \sqrt{\rho(t)} \rho_{\le E_\Delta}(t) \sqrt{\rho(t)}
    = \left(\sqrt{\rho(t)} \mc P_{\le E_\Delta} \sqrt{\rho(t)}\right)^2 / p_{\Delta}(t)$,  
we get that the fidelity $
    F(\rho(t),\rho_{\le E_\Delta}(t)) =
    \Tr\sqrt{
        \sqrt{\rho(t)} \rho_{\le E_\Delta}(t) \sqrt{\rho(t)} }
    = \sqrt{p_{\Delta}(t)}. $ 
By the Fuchs--van de Graaf inequality~\cite{FuchsVanDeGraaf2002} and \cref{eq:high_energy_convergence}, we have
\begin{equation}
    \begin{aligned}
        \norm{\rho(t)- \rho_{\le E_\Delta}(t) }_1 & \le 2\sqrt{1-F(\rho(t),\rho_{\le E_\Delta}(t))^2}
    = 2\sqrt{1-p_{\Delta}(t)} = 2 \sqrt{ \tr \left[\mc P_{> E_\Delta}  \rho(t) \right]} \le 2 \sqrt{L} e^{-t/2} .
    \end{aligned}
\end{equation}

\end{proof}

\subsection{Free bosonic systems}
\label{appendix:boson}

We begin with the single bosonic mode.
The Hamiltonian and the coupling operator are
\begin{equation}
    H = \Omega a^\dag a =: \Omega N~({\Omega>0} ),\quad A = a + a^\dag.
    \label{eq:hamil_coup_boson}
\end{equation}
Here 
\begin{equation}
N:= a^\dag a
\end{equation}
is the number operator. Using free-bosonic version of the Thouless theorem \cref{eq:thouless_fermion}, we have
\begin{equation}
    e^{\I Hs} a e^{-\I Hs} = e^{-\I \Omega s} a, \quad e^{\I Hs} a^\dag e^{-\I Hs} = e^{\I \Omega s} a^\dag.
\end{equation}
Therefore, the jump operator is constructed as
\begin{equation}\label{eq:jump_boson}
    K = \int_{-\infty}^{\infty} f(s) e^{\I Hs} A e^{-\I Hs} ds = \hat{f}(-\Omega) a + \hat{f}(\Omega) a^\dag.
\end{equation}
We assume a realistic filter with finite resolution $\Delta > \Omega$, i.e., $0 \le \hat{f}(\Omega) < \hat{f}(-\Omega) \le 1 $. Below, we denote the cooling and heating strength by
\begin{equation}
    \gamma_c:=\hat{f}(-\Omega),\quad \gamma_h : = \hat{f}(\Omega).
\end{equation}

Before moving on to prove \cref{thm:boson-drift}, let us first analyze the steady state properties. Since the Hamiltonian is quadratic and the jump operator is linear, the Lindblad dynamics is quasi-free, so the steady state is a bosonic Gaussian state, which is completely characterized by the moments $\langle a\rangle$, $\langle a^2\rangle$, and $\langle N\rangle = \langle a^\dag a\rangle$. These expectation values will indeed converge exponentially fast to their steady values. The relevant calculations are collected in the following proposition.
\begin{prop}
The quasi-free dynamics associated with \cref{eq:hamil_coup_boson} has a bosonic Gaussian stationary state with $\langle N\rangle_{\rm ss} = \frac{\gamma_h^2}{\gamma_c^2-\gamma_h^2}$, $\langle a^2\rangle_{\rm ss} = \frac{\gamma_c \gamma_h}{\gamma_h^2 - \gamma_c^2 - 2\I \Omega}$, and $\langle a\rangle_{\rm ss} = 0$. The convergence of these moments is exponential.
\end{prop}
These properties have been extensively studied in previous works~\cite{Prosen2010,BarthelZhang2022}. For completeness, we also include the derivation here.
\begin{proof}
Similar to the fermionic case, we compute the derivatives of these observables in the Heisenberg picture. Let us denote $\braket{O}_t := \braket{O}(t) {= \Tr(\rho(t ) O)}$ for an observable $O$.

 \paragraph{The number operator $N$.}
Using the commutation relations $[N,a]=-a$, $[N,a^\dag]=a^\dag$,
we obtain that
\begin{equation}
    \begin{aligned}
    \mc D_{K}^\dag[N]
    =\frac12 K^\dag[N,K]+\frac12[K^\dag,N]K%
    =(\gamma_h^2-\gamma_c^2)N+\gamma_h^2.
    \end{aligned}
\end{equation}
Take expectation in state $\rho$, and we obtain
\begin{equation}
 \dv{\langle N\rangle}{t} = \Tr\!\left(\rho {\mc L}^{\dag}[N]\right)
 =(\gamma_h^2-\gamma_c^2)\langle N\rangle+\gamma_h^2.
\end{equation}
Its solution is
\begin{equation}\label{eq:N_boson}
   \langle N\rangle_t = \frac{\gamma_h^2}{\gamma_c^2-\gamma_h^2} + \left(\langle N\rangle_0 - \frac{\gamma_h^2}{\gamma_c^2-\gamma_h^2}\right)e^{-(\gamma_c^2-\gamma_h^2)t},\quad \langle N\rangle_{\rm ss} = \frac{\gamma_h^2}{\gamma_c^2-\gamma_h^2}.
\end{equation}

   \paragraph{The two-photon coherence $a^2$.} The coherent contribution to the Heisenberg picture derivative is
    \begin{equation}
        \I[H, a^2] = \I \Omega[a^\dag a, a^2] = \I \Omega (-2a^2) = -2\I \Omega a^2.
    \end{equation}
    For the dissipation, using the commutators $[a^2, K] = 2\gamma_h a$ and $[K^\dag, a^2] = -2\gamma_c a$, we find
    \begin{equation}
        \begin{aligned}
        \mc D_K^\dag[a^2] &= \frac12 K^\dag[a^2, K] + \frac12 [K^\dag, a^2] K 
        = (\gamma_h^2 - \gamma_c^2)a^2 - \gamma_c \gamma_h.
        \end{aligned}
    \end{equation}
    Therefore
    \begin{equation}\label{eq:b2_boson}
        \dv{\langle a^2\rangle}{t} = \Tr(\rho \mc L^\dag[a^2]) = (-2\I \Omega + \gamma_h^2 - \gamma_c^2)\langle a^2 \rangle - \gamma_c \gamma_h.
    \end{equation}
    The solution is 
    \begin{equation}\label{eq:a2_boson_ss}
        \langle a^2\rangle_t = \langle a^2\rangle_{\rm ss} + \left(\braket{a^2}_0 - \langle a^2\rangle_{\rm ss}\right) e^{-\left(\gamma_c^2-\gamma_h^2 +2\I \Omega\right)t}, \quad \langle a^2\rangle_{\rm ss} = \frac{\gamma_c \gamma_h}{\gamma_h^2 - \gamma_c^2 - 2\I \Omega}.
    \end{equation}

    \paragraph{The lowering operator $a$.}
    The coherent contribution is
    \begin{equation}
        \I [H, a] = \I [\Omega a^\dag a,a]   = \I \Omega ( a^\dag a a - aa^\dag a ) = -\I \Omega a,
    \end{equation}
    and for the dissipative part
    \begin{equation} 
        \begin{aligned}
        \mc D^\dag_K [a] 
        & =  \frac12 K^\dag [a, K] + \frac12 [K^{\dag}, a] K 
        = \frac{\gamma_h^2-\gamma_c^2}{2}a.
        \end{aligned}
    \end{equation}
    Therefore
    \begin{equation}\label{eq:b_boson}
        \begin{aligned}
        & \dv{\langle a\rangle}{t} = \Tr(\rho \mc L^\dag[a]) =  \left(-\I \Omega + \frac{\gamma_h^2-\gamma_c^2}{2}\right) \langle a\rangle 
        \end{aligned}
\end{equation}
which gives
\begin{equation}\label{eq:a_boson_ss}
        \langle a\rangle_t = \langle a\rangle_0 e^{- \left(\I  \Omega + \frac{\gamma_c^2 - \gamma_h^2}{2}\right)t},\quad \langle a\rangle_{\rm ss} = 0.
    \end{equation}
    
\end{proof}

We discuss the steady state properties of the bosonic system with realistic filters. Similar to the free fermionic case, this can also be done by constructing a suitable drift operator
and applying the argument of Lemma \ref{lem:drift} to obtain the exponential tail bound of the steady state population distribution.

\begin{thm}[Rigorous version of \cref{thm:boson_population}]\label{thm:boson-drift} Let $s: = {\gamma_h} / {\gamma_c}  = {\hat{f}(\Omega)} / {\hat{f}(-\Omega)} < 1$, and choose $\beta$ such that
\begin{equation}\label{eq:boson-beta-range}
    0<\beta< \frac{1}{\Omega} \log\frac1s,
    \qquad\text{equivalently}\qquad  s< e^{-\beta \Omega}.
\end{equation}
Then there exists a constant $C_{\beta} > 0$, such that for any $E > 0$,
\begin{equation}
    \Tr(\mc P_{>E} \rho_{\rm ss}) \le C_{\beta} e^{-{\beta E}}.
\end{equation}
\end{thm}

\begin{proof} We choose the drift operator to be $\mathsf F_\beta := e^{\beta \Omega N}$, and denote $r_{\beta}:=e^{\beta \Omega}$. 
The diagonal matrix elements of $\mc L^\dagger[\mathsf F_\beta]$ in the Fock basis are
\begin{equation}\label{eq:boson-diag-LF}
    \langle n|\mc L^\dagger[\mathsf F_\beta]|n\rangle
    =
    r_\beta^n(r_\beta-1)\left(\gamma_h^2(n+1)-\frac{\gamma_c^2}{r_\beta}n\right).
\end{equation}
The only off-diagonal matrix elements are
\begin{equation}\label{eq:boson-offdiag-LF}
    \langle n+2|\mc L^\dagger[\mathsf F_\beta]|n\rangle
    =
    \langle n|\mc L^\dagger[\mathsf F_\beta]|n+2\rangle
    =
    -\frac{\gamma_c\gamma_h}{2}(r_\beta-1)^2r_\beta^n\sqrt{(n+1)(n+2)}.
\end{equation}
\cref{eq:boson-diag-LF,eq:boson-offdiag-LF} follow from inserting
the action of the jump operator $K$ on the Fock states
\begin{equation}
    K\ket n=\gamma_c\sqrt n\ket{n-1}+\gamma_h\sqrt{n+1}\ket{n+1}.
\end{equation}
Define the normalized drift
\begin{equation}
    \wt{\mathsf F_\beta} :=\mathsf F_\beta^{-1/2}\mc L^\dagger[\mathsf F_\beta]\mathsf F_\beta^{-1/2}.
\end{equation}
By $\mathsf F_\beta\ket{n} = {r_\beta^n}\ket{n}$, the diagonal and off-diagonal elements of $\wt{\mathsf F_\beta}$ are given by \cref{eq:boson-B-diag,eq:boson-B-offdiag}:
\begin{equation}\label{eq:boson-B-diag}
    (\wt{\mathsf F_\beta})_{n,n}
    =
    (r_\beta-1)\left(\gamma_h^2(n+1)-\frac{\gamma_c^2}{r_\beta}n\right),
\end{equation}
\begin{equation}\label{eq:boson-B-offdiag}
    (\wt{\mathsf F_\beta})_{n+2,n}=(\wt{\mathsf F_\beta})_{n,n+2}
    =
    -\frac{\gamma_c\gamma_h}{2r_\beta}(r_\beta-1)^2\sqrt{(n+1)(n+2)}.
\end{equation}
Set
\begin{equation}
    c_n:=\frac{\gamma_c\gamma_h}{2r_\beta}(r_\beta-1)^2\sqrt{(n+1)(n+2)} \ge 0, \quad 
    c_{-2}=c_{-1}=0.
\end{equation}
For any $\ket{\varphi}=\sum_n\varphi_n\ket n\in \ell^2(\mathbb{N})$, the inequality $2|uv|\le |u|^2+|v|^2$ gives
\begin{equation}\label{eq:boson-quadratic-form}
    \langle \varphi|\wt{\mathsf F_\beta}|\varphi\rangle
    \le
    \sum_{n\ge0}A_n|\varphi_n|^2,
\end{equation}
\begin{equation}\label{eq:boson-An}
    \begin{aligned}
    A_n&=(\wt{\mathsf F_\beta})_{n,n}+c_n + c_{n-2} \le \left[(r_\beta-1)\left(\gamma_h^2 -\frac{\gamma_c^2}{r_\beta}\right) + \frac{\gamma_c\gamma_h}{r_\beta}(r_\beta-1)^2  \right] n + (r_\beta-1) \gamma_h^2 + \frac{\gamma_c\gamma_h}{r_\beta}(r_\beta-1)^2\\
    & = \frac{\gamma_c^2(r_\beta-1)}{r_\beta} (s+1) (r_\beta s-1) n + (r_\beta-1) \gamma_h^2 + \frac{\gamma_c\gamma_h}{r_\beta}(r_\beta-1)^2 =: q_\beta n + d_\beta.
    \end{aligned}
\end{equation}
By \cref{eq:boson-beta-range}, $q_\beta<0$ and $d_\beta >0$. For any fixed $\kappa >0$, there exists $M = \lceil -(d_\beta+\kappa    )/ q_\beta \rceil$ such that $A_n\le -\kappa$ for all $n\ge M$. Therefore, if we define  
\begin{equation}
    C:= \max\{0, \max_{0\le n<M} (A_n+\kappa)\} \le d_\beta + \kappa<\infty,
\end{equation}
then for every $n$, we have by definition that $A_n\le -\kappa + C \mathbf{1}_{n<M}$. Therefore, by the quadratic-form inequality \cref{eq:boson-quadratic-form}, 
\begin{equation}
    \langle \varphi|\wt{\mathsf F_\beta}|\varphi\rangle\le \sum_{n\ge0}(-\kappa + C \mathbf{1}_{n<M})|\varphi_n|^2 = -\kappa \norm{\ket{\varphi}}^2  + C\norm{\mc P_{n<M}\ket{\varphi}}^2 = \langle \varphi|(-\kappa I + C \mc P_{n<M})|\varphi\rangle.
\end{equation}
where $\norm{\cdot}$ denotes the $\ell^2$ norm and $\mc P_{n<M}$ is the projection onto the span of $\ket{0},\ket{1},\ldots,\ket{M-1}$. Therefore we have the operator inequality on the space of $\ell^2$-integrable vectors:
\begin{equation}\label{eq:Bbeta-source}
    \wt{\mathsf F_\beta}\le -\kappa I + C \mc P_{n<M}.
\end{equation}
Note that $\mc P_{n<M} = \mathbf{1}_{n<M}(N)$ and $\mathsf F_\beta = r_\beta^N$, thus we have 
\begin{equation}
    \mathsf F_\beta^{\frac12}\mc P_{n<M}\mathsf F_\beta^{\frac12} = r_\beta^{N} \mathbf{1}_{n<M}(N) = r_\beta^{N} \mc P_{n<M} \le r_\beta^{M-1} \mc P_{n<M} \le r_\beta^{M-1} I.
\end{equation}
Thus conjugating \cref{eq:Bbeta-source} by $\mathsf F_\beta^{1/2}$ gives
\begin{equation}
    \mc L^\dagger[\mathsf F_\beta]
    \le
    -\kappa \mathsf F_\beta+C r_\beta^{M-1}I.
\end{equation}
Thus the drift inequality holds with $\theta:=C r_\beta^{M-1}$.  The stationary bound follows from Lemma \ref{lem:drift}, which gives
\begin{equation}
    \Tr(\mathsf F_\beta \rho_{ss}) \le \frac{\theta }{\kappa } \le \frac{d_\beta + \kappa}{\kappa} r_\beta^{-\frac{d_\beta+\kappa}{ q_\beta}}.
\end{equation}
By choosing e.g. $\kappa = -\frac{q_\beta}{2}$, we get $\Tr(\mathsf F_\beta \rho_{ss}) \le \left(1-\frac{2d_\beta}{q_\beta}\right) r_\beta^{\frac12 -\frac{d_\beta}{q_\beta}}$. Note that $\lim_{\beta\downarrow 0 } \left(-\frac{d_\beta}{q_\beta}\right) = \frac{s^2}{1-s^2}$, we have that the upper bound is finite for any $0\le \beta < \frac{1}{\Omega}\log \frac1s$.

We can then obtain the tail bound for the steady state using Markov's inequality. For any $E\ge 0$, we have by setting $\mathfrak g_k = e^{\beta\lambda_k}$ and noting that $\mathfrak g_E \ge \inf_{\lambda>E} e^{\beta\lambda} = e^{\beta E}$ in Lemma~\ref{lem:drift} that
\begin{equation}
    \Tr(\mc P_{>E}\rho_{ss}) \le e^{-{\beta E}}\Tr(e^{\beta \Omega N}\rho_{ss}) \le e^{- {\beta E}} \frac{\theta}{\kappa}.
\end{equation}

\end{proof}

\begin{rem}\label{rem:squeezed-threshold}
In the purely dissipative case i.e. $\mathcal L = \mathcal D$ and $\mathcal L_H = 0$, the 
steady state is actually analytically solvable and is given by 
a squeezed vacuum state $
    \ket{\Psi_s}
    =(1-s^2)^{1/4}\sum_{k\ge0}(-s)^k\frac{\sqrt{(2k)!}}{2^k k!}\ket{2k}.$ 
The physical occupation is $m=2k$, and the populations behave as
\begin{equation}
    |\langle 2k|\Psi_s\rangle|^2
    =(1-s^2)^{1/2}s^{2k}\frac{(2k)!}{4^k(k!)^2}
    \sim s^m k^{-1/2}.
\end{equation}
Therefore the exact squeezed-state moment $\langle e^{\eta N}\rangle$ is finite precisely for $e^\eta s<1$, i.e.
$
    \eta<\log\frac1s.
$ 
Taking $\eta=\beta\Omega$ gives the same admissible range as \cref{eq:boson-beta-range}. The Lyapunov proof above could be viewed as the dynamical approach of the same exponential threshold, which is obtained without any prior knowledge of the exact steady state.
\end{rem}

\subsection{Additional results on free bosonic systems}\label{app:additional_boson}
Since the steady state is a Gaussian state, we could also understand the free bosonic systems by directly solving the steady state or the covariance matrix. In this section, we present some additional results about the steady state of the free bosonic system such as the trace-distance convergence and the generalization to the multimode case, which may be of independent interest.
\subsubsection{Purely dissipative case}
We first consider the purely dissipative case
i.e., $\mc L_H =0$ and $\mc L =\mc D$. In this case, the steady state is a unique pure state in $\ker K$ and can be characterized by a squeezed vacuum state~\cite{MarianMarian1993,GrynbergAspectFabre2010}
\begin{equation}
\ket{\Psi_\zeta} :=  S(\zeta)\ket{0}_a=\exp\!\left(\frac12 (\zeta^\ast aa  - \zeta a^\dag a^\dag)\right) \ket{0}_a,
\end{equation}
where the squeeze parameter $\zeta$ depends only on the leakage ratio:
\begin{equation}
    \zeta = \arctanh s, \quad s := \gamma_h / \gamma_c {<1}.
\end{equation} 
The Fock basis expansion of $\ket{\Psi_\zeta}$ is
\begin{equation}\label{eq:squeezed_vacuum_state_alpha_zero}
   \ket{\Psi_\zeta}
   = (1-s^2)^{1/4} \sum_{n=0}^{\infty} (-s)^n \frac{\sqrt{(2n)!}}{2^n n!} \ket{2n}.
\end{equation}
By Stirling's formula
$
\frac{\sqrt{(2n)!}}{2^n n!} \sim n^{-1/4}$, the coefficients satisfy
\begin{equation}
|\langle 2n|\Psi_\zeta\rangle | \sim |s|^n n^{-1/4},\quad \langle 2n+1|\Psi_\zeta\rangle = 0.
\end{equation}
Since $0\le s<1$, the Fock-state populations decay exponentially in $n$, up to a subleading algebraic prefactor.

It is useful for numerically simulating the profile of the steady-state photon number distribution. In particular, when $s = 0$, $\ket{\Psi_{\zeta=0}} = \ket{0}_a$ is the vacuum state. We derive \cref{eq:squeezed_vacuum_state_alpha_zero} and verify that $K  \ket{\Psi_\zeta} = 0$ in Appendix \ref{appendix:BosonPopulation}.

Since the steady state is a pure state, we can utilize the Fuchs--van de Graaf argument to further show the exponential convergence of the state $\rho(t)$ to the steady state $\ket{\Psi }$ in trace distance. Specifically, we have   
\begin{prop}[Rigorous version of Proposition~\ref{prop:free_boson}, 
with $\mc L_H = 0$]
\label{prop:boson_convergence}
If the initial state $\rho(0)$ satisfies that $\Tr(K^\dag K\rho(0))$ is well-defined and finite, then the trace distance between $\rho(t)$ and the steady state $\rho_{\rm ss}$ decays exponentially with a rate $e^{-(\gamma_c^2-\gamma_h^2)t/2}$ and an initial-state-dependent prefactor. Here $\rho_{\rm ss}$ is the squeezed vacuum state \cref{eq:squeezed_vacuum_state_alpha_zero}.
\end{prop}
\begin{proof}
    We define the ``squeezed number operator'' as follows
\begin{equation}
    \wt a = \frac1{\sqrt{1-s^2}}(a + sa^\dag),\quad \wt N = \wt a^\dag \wt a = \frac1{\gamma_c^2-\gamma_h^2}K^\dag K.
\end{equation}
Thus we have $\wt N \ket{\Psi_\zeta} = 0$ by $\ket{\Psi_\zeta} \in \ker K$
. In fact, we can view the operator $\wt a$ as a Bogoliubov-transformed bosonic annihilation operator, and we can verify that $\wt a$ and $\wt a^\dag$ satisfy the canonical commutation relation (CCR) $[\wt a, \wt a^\dag] = 1$ and $[\wt a, (\wt a^\dag)^n] = n (\wt a^\dag)^{n-1}$. 
Therefore, we have 
\begin{equation}
    \begin{aligned}
    \wt N (\wt a^\dag)^{n} \ket{\Psi_\zeta} & = \wt a^\dag \wt a (\wt a^\dag)^{n} \ket{\Psi_\zeta} %
    = n (\wt a^\dag)^{n} \ket{\Psi_\zeta}.
    \end{aligned}
\end{equation}
Therefore we can define the orthonormal eigenbasis of $\wt N$ as $\ket{n}_{\wt a} = \frac1{\sqrt{n!}} (\wt a^\dag)^n \ket{\Psi }$, which is also a Fock basis for the bosonic mode. 
The steady state $\rho_{\rm ss} = \dyad{\Psi } = \dyad{0}_{\wt a}$ is the vacuum in this basis.
Thus, we can compute the dynamics of $\wt N$ as follows:
\begin{equation}
    \mc D_K^\dag [\wt N] = (\gamma_c^2-\gamma_h^2) \mc D^\dag _{\wt a}[\wt a^\dag \wt a] = -(\gamma_c^2-\gamma_h^2)\wt N\implies \Tr(\rho(t) \wt N) = e^{-(\gamma_c^2-\gamma_h^2)t} \Tr(\rho(0) \wt N).
\end{equation}
Let us write both $\rho_{\rm ss}$ and $\wt N$ in the squeezed Fock basis $\{|n\rangle_{\wt a}\}$:
\begin{equation}
    I - \rho_{\rm ss} = I  -\dyad{0}_{\wt a}  = \sum_{n=1}^\infty |n\rangle_{\wt a} \langle n|_{\wt a},\quad \wt{N} = \sum_{n=1}^\infty n |n\rangle_{\wt a} \langle n|_{\wt a},
\end{equation}
which implies that $I - \rho_{\rm ss} \le \wt{N}$. Therefore, by the Fuchs--van de Graaf inequality \cite{FuchsVanDeGraaf2002}, we have
\begin{equation}
    \begin{aligned}
    d(\rho(t), \rho_{\rm ss})%
    \le \sqrt{1 - \Tr(\dyad{0}_{\wt a} \rho(t))} %
    \note{$1-\rho_{\rm ss}\le \wt N$}\le \sqrt{\Tr(\wt{N} \rho(t))} = e^{-(\gamma_c^2-\gamma_h^2)t/2} \sqrt{\Tr(\wt{N} \rho(0))}.
    \end{aligned}
\end{equation}
Because $\Tr(\wt  N\rho (0)) = \frac1{\gamma_c^2-\gamma_h^2} \Tr(K^\dag K \rho(0))$ is well-defined and finite by the assumption, we get the required exponential convergence in the trace distance.
\end{proof}

The analysis above can be easily extended to multi-mode free bosonic systems. The Hamiltonian and the coupling operators are given by
\begin{equation}
    H = \sum_{i,j=1}^L h_{ij} a_i^\dag a_j,\quad A_i = a_i + a_i^\dag, \quad i=1,2,\cdots, L,
\end{equation}
where $h$ is a positive definite real symmetric matrix and we choose a real orthogonal diagonalization $h = U \text{diag}(\varepsilon_1, \cdots, \varepsilon_L)U^\dag$. We can define the canonical modes $\wt a_i = \sum_{j=1}^L U_{ji} a_j$ such that $H = \sum_{i=1}^L \varepsilon_i \wt a_i^\dag \wt a_i$. The bosonic operators in the canonical mode basis still satisfy the CCR $[\wt a_i, \wt a_j^\dag] = \delta_{ij}$ and $[\wt a_i, \wt a_j] = 0$. We have a similar decomposition of the Lindbladian as in the free fermion case Proposition \ref{prop:free_fermionic_lindblad}.
\begin{prop}\label{prop:free_bosonic_lindblad}
    For the multi-mode free bosonic system defined above, with the coupling operator set $\{a_i + a_i^\dag\}_{i=1}^L$, assume $\gamma_{c,k}>\gamma_{h,k}\ge 0$ for all $k$. For the purely dissipative dynamics generated by the filtered jumps, the Lindbladian can be decomposed as
    \begin{equation}
    \mc L = \mc D = \sum_{k=1}^L \mc D_{J_k },\quad J_k = \hat{f}(-\varepsilon_k) \wt a_k + \hat{f}(\varepsilon_k) \wt a_k^\dag.
    \end{equation}
    The steady state is a tensor product of single-mode squeezed vacuum states $\rho_{\rm ss} = \bigotimes_{i=1}^L \dyad{\Psi_i}$, where each $\ket{\Psi_i}$ is determined by the corresponding filter leakage $s_i = \hat{f}(\varepsilon_i) / \hat{f}(-\varepsilon_i)$ by \cref{eq:squeezed_vacuum_state_alpha_zero}.
\end{prop}
\begin{proof}

Using the Thouless theorem $e^{\I Hs} \wt a_i e^{-\I Hs} = e^{-\I \varepsilon_i s} \wt a_i$ and expanding the coupling operators in the canonical mode basis, we can construct the jump operators as
\begin{equation}
    K_i = \int_{-\infty}^\infty f(s) e^{\I Hs} A_i e^{-\I Hs} ds = \sum_{k=1}^L U_{ik} \hat{f}(-\varepsilon_k) \wt a_k + \sum_{k=1}^L U_{ik} \hat{f}(\varepsilon_k) \wt a_k^\dag =: \sum_{k=1}^L U_{ik} J_k,
\end{equation}
where we denote 
\begin{equation}
    J_k =  \gamma_{c,k} \wt a_k + \gamma_{h,k} \wt a_k^\dag,\quad \gamma_{c,k} := \hat{f}(-\varepsilon_k),\quad \gamma_{h,k} := \hat{f}(\varepsilon_k).
\end{equation}
The dissipator is given by $\mc D = \sum_{i=1}^L \mc D_{K_i}$, which can be simplified using the orthogonality of $U$ as
\begin{equation}
    \begin{aligned}
          \mc D [\rho] = \sum_{k=1}^L J_k \rho J_k^\dag - \frac12 \{J_k^\dag J_k , \rho\} =\sum_{k=1}^L \mc D_{J_k}[\rho] .
    \end{aligned}
\end{equation}
We then look at the dynamics of the operators $\wt a_i$, $\wt a_i^\dag \wt a_i$ and $\wt a_i^2$ on each canonical mode. This is equivalent to the single-mode case, and we have the following identities for the dissipator using the CCR (see also \cref{eq:N_boson} \cref{eq:b2_boson} and  \cref{eq:b_boson}):
\begin{equation}
    \mc D_{J_i}^\dag [\wt a_i] =\frac{\gamma_{h,i}^2-\gamma_{c,i}^2}{2} \wt a_i, \quad     \mc D_{J_i}^\dag [\wt a_i^\dag \wt a_i] = (\gamma_{h,i}^2-\gamma_{c,i}^2) \wt a_i^\dag \wt a_i + \gamma_{h,i}^2, \quad  \mc D_{J_i}^\dag [\wt a_i^2] = (\gamma_{h,i}^2-\gamma_{c,i}^2) \wt a_i^2 - \gamma_{c,i} \gamma_{h,i},
\end{equation}
Also, similar to \cref{eq:steady_state_free_fermion} in the multi-site free fermion case, the steady state of this purely dissipative generator is a tensor product of single-mode squeezed vacuum states, each of which is determined by the corresponding filter leakage $s_i:=\gamma_{h,i}/\gamma_{c,i}$. Specifically,
we define $\ket{\Psi_i} $ as in \cref{eq:squeezed_vacuum_state_alpha_zero} for the $i$-th mode corresponding to the $i$-th leakage ratio $s_i : = \gamma_{h,i} / \gamma_{c,i} = \wh f(\varepsilon_i)/\wh f(-\varepsilon_i)$. Then {by the single-mode construction above} in  Proposition~\ref{prop:boson_convergence}, 
   \begin{equation}
   J_i \ket{\Psi_i} = (\gamma_{c,i}\wt a_i +\gamma_{h,i}\wt a_i^\dag)\ket{\Psi_i}= 0.
   \end{equation}  Let {$\rho_i := \ketbra{\Psi_i}$}, and {$\rho_{\rm ss} = \bigotimes_{i=1}^L \rho_i$} is the steady state of the dynamics.
   
\end{proof}

Using the Fuchs--van de Graaf inequality, we can further show the exponential convergence of $\rho(t)$ to $\rho_{\rm ss}$ in trace distance.
\begin{thm}\label{thm:boson_convergence_multisite}
    Assume $\gamma:=\min_i \left[ \wh f(-\varepsilon_i)^2 -  \wh f(\varepsilon_i)^2 \right]>0$ and finite initial squeezed occupation $\sum_i\Tr(\wt N_i\rho(0))<\infty$, where $\wt N_i$ is defined in the proof below. The trace distance between $\rho(t)$ and the steady state $\rho_{\rm ss}$ decays exponentially with a rate $e^{-\gamma t/2}$, where $\gamma = \min_i \left[ \wh f(-\varepsilon_i)^2 -  \wh f(\varepsilon_i)^2 \right]$.
\end{thm}
\begin{proof}
    We can define the ``squeezed number operator'' for each mode as $\wt N_i = \frac1{\gamma_{c,i}^2 - \gamma_{h,i}^2} J_i^\dag J_i$, which satisfies 
   \begin{equation}
    \mc D^\dag [\wt N_i]  = -(\gamma_{c,i}^2-\gamma_{h,i}^2) \wt N_i,\quad \wt N_i\ket{\Psi_i} = 0.
   \end{equation}
     Similar to Proposition~\ref{prop:boson_convergence}, we have {$I- \rho_i \le \wt N_i$} and thus 
     \begin{equation}
           I-\rho_{\rm ss} \le \sum_{i=1}^L  (I-\rho_i) \le \sum_{i=1}^L \wt N_i 
     \end{equation}
    by the operator inequality $1-\bigotimes_{i=1}^L P_i \le \sum_{i=1}^L (1-P_i)$ for $P_i$ being orthogonal projectors. Therefore, by the Fuchs--van de Graaf inequality \cite{FuchsVanDeGraaf2002}, we have
\begin{equation}
    \begin{aligned}
    d(\rho(t), \rho_{\rm ss}) & \le \sqrt{1 - \Tr(\rho_{\rm ss} \rho(t))} = \sqrt{\Tr((I-\rho_{\rm ss})\rho(t))} \le \sqrt{\sum_{i=1}^L \Tr(\wt N_i \rho(t))} \le    e^{-\gamma t/2} \sqrt{\sum_{i=1}^L \Tr(\wt N_i \, \rho(0))},
    \end{aligned}
\end{equation}
where $\gamma = \min_i (\gamma_{c,i}^2 - \gamma_{h,i}^2)$.  
\end{proof}

\subsubsection{Free bosonic systems with both Hamiltonian and dissipative contributions}

    The fact that the equations of motion of the first and second moments ($a$, $a^\dag a$, $a^2$) form a closed system is characteristic of quasi-free Lindblad dynamics in general. Under such dynamics, the linear jump operators preserve the Gaussianity of the state. If the initial state of the dissipative protocol is assumed to also be a Gaussian state, we can carry out a unified analysis for both the purely dissipative case and the case with Hamiltonian contribution by the recently developed techniques for analyzing the tight upper bound of the trace distance between two Gaussian states, for both bosonic and fermionic systems \cite{BittelMeleTironeLami2025,BittelMeleEisertLeone2025}.

\begin{thm}[\protect{\cite[Theorem 1]{BittelMeleTironeLami2025}}] \label{thm:trace_distance_gaussian} Let $\rho(V,\bv m)$ and $\rho(W,\bv t)$ be two bosonic Gaussian states with covariance matrices $V$ and $W$ and first-moment vectors $\bv m$ and $\bv t$. Then the trace distance between $\rho(V,\bv m)$ and $\rho(W,\bv t)$ is upper bounded by
    \begin{equation}
        d(\rho(V,\bv m), \rho(W,\bv t)) \le \frac{\sqrt 3+1}{16} \Tr(\abs{V-W} \Xi ^T (V+W) \Xi) + \sqrt{\frac{\min(\norm{V}  , \norm{W} )}{2}} \norm{\bv m - \bv t},
     \end{equation}
    where $\norm{\cdot}$ denotes the spectral norm, the matrix absolute value is defined as $\abs{A} = \sqrt{A^\dag A}$, and the first-moment vector $\bv m$ and the positive semidefinite covariance matrix $V$ are respectively defined as
    \begin{equation}
           \mathbf{m} = \frac{1}{\sqrt{2}} \begin{pmatrix} \langle a \rangle + \langle a^\dagger \rangle \\-\I(\langle a \rangle - \langle a^\dagger \rangle) \end{pmatrix} = \sqrt2 \mqty(\Re \langle a\rangle \\ \Im \langle a\rangle) ,\quad \Xi = \begin{pmatrix} 0 & 1 \\ -1 & 0 \end{pmatrix}
    \end{equation}
    and 
    \begin{equation}
            V = \begin{pmatrix} \mathrm{Var} (a) + \mathrm{Var} (a^\dagger) + 2\langle a^\dag a\rangle_{\rm c} + 1 & \I(\mathrm{Var} (a^\dagger) - \mathrm{Var} (a)) \\ \I(\mathrm{Var} (a^\dagger) - \mathrm{Var} (a)) & -(\mathrm{Var} (a) + \mathrm{Var} (a^\dagger)) + 2\langle a^\dag a\rangle_{\rm c}  + 1 \end{pmatrix} 
    \end{equation}
    where we introduce the short-hand notations
    \begin{equation}
       \langle AB\rangle_{\rm c} := \langle AB\rangle - \langle A\rangle \langle B\rangle  ,\quad  \mathrm{Var} (A) := \langle A^2\rangle_{\rm c} =  \langle A^2\rangle -\langle A\rangle^2.
    \end{equation}
\end{thm}
Using \cref{thm:trace_distance_gaussian}, we can analyze the convergence of the state $\rho(t)$ to the steady state $\rho_{\rm ss}$ for both the purely dissipative case and the case with Hamiltonian contribution, by analyzing the dynamics of the covariance matrix and the first-moment vector.
\begin{thm}[Rigorous version of Proposition \ref{prop:free_boson}, with both Hamiltonian and dissipative contributions]
    If the initial state $\rho(0)$ is a Gaussian state, then
    \begin{equation}
        d(\rho(t), \rho_{\rm ss}) \le  C e^{-(\gamma_c^2-\gamma_h^2)t/2},
    \end{equation}
    where $C$ is a constant that depends on the initial state $\rho(0)$ and the system parameters $\Omega,\gamma_c,\gamma_h$, but is independent of time $t$. Here $\rho_{\rm ss}$ is the squeezed thermal state with the stationary first and second moments computed in \cref{eq:N_boson,eq:a2_boson_ss,eq:a_boson_ss}.
\end{thm}
\begin{proof}
    We calculate the difference of the covariance matrices $V(t)$ and $V_{\rm ss}$ by noticing that 
    \begin{equation}
        \begin{aligned}
      \abs{  \mathrm{Var} (a)_t - \mathrm{Var} (a)_{\rm ss}}
    &\note{$\langle a\rangle_{\rm ss} = 0$ by \cref{eq:b_boson}} =  \abs{\langle a^2\rangle_t  -\langle a^2\rangle_{\rm ss}  - \langle a\rangle_t^2}\\
    & \note{\cref{eq:b2_boson,eq:b_boson}}\le \abs{\langle a^2\rangle_{\rm ss} - \langle a^2\rangle_0 } e^{-(\gamma_c^2-\gamma_h^2)t} + \abs{\langle a\rangle_0}^2 e^{-(\gamma_c^2-\gamma_h^2)t}.
        \end{aligned}
    \end{equation}
    Note that $\langle a^2\rangle = \Tr(\rho a^2) = \Tr((\rho a^2)^\dag)^\ast = \langle a^{\dag 2}\rangle^\ast$, thus similar to \cref{eq:b2_boson,eq:b_boson}, we can also have 
    \begin{equation}
        \begin{aligned}
        \langle a^{\dag 2}\rangle_t& = \langle a^{\dag 2}\rangle_{\rm ss} + {\left(\braket{a^{\dag 2}}_0 - \langle a^{\dag 2}\rangle_{\rm ss}\right) e^{-(\gamma_c^2-\gamma_h^2 -2\I \Omega)t}}, \quad \langle a^{\dag 2}\rangle_{\rm ss} = \frac{\gamma_c \gamma_h}{\gamma_h^2 - \gamma_c^2 + 2\I \Omega}.
        \end{aligned}
    \end{equation}
    Moreover, for the connected correlation term, we have
    \begin{equation}
        \begin{aligned}
        \abs{ \braket{a^\dag a}_{\mathrm{c}, t}
        - \langle a^\dag a\rangle_{\rm c, ss}}& \le \abs{\langle a^\dag a\rangle_t - \langle a^\dag a\rangle_{\rm ss}} + \abs{\langle a^\dag\rangle_t \langle a\rangle_t - \langle a^\dag\rangle_{\rm ss} \langle a\rangle_{\rm ss}}\\
        & \note{\cref{eq:N_boson,eq:b_boson}} \le \abs{\langle a^\dag a\rangle_{\rm ss} - \langle a^\dag a\rangle_0} e^{-(\gamma_c^2-\gamma_h^2)t} + \abs{\braket{a^\dag}_0 \braket{a}_0} e^{-(\gamma_c^2-\gamma_h^2)t}.
        \end{aligned}
    \end{equation}
    Therefore, each of the four entries of the matrix $V(t) - V_{\rm ss}$ converges to zero with an exponential rate $e^{-(\gamma_c^2-\gamma_h^2)t}$, and we have
    \begin{equation}
        \norm{V(t) - V_{\rm ss}} \le C_1' e^{-(\gamma_c^2-\gamma_h^2)t},
    \end{equation}
    for some constant $C_1'$ that depends on the initial state and the system parameters $\Omega,\gamma_c,\gamma_h$, but is independent of time $t$.
    And finally, for the first-moment vector, we have
    \begin{equation}
        \begin{aligned}
        \norm{\bv m_t - \bv m_{\rm ss}} & = \sqrt{2} \sqrt{(\Re \langle a\rangle_t - \Re \langle a\rangle_{\rm ss})^2 + (\Im \langle a\rangle_t - \Im \langle a\rangle_{\rm ss})^2} = {\sqrt{2}}\abs{\langle a\rangle_t - \langle a\rangle_{\rm ss}}\\
        &\note{\cref{eq:b_boson}} = {\sqrt{2}}\abs{\langle a\rangle_0} e^{-(\gamma_c^2-\gamma_h^2)t/2}.
        \end{aligned}
    \end{equation}
    Putting all the pieces together, and by \cref{thm:trace_distance_gaussian} and the trace H\"older inequality, we get the convergence of $d(\rho(t), \rho_{\rm ss})$
    \begin{equation}
        \begin{aligned}
        d(\rho(t), \rho_{\rm ss}) &\le \frac{\sqrt 3+1}{16} \Tr(\abs{V(t)-V_{\rm ss}} \Xi ^T (V(t)+V_{\rm ss}) \Xi) + \sqrt{\frac{\min(\norm{V(t)} , \norm{V_{\rm ss}} )}{2}} \norm{\bv m_t - \bv m_{\rm ss}}\\
                & \le \frac{\sqrt 3+1}{16} \norm{V(t)-V_{\rm ss}}\Tr(V(t)+  V_{\rm ss}) + \sqrt{\frac{  \norm{V_{\rm ss}}}{2}} \norm{\bv m_t - \bv m_{\rm ss}}\\
                & \le  \frac{\sqrt 3+1}{8}C_1' e^{-(\gamma_c^2-\gamma_h^2)t}\max_{t\ge 0}(4\langle a^\dag a\rangle_t -4\abs{\langle a\rangle_t}^2+ 2) + C_2 e^{-(\gamma_c^2-\gamma_h^2)t/2} \\
                 & \note{\cref{eq:N_boson}} \le C_1 e^{-(\gamma_c^2-\gamma_h^2)t} + C_2 e^{-(\gamma_c^2-\gamma_h^2)t/2} \le (C_1+C_2)  e^{-(\gamma_c^2-\gamma_h^2 ) t/2}
        \end{aligned}
    \end{equation}
    where $C_1,C_2$ are constants that depend on the initial state and the system parameters $\Omega,\gamma_c,\gamma_h$, but are independent of time $t$.
\end{proof}

\subsubsection{Details on the bosonic squeezed state}\label{appendix:BosonPopulation}

We first derive the explicit form of the steady state in the purely dissipative case \cref{eq:squeezed_vacuum_state_alpha_zero}. For this, we need the following Proposition~\ref{prop:squeezed_vacuum_expectation} on the expectation values of $a$, $a^\dag a$ and $a^2$ for a squeezed vacuum state $\ket{\Psi_\zeta}$ which is a well-known result in quantum optics.

\begin{prop}[\cite{GrynbergAspectFabre2010}]\label{prop:squeezed_vacuum_expectation}
For the squeezed vacuum state $\ket{\Psi_\zeta} $ in \cref{eq:squeezed_vacuum_state_alpha_zero}, we write the parameter $\zeta \in \CC$ as $\zeta = re^{2\I\theta}$. Then
the expectation values of $a$, $a^\dag a$ and $a^2$ are given by
\begin{equation}\label{eq:squeezed_vacuum_expectation}
    \langle a^2 \rangle_{\zeta} = -e^{2\I \theta } \cosh r\sinh r,\quad \langle a^\dag a\rangle_{\zeta} =  \sinh^2 r,\quad \langle a \rangle_{\zeta} = 0.
\end{equation}
\end{prop}
\begin{proof}[Proof of \cref{eq:squeezed_vacuum_state_alpha_zero}]
    In the purely dissipative case, the steady state moments are given by \cref{eq:N_boson,eq:b2_boson,eq:b_boson} with the Hamiltonian contribution removed, which yields
    \begin{equation}\label{eq:steady_state_boson_no_Hamiltonian}
    \langle N\rangle_{\rm ss} = \frac{s^2}{1-s^2},\quad \langle a^2\rangle_{\rm ss} = \frac{s}{s^2-1},\quad \langle a\rangle_{\rm ss} = 0.
    \end{equation}
    We solve for the parameter $\zeta$ by matching \cref{eq:steady_state_boson_no_Hamiltonian} with \cref{eq:squeezed_vacuum_expectation} in Proposition~\ref{prop:squeezed_vacuum_expectation}. We find that
    \begin{equation}
        r = \text{arctanh}(s),\quad \theta = \pi.
    \end{equation}
    The Fock basis expansion of the squeezed vacuum state $\ket{\Psi_\zeta} = \exp\!\left(\frac12 (\zeta^\ast aa  - \zeta a^\dag a^\dag)\right) \ket{0}_a$ is given by \cite{GrynbergAspectFabre2010}
    \begin{equation}
       \ket{\Psi_\zeta} = \frac{1}{\sqrt{\cosh r}} \sum_{n=0}^{\infty} (-e^{2\I \theta} \tanh r)^n \frac{\sqrt{(2n)!}}{2^n n!} \ket{2n}.
    \end{equation}
    Plugging in the values of $r$ and $\theta$, we get the Fock basis expansion of the steady state as
    \begin{equation}
        \ket{\Psi_{\zeta}}  = (1-s^2)^{1/4} \sum_{n=0}^{\infty} (-s)^n \frac{\sqrt{(2n)!}}{2^n n!} \ket{2n}.
    \end{equation}
We verify that $\ket{\Psi_\zeta}$ is in $\ker K = \ker (a+sa^\dag)$ by directly applying $K$ to $\ket{\Psi_\zeta}$:
\begin{equation}
\begin{aligned}
    K \ket{\Psi_\zeta} &= {\gamma_c } (a + sa^\dag) \ket{\Psi_\zeta} 
    \\
    & \propto   \sum_{ n=0}^\infty  (-s)^{n+1} \sqrt{2n+2} \cdot \frac{\sqrt{(2n+2)!}}{2^{n+1} (n+1)!} \ket{2n+1} - \sum_{n=0}^\infty   (-s)^{n+1} \sqrt{2n+1} \cdot \frac{\sqrt{(2n)!}}{2^n n!} \ket{2n+1}.
\end{aligned}
\end{equation}
Note that $
    \sqrt{2n+2} \cdot \frac{\sqrt{(2n+2)!}}{2^{n+1} (n+1)!} 
    = \sqrt{2n+1} \cdot \frac{\sqrt{(2n)!}}{2^n n!},
$
therefore the two sums cancel each other and we get $K \ket{\Psi_\zeta} = 0$. This is consistent with the fact that $\ket{\Psi_\zeta}$ is a pure steady state of the dynamics, and thus must be annihilated by the jump operator $K$.
\end{proof}

In the numerical simulation in Sec.~\ref{sec:free_bosons}, although the steady state is a pure squeezed vacuum state, the states along the evolution are mixed states. In this section, we provide some details on how to compute the population $\rho_{nn}(t) = \bra{n}\rho(t)\ket{n}$ of the Fock state $\ket{n}$ at time $t$ using the covariance matrix of a mixed squeezed thermal state. For simplicity, we assume that $\wh f(\Omega ) = s$ and $\wh f(-\Omega) = 1$ with $0\le s <1$. We assume that the initial state has zero first moment, i.e., $\langle a \rangle_0 = 0$, which means that $\langle a\rangle_t = 0$ for all $t\ge 0$ and thus we do not need a displacement operator in the following discussion. For a systematic discussion on photon-number distribution in more general cases, we refer to \cite{MarianMarian1993}.

A squeezed thermal state can be written as
\begin{equation}
    \rho = S(\zeta) \rho_{\text{th}} S^\dag(\zeta),
\end{equation}
where
\begin{equation}
    \rho_{\rm th}  = \sum_{n=0}^\infty \frac{\overline{n}^n}{(\overline{n}+1)^{n+1}} \ket{n}\bra{n},\quad \overline{n} = \frac12 \sqrt{(2 \mathsf N+1)^2 - 4|\mathsf M|^2} - \frac12,
\end{equation}
with $\mathsf N = \langle a^\dag a\rangle$ and $\mathsf M = \langle a^2\rangle$. The squeezing operator is $S(\zeta) = \exp\!\left(\frac12 (\zeta^\ast aa  - \zeta a^\dag a^\dag)\right)$, and the squeezing parameter $\zeta$ is determined by 
\begin{equation}
    \cosh 2r = \frac{2\mathsf N+1}{2\overline{n} +1} ,\quad  {e^{2\I\theta} \sinh 2r = -\frac{2\mathsf M}{2\overline{n} +1}},\quad \zeta = r e^{2\I\theta}.
\end{equation}
Then the exact probability $\rho_{nn}(t)$ can be computed by the following formula \cite{MarianMarian1993}:
\begin{equation}\label{eq:photon_number_distribution}
    \rho_{n,n} = \bra{n}\rho\ket{n} = \frac{\mathsf{h}_2^{n/2}}{\mathsf{h}_0^{(n+1)/2}} \Theta_n\!\left(\frac{\mathsf{h}_1}{ \sqrt{\mathsf{h}_0 \mathsf{h}_2}}\right),
\end{equation}
where $\Theta_n(x)$ denotes the $n$-th Legendre polynomial and the variables $\mathsf{h}_0$, $\mathsf{h}_1$ and $\mathsf{h}_2$ are defined as
\begin{equation}
    \mathsf{h}_0    = (\mathsf{N}+1)^2 - |\mathsf{M}|^2 ,\quad
    \mathsf{h}_1 = \mathsf{N}(\mathsf{N}+1) - |\mathsf{M}|^2 , \quad
    \mathsf{h}_2 = \mathsf{N}^2 - |\mathsf{M}|^2.
\end{equation}
Using \cref{eq:photon_number_distribution} one could compute the photon number distribution at any time $t$ using the moment information $\mathsf{N}(t) = \langle a^\dag a\rangle_t$ and $\mathsf{M}(t) = \langle a^2\rangle_t$. This shows the exponentially decaying behavior of the population $\rho_{nn}(t)$ in the Fock basis for bosonic squeezed thermal states.

\section{The subsystem bipartite spin fluctuation and the Luttinger parameter}
\label{appendix:analysis}

For $J_1$--$J_2$ and XXZ spin chains in \cref{sec:models}, we make use of the fact that the subsystem bipartite spin fluctuation \cite{ SongRachelFlindtEtAl2012, RachelLaflorencieSongEtAl2012}
\begin{equation}
\mathcal{F}(\ell) = \sum_{i,j\in \mathcal{A}_\ell} \left( \langle S_i^z S_j^z \rangle - \langle S_i^z \rangle \langle S_j^z \rangle \right),\quad \mathcal{A}_\ell = \{1,2,\cdots,\ell\},
\end{equation}
of the ground state has a logarithmic scaling with subsystem size $\ell$
\begin{equation}\label{eq:fluctuation_scaling}
\mathcal{F}_{\text{OBC}}(\ell) = \frac{\mc K}{2
\pi^2} \log \wt \ell +  \mc O(\ell^{-\alpha_1}), \qquad \mathcal{F}_{\text{PBC}}(\ell) = \frac{\mc K}{\pi^2} \log \wt \ell + \mc O(\ell^{-\alpha_2})
\end{equation}
in the gapless phase, where $\widetilde{\ell} = \frac{L}{\pi} \sin \left( \frac{\pi \ell}{L} \right)$ is the conformal distance which accounts for finite-size effects, and
$\alpha_{1,2}>0$ are non-universal constants that correspond to the decaying rate of the subleading term. 
In the N\'eel or VBS phase, $\mathcal{F}(\ell)$ however saturates to a constant in the thermodynamic limit
\cite{SongRachelFlindtEtAl2012}. In fact, it is shown that in the 1D case, if the spin is globally conserved, then $\mathcal{F}(\ell)$ obeys the same scaling as the entanglement entropy \cite{SongRachelLeHur2010,RachelLaflorencieSongEtAl2012}.

In practice, we extract the Luttinger parameter $\mc K$ in the TLL regime by fitting the numerical data using the model 
up to the subleading term  \cite{SongRachelFlindtEtAl2012, RachelLaflorencieSongEtAl2012}
\begin{equation}\label{eq:fluctuation_fitting_J1J2}
    \mathcal{F}_{\text{PBC}}(\ell) = \frac{\mc K}{\pi^2} \log \wt \ell + a_0 + a_1 \frac{(-1)^\ell}{\wt \ell  }+ a_2(-1)^{\ell}  +\mc O(\ell^{-2}).
\end{equation}
For the XX Hamiltonian, it indeed holds that $a_1=a_2=0$~\cite{SongRachelFlindtEtAl2012}.

\section{Additional numerical results}
\label{appendix:numerical_1D}

\subsection{$J_1$--$J_2$ model}\label{appendix:J1J2_supplementary}

\paragraph{The Luttinger parameter.}
The Luttinger parameter $\mc K$ serves as a powerful probe of quantum phase transitions in one-dimensional systems with gapless phases described by Luttinger liquid theory. For
benchmarking purposes, we compute the ground state of 
the $J_1$--$J_2$ model with PBCs, using density matrix renormalization group (DMRG) with a maximal bond dimension $\chi = 800$, and extract the Luttinger parameter $\mc K$ by fitting the subsystem bipartite spin fluctuation $\mathcal{F}$. The results are shown in \cref{fig:DMRG_K_J1J2}. In the gapless phase, the fitted Luttinger parameter remains greater than $\frac12$, while {the same fitted logarithmic coefficient} decreases sharply upon entering the gapped phase. The $\mc K = \frac12$ transition point 
occurs approximately at the critical point $J_2/J_1 \approx 0.241$. Notably, even for relatively modest system sizes, the fitted $\mc K$ already exhibits a clear signature of the phase transition, making this approach an efficient probe of quantum criticality and phase boundaries, as discussed in Ref.~\cite{RachelLaflorencieSongEtAl2012}.

\begin{figure}[!h]
    \centering
    \includegraphics[width=0.4\linewidth]{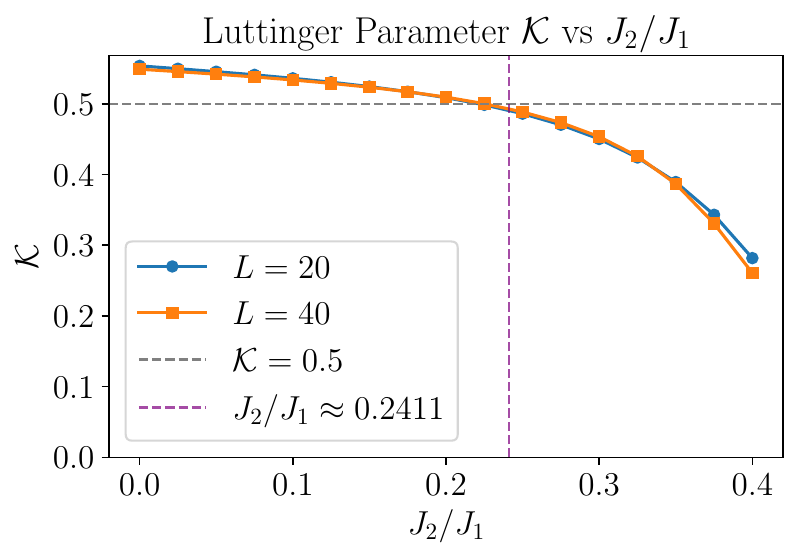}
    \caption{ {The fitted logarithmic coefficient, interpreted as the Luttinger parameter in the TLL regime,} for the ground state of $J_1$--$J_2$ model with PBCs. The critical point is around $J_2/J_1 \approx 0.241$.}
    \label{fig:DMRG_K_J1J2}
\end{figure}

 We provide detailed numerical results for fitting {the logarithmic coefficient $\mc K$} for the $J_1$--$J_2$ model along the cooling dynamics in \cref{sec:J1J2}. The least-squares fitting is performed using the Levenberg--Marquardt algorithm according to the model specified in \cref{eq:fluctuation_fitting_J1J2}. We use the filter resolution $\Delta = 0.5$ for various Hamiltonian parameters $J_2/J_1$ within the critical region $0\le J_2/J_1 \le 0.35$. The results are shown in \cref{fig:fit_J1J2_res0.5}.  

  \begin{figure}[!h]
    \centering
   \begin{subfigure}[b]{0.34\linewidth}
        \caption{$J_2/J_1 = 0.150$}
    \includegraphics[width=\linewidth]{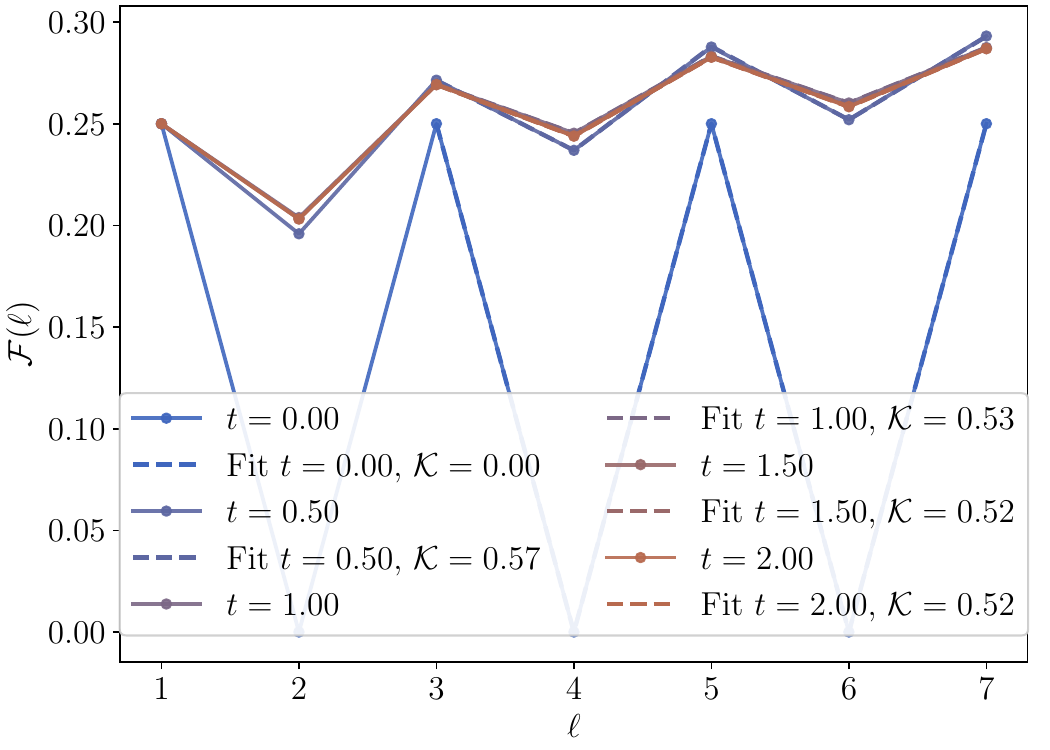}
    \end{subfigure}
    \begin{subfigure}[b]{0.34\linewidth}
           \caption{$J_2/J_1 = 0.180$}
    \includegraphics[width=\linewidth]{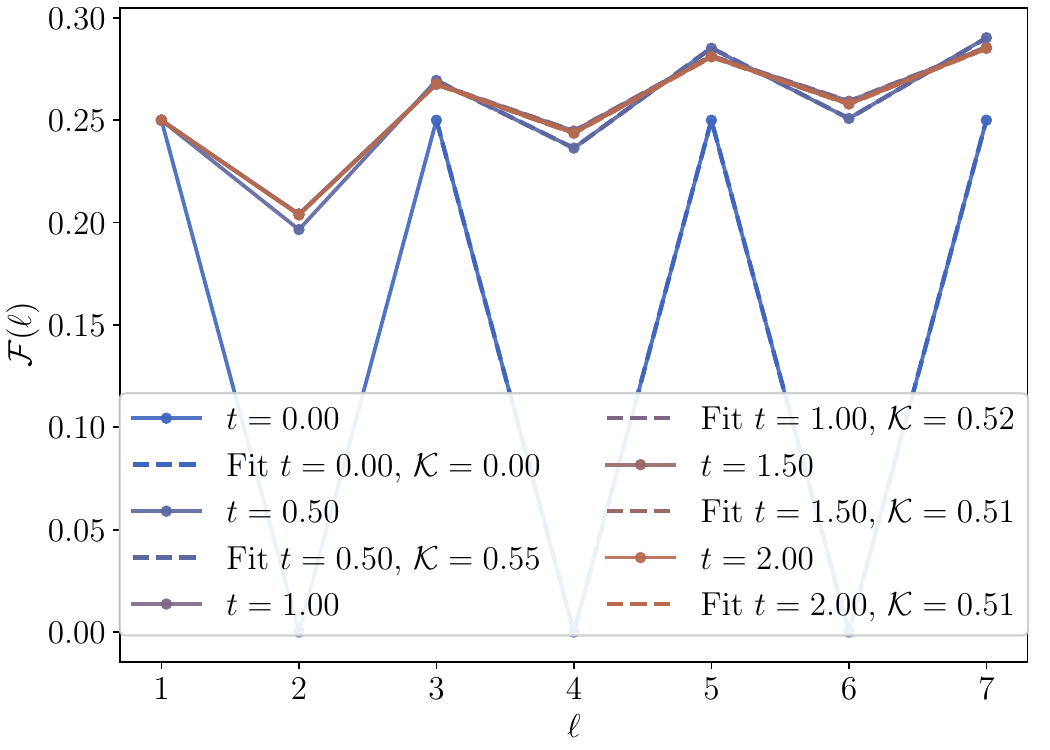}
    \end{subfigure}
    \\
    \begin{subfigure}[b]{0.34\linewidth}
              \caption{$J_2/J_1 = 0.280$}
    \includegraphics[width=\linewidth]{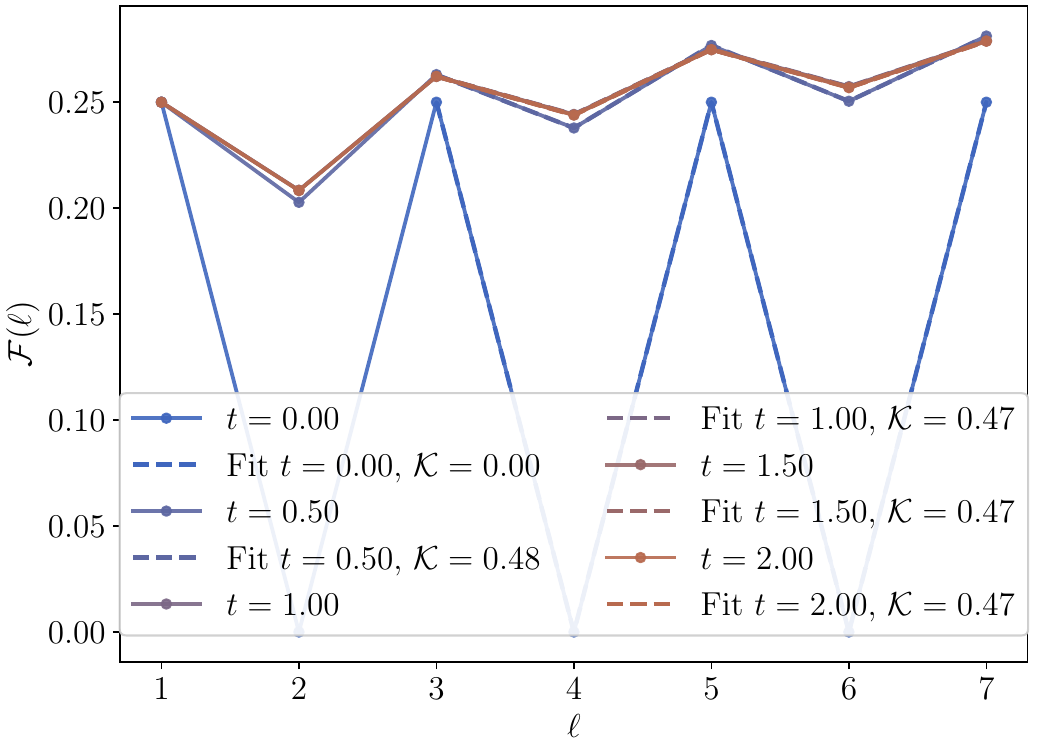}
    \end{subfigure}
    \begin{subfigure}[b]{0.34\linewidth}
                \caption{$J_2/J_1 = 0.350$}
    \includegraphics[width=\linewidth]{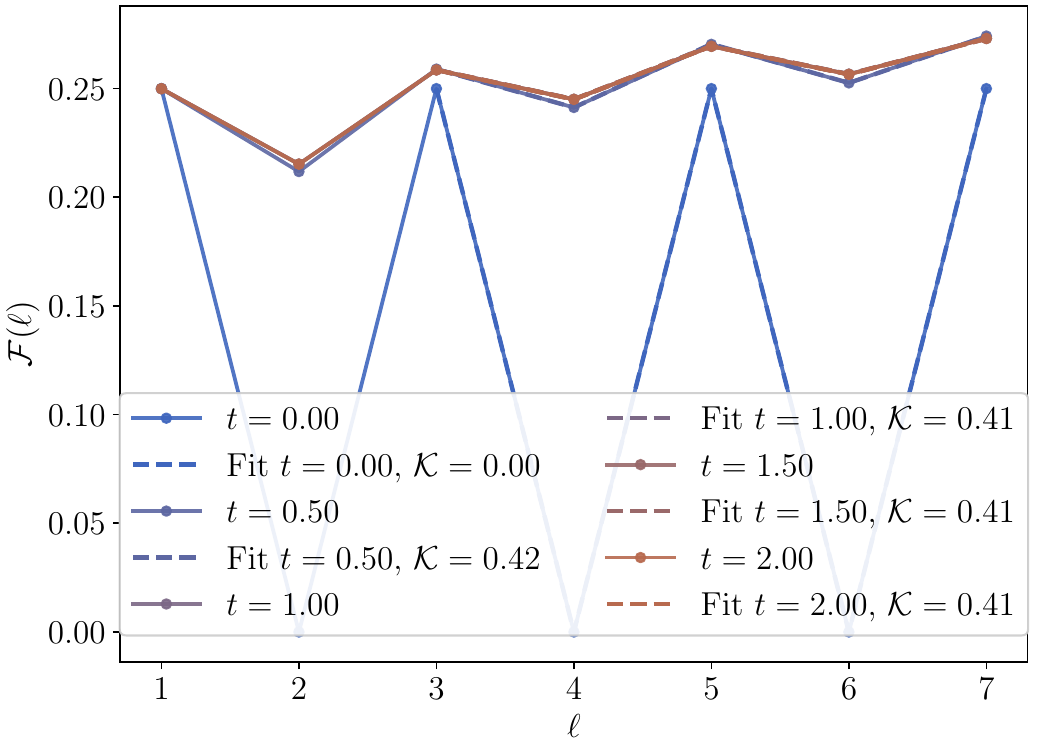}
    \end{subfigure}
\caption{Fitting the Luttinger parameter $\mc K$ along the cooling dynamics for the $J_1$--$J_2$ model with filter resolution $\Delta = 0.5$. The results are shown for $J_2/J_1 = 0.150,~0.180,~ 0.280,~ 0.350$.\label{fig:fit_J1J2_res0.5}}
  \end{figure}

{At the Majumdar--Ghosh initial state, which is the ground state at $J_2/J_1 = 0.5$, the fitting quality is poor because \cref{eq:fluctuation_fitting_J1J2} is derived from the CFT description of the Luttinger liquid and is not expected to apply deep inside the gapped phase. As the cooling dynamics proceeds, the fitting quality improves. For $J_2/J_1 < 0.241$, the fitted Luttinger parameter approaches values larger than $\frac12$, consistent with the gapless phase. For $J_2/J_1 > 0.241$, the same fit should instead be read as an effective finite-size diagnostic, where the thermodynamic fluctuation is expected to saturate in the gapped phase, and the logarithmic coefficient drops below $\frac12$.}

 \begin{figure}
    \centering
    \begin{subfigure}{0.66\linewidth}
        \overlaption[fig:quench_S_pi]{\includegraphics[width=\linewidth]{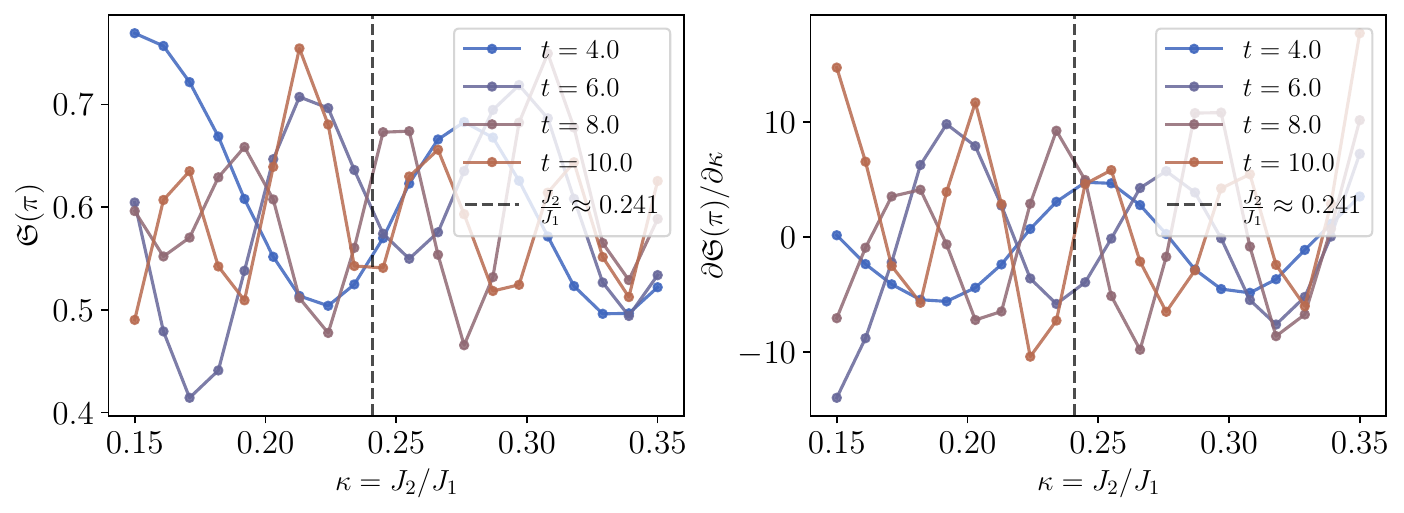}}
    \end{subfigure}\\
    \begin{subfigure}{0.66\linewidth}
       \overlaption[fig:quench_D]{\includegraphics[width=\linewidth]{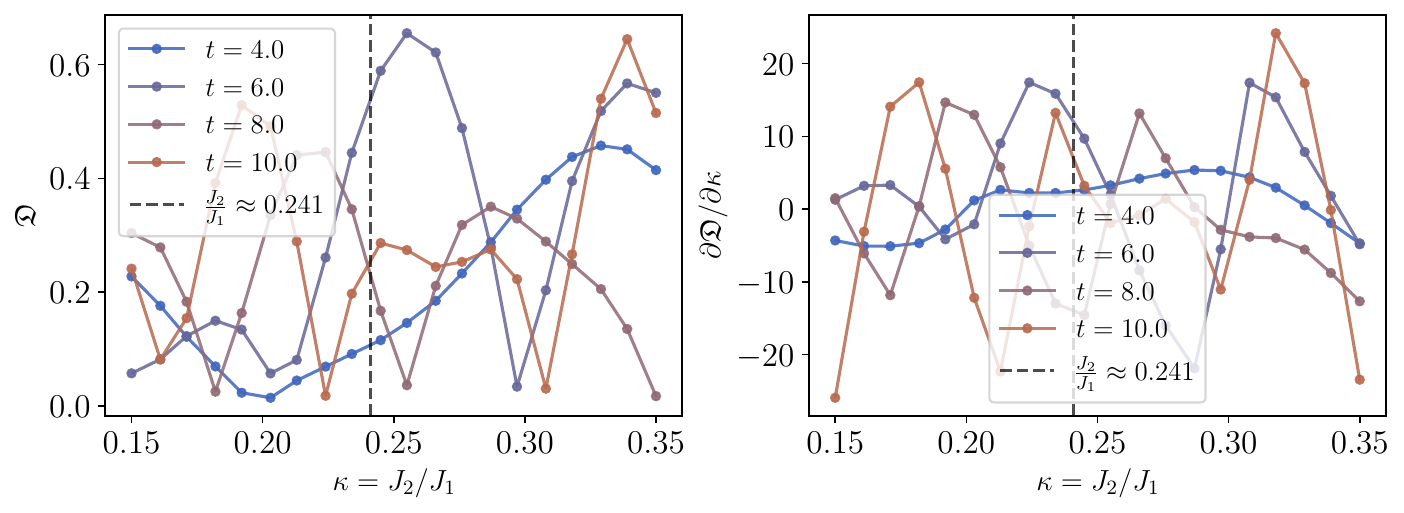}}
    \end{subfigure}\\
     \begin{subfigure}{0.33\linewidth}
        \overlaption[fig:quench_S_pi_trajectory]{\includegraphics[width=\linewidth]{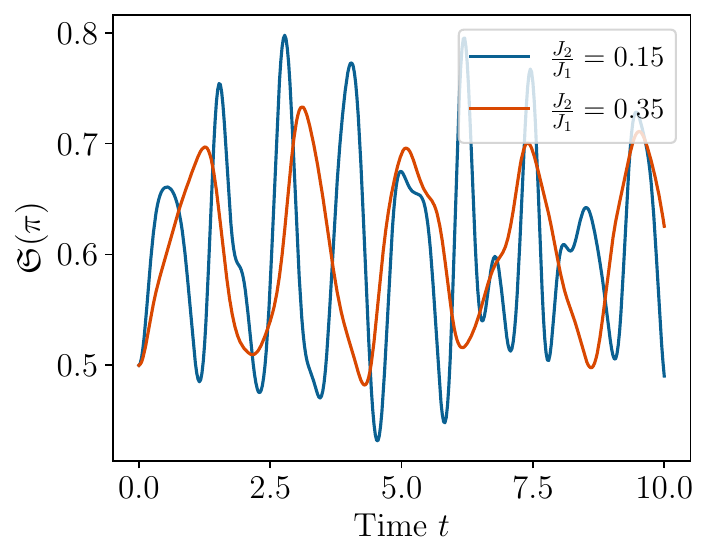}}
    \end{subfigure}
     \begin{subfigure}{0.33 \linewidth}
        \overlaption[fig:quench_D_trajectory]{\includegraphics[width=\linewidth]{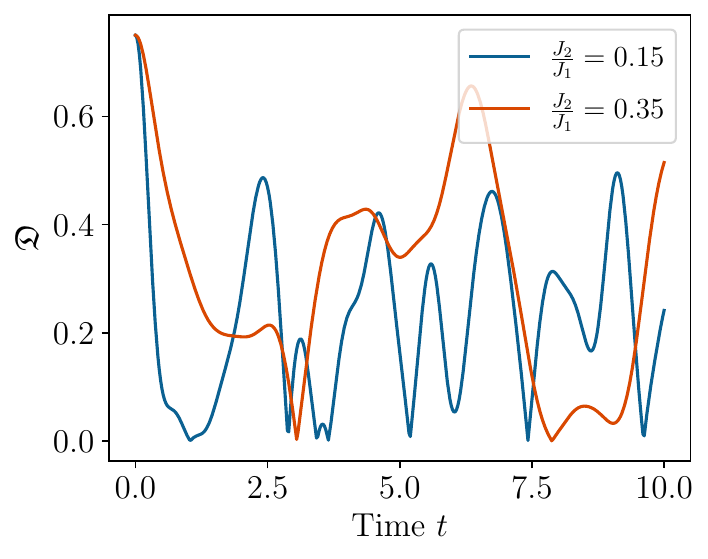}}
    \end{subfigure}
    \caption{The quench dynamics of the $J_1$--$J_2$ model with $L=14$ and PBC. The system is initialized in the Majumdar--Ghosh state. {\bf (a)} The spin structure factor $\mathfrak{S}(\pi)$ and its derivative with respect to $\kappa = J_2/J_1$ from the real time evolution at $t= 4,6,8,10$. {\bf (b)}  The dimer order parameter $\mathfrak D$ and its derivative with respect to $\kappa = J_2/J_1$ from the real time evolution at $t= 4,6,8,10$. {\bf (c, d)} The time evolution of $\mathfrak{S}(\pi)$ and $\mathfrak D$ for $J_2/J_1 = 0.150$ and $ 0.350$.}\label{fig:quench_J1J2}
 \end{figure}

\paragraph{Benchmark against the quench dynamics.}
 
As discussed in the main text, while the quench dynamics is a useful tool for probing the quantum phase diagrams, it is not easy to reveal informative signals in the quench dynamics for the BKT transition in the $J_1$--$J_2$ model. Here we provide the numerical results for the quench dynamics in the $J_1$--$J_2$ model in the vicinity of the critical point $J_2/J_1 \approx 0.241$. The system is initialized in the Majumdar--Ghosh state. We compute the spin structure factor $\mathfrak{S}(q)$ at $\pi$ as well as the dimer order parameter $\mathfrak D$, defined as
\begin{equation}
    \mathfrak S(q ) := \frac{1}{L} \sum_{i,j} e^{-\I q(i-j)} \langle S_i^z S_j^z \rangle, \quad \mathfrak D = \frac{1}{L} \sum_{i}  \abs{\langle \mathbf{S}_i \cdot \mathbf{S}_{i+1}   \rangle - \langle \mathbf{S}_{i-1} \cdot \mathbf{S}_{i}\rangle}.
\end{equation}
At time $t$, we compute how the observables change with respect to the parameter $\kappa = J_2/J_1$ in the vicinity of $\kappa_{\rm crit} \approx 0.241$, following \cite{HaldarMallayyaHeylEtAl2021}. The results are shown in \cref{fig:quench_J1J2}. We see that the spin structure factor $\mathfrak{S}(\pi)$ and the dimer order parameter $\mathfrak D$ do not show clear signatures of the phase transition along the quench dynamics. The derivatives of these observables with respect to $\kappa$ also do not show clear peaks around the critical point. {This finite-size and finite-time response is consistent with the difficulty of resolving a BKT transition from such quench observables, since the VBS correlation length grows exponentially near the transition.}

\subsection{Kitaev honeycomb model}\label{appendix:kitaev_supplementary}

{{Besides the Chern number, another phase-sensitive diagnostic} for this phase transition is the mutual information between the two sites in the unit cell~\cite{CuiCaoFan2010}, which can be computed as}
\begin{equation}   
    I(\bv sA,\bv sB) = 1 - H_b\left( {
        \frac12 + 2\langle S_{\bv sA}^z S_{\bv sB}^z \rangle
    }\right),
\end{equation} 
since {for the states considered here} the reduced density matrix of the two sites has an explicit form 
\begin{equation}
\rho_{\bv sA,\bv sB} = \diag\left(\frac14+\langle S_{\bv sA}^z S_{\bv sB}^z \rangle, \frac14-\langle S_{\bv sA}^z S_{\bv sB}^z \rangle, \frac14-\langle S_{\bv sA}^z S_{\bv sB}^z \rangle, \frac14+\langle S_{\bv sA}^z S_{\bv sB}^z \rangle\right).
\end{equation}
Here $H_b(p) = -p\log_2(p) -(1-p)\log_2(1-p)$ is the binary entropy. $I(\bv sA,\bv sB)$ is a function of $J_z$ and becomes non-analytic at the critical point $J_z = 0.5$. {The change in the derivative of $I(\bv sA,\bv sB)$ with respect to $J_z$ therefore gives another diagnostic of the transition.} We will discuss the details that are relevant for simulating the covariance matrix dynamics and computing the $z$-bond correlation function in Appendix~\ref{appendix:kitaev_zbond}.

As shown in \cref{fig:dissipative_2D}, {after a short evolution time the two regimes are already distinguishable from the $J_z$ dependence of $I(\bv sA,\bv sB)$ and its finite-difference derivative}, as shown in \cref{fig:dissipative_2Da,fig:dissipative_2Db}. At the same time, the {trace-norm discrepancy} between the evolved state and the ground state remains large for a wide range of the target $J_z$, as shown in \cref{fig:dissipative_2Dc}. This comparison indicates that phase identification  requires substantially less time than full ground-state preparation, even without exploiting the topological invariants.

\begin{figure*} 
    \centering
    \begin{subfigure}{0.32\linewidth}
    \overlaption[fig:dissipative_2Da]{\includegraphics[width=0.9\linewidth]{{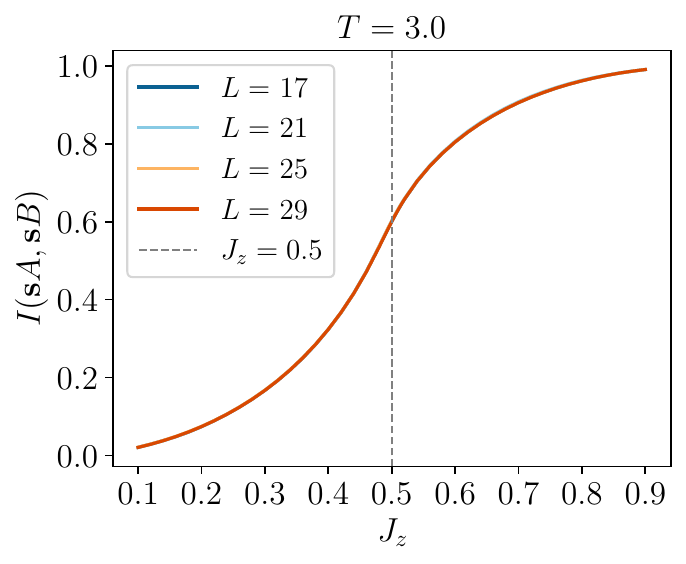}}}
    \end{subfigure}
    \hfill
    \begin{subfigure}{0.307\linewidth}
    \overlaption[fig:dissipative_2Db]{\includegraphics[width=.9\linewidth]{{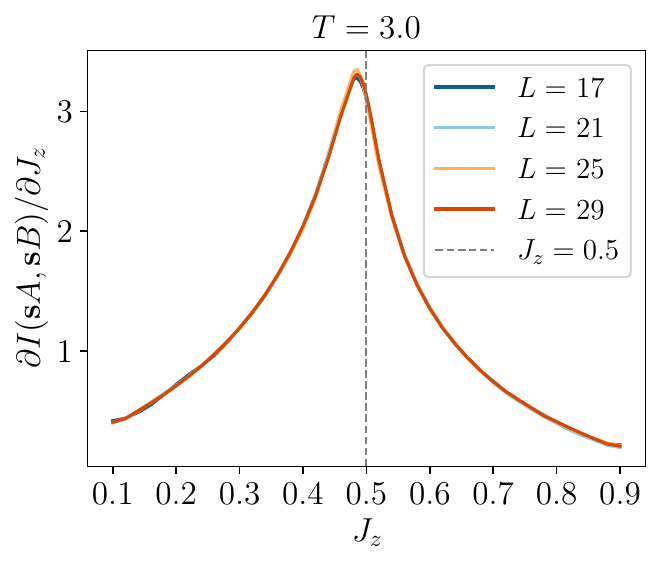}}}
    \end{subfigure}
    \hfill
    \begin{subfigure}{0.345\linewidth}
    \overlaption[fig:dissipative_2Dc]{\includegraphics[width=.9\linewidth]{{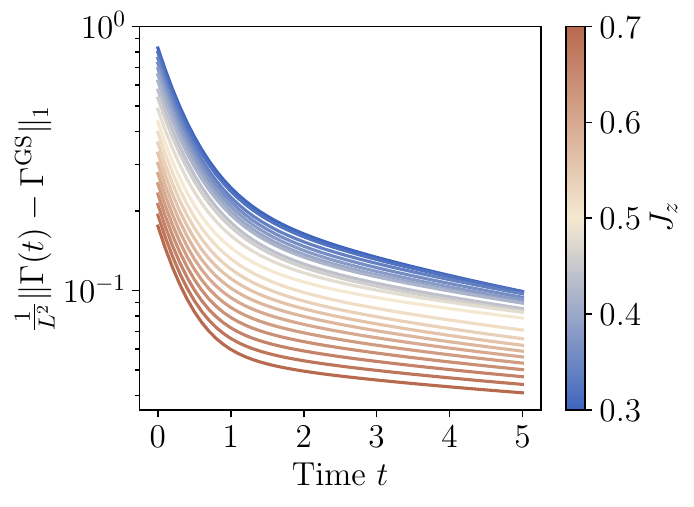}}}
    \end{subfigure}
    \caption{The dissipative cooling dynamics of the Kitaev honeycomb model with $L = 17,21,25,29$. The filter resolution $\Delta = 0.5$. We run the simulation for several values of $J_z$ from $0.3$ to $0.7$. {\bf (a)} The mutual information $I(\bv sA,\bv sB)$ between the two sites in the unit cell, computed at time $T=3$. {\bf (b)} The derivative of the mutual information with respect to $J_z$ evaluated using finite difference.  {\bf (c)} The discrepancy in the trace norm between the evolved state and the ground state along the evolution up to time $T=5$.  }
    \label{fig:dissipative_2D}
\end{figure*}

\subsection{Trends in the realistic cooling dynamics for free and interacting systems}\label{appendix:free_and_interacting_MoreNumerical}

\paragraph{Free fermions.}
We provide more results on the dissipative cooling for the tight-binding model that is equivalent to the XX model by Jordan--Wigner transformation
\begin{equation}
    H = \sum_{i=1}^{L-1}(c_{i+1}^\dag  c_i+ c_i^\dag c_{i+1}).
\end{equation}
Here we apply the same setting of the initial state and jump operators as in \cref{sec:free_fermion}, and we use the realistic filters with resolution $\Delta = 0.3,0.6$ and $1.2$. The results are shown in 
\cref{fig:imperfect_numerical_FF}. 
\begin{figure*}[!htbp]
    \centering
    \includegraphics[width=0.31\linewidth]{{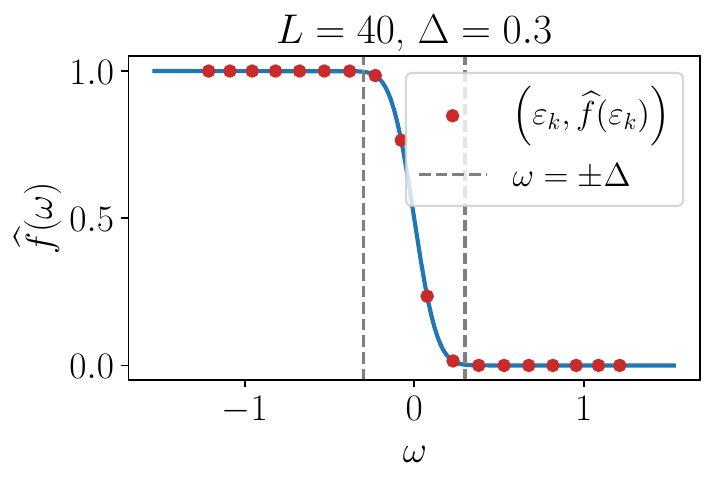}}
    \hfill
    \includegraphics[width=0.31\linewidth]{{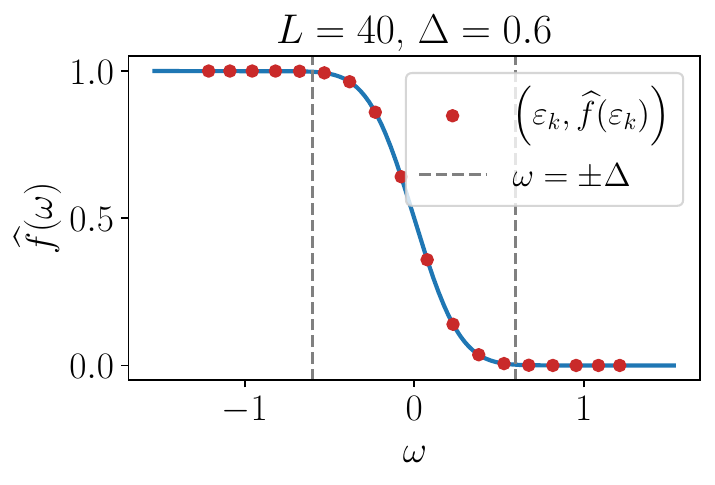}}
    \hfill
    \includegraphics[width=0.31\linewidth]{{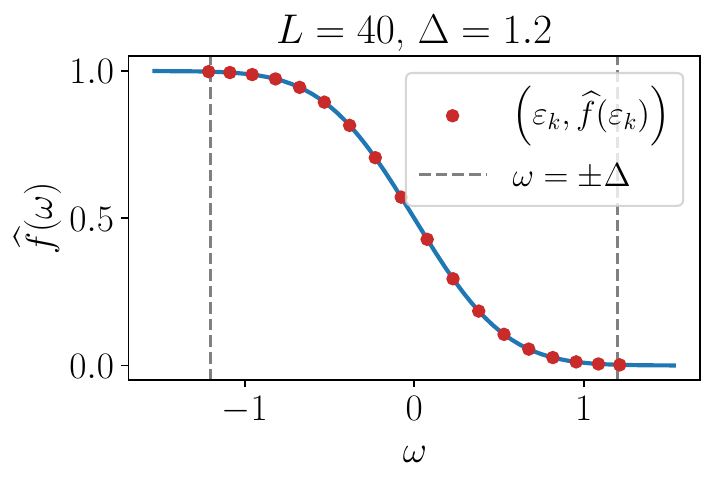}}
    \vspace{0.2cm}
    \includegraphics[width=0.325\linewidth]{{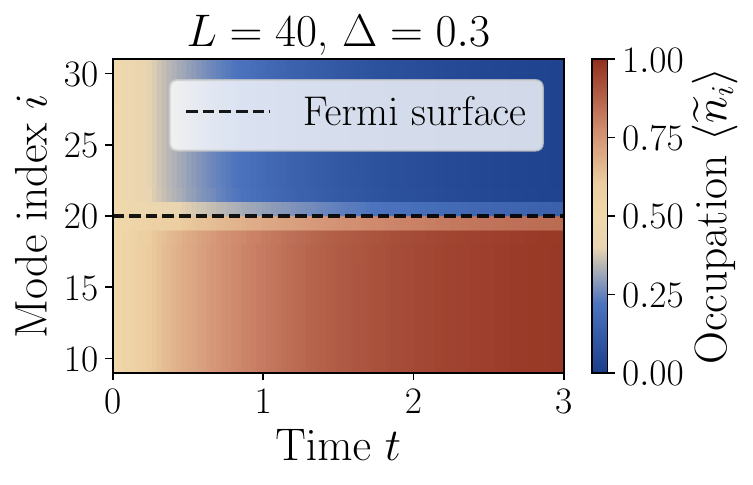}}
    \hfill
    \includegraphics[width=0.325\linewidth]{{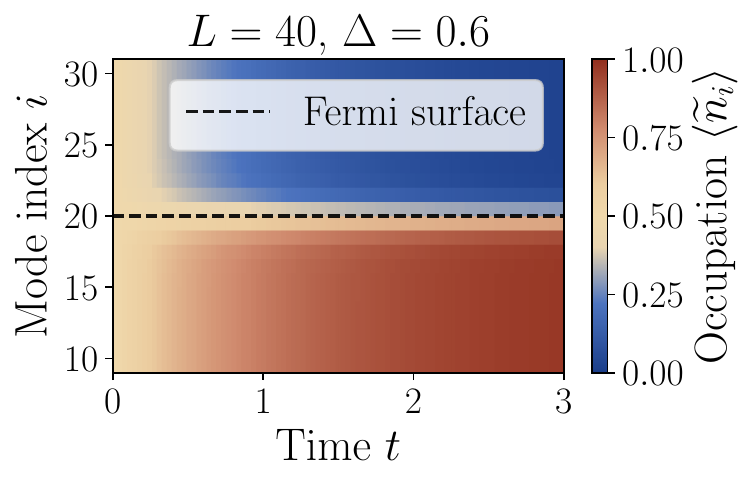}}
    \hfill
    \includegraphics[width=0.325\linewidth]{{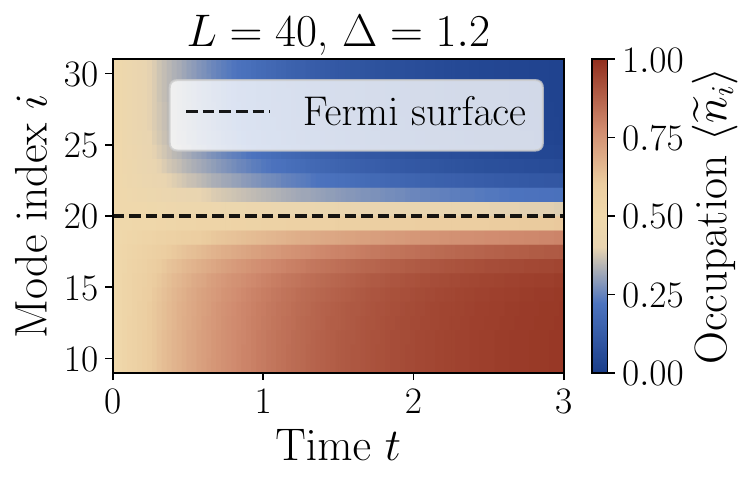}}
    \caption{\label{fig:imperfect_numerical_FF}Numerical validations of the realistic filter. Top row: filter profiles. Bottom row: the time evolution of canonical occupations. The results are shown for resolutions $\Delta = 0.3, 0.6, 1.2$ at a system size $L=40$.}
\end{figure*}

We observe that outside the filter-resolution window $\Delta$, the occupation exhibits a sharp cutoff, and the corresponding quasiparticle modes {relax to their ground-state occupations, yielding the exact half-filled profile away from the Fermi window}. In contrast, within $\Delta$, the quasiparticle modes become ``soft'' and may be partially populated, with occupation rates determined by the ratio between the cooling and heating strengths of the filter (see, for example, \cref{eq:steady_state_free_fermion}). The behavior shown in \cref{fig:imperfect_numerical_FF} is consistent with the theoretical predictions established in \cref{appendix:free_fermion}.

\begin{figure*}[!htbp]
    \centering
    \includegraphics[width=0.80\linewidth]{{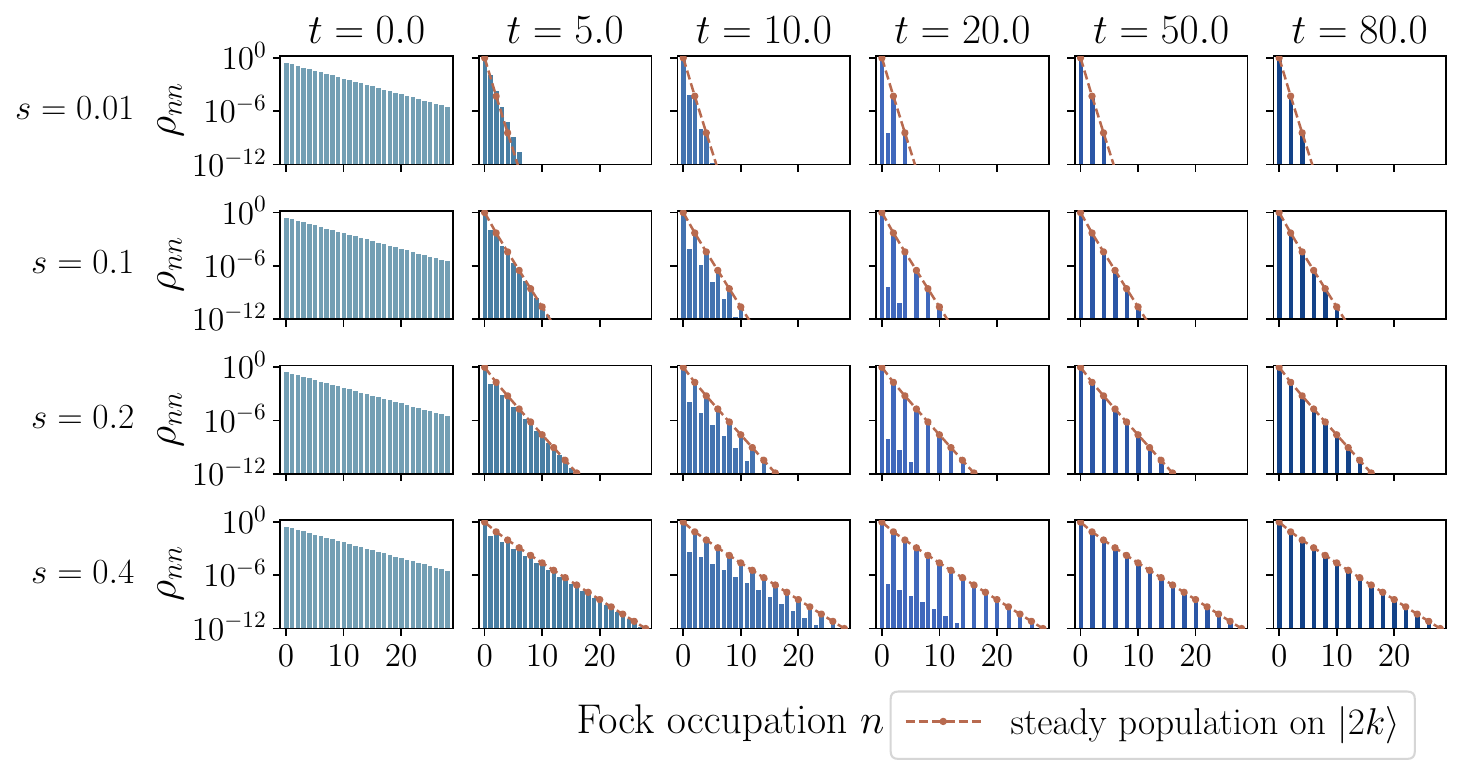}}
    \caption{\label{fig:imperfect_numerical_boson} Numerical validation of the realistic filter for a free bosonic mode. The dashed line represents the steady-state populations calculated using \cref{eq:squeezed_vacuum_state_alpha_zero}. Only populations of the even-numbered Fock states $\ket{2k}$ ($k=0,1,2,\ldots$) are shown; odd-state populations are omitted since they vanish in the steady state, i.e., $(\rho_{\rm ss})_{2k+1,2k+1} = 0$.}
    
\end{figure*}

\begin{figure*}[!htbp]
    \centering
    \includegraphics[width=0.24\linewidth]{{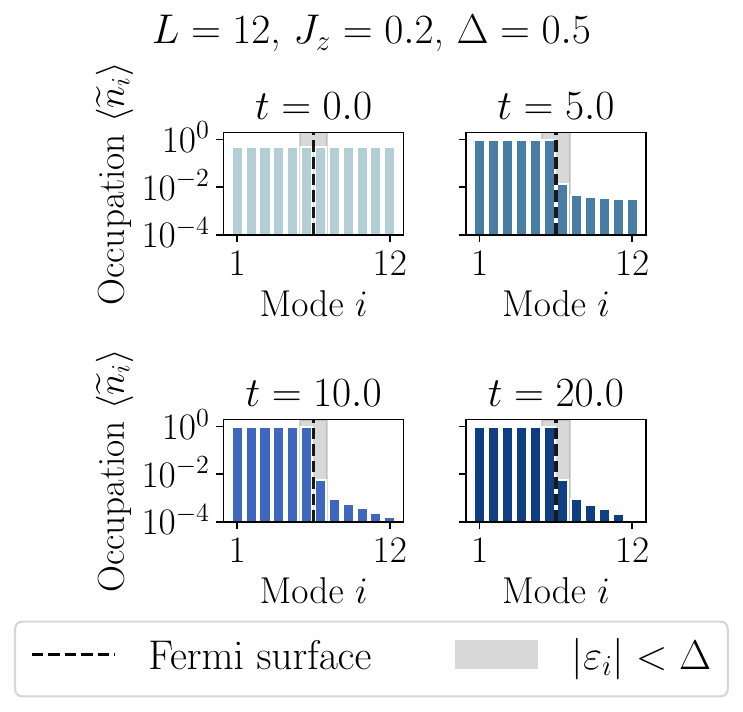}}
    \includegraphics[width=0.24\linewidth]{{figs/imperfect/freefermion_freeboson_interacting/snapshot_bars_N12_res0.5_Jz0.500.pdf}}
    \includegraphics[width=0.24\linewidth]{{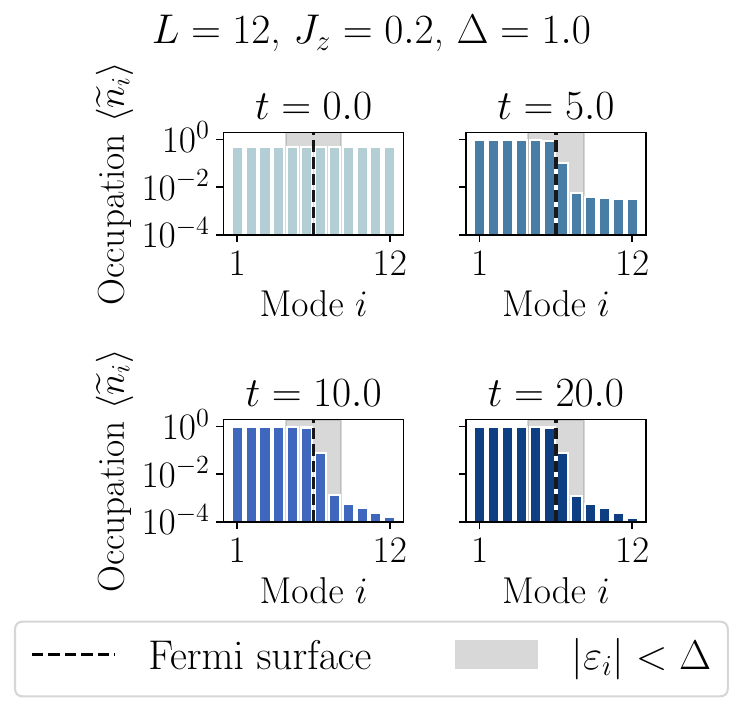}}
    \includegraphics[width=0.24\linewidth]{{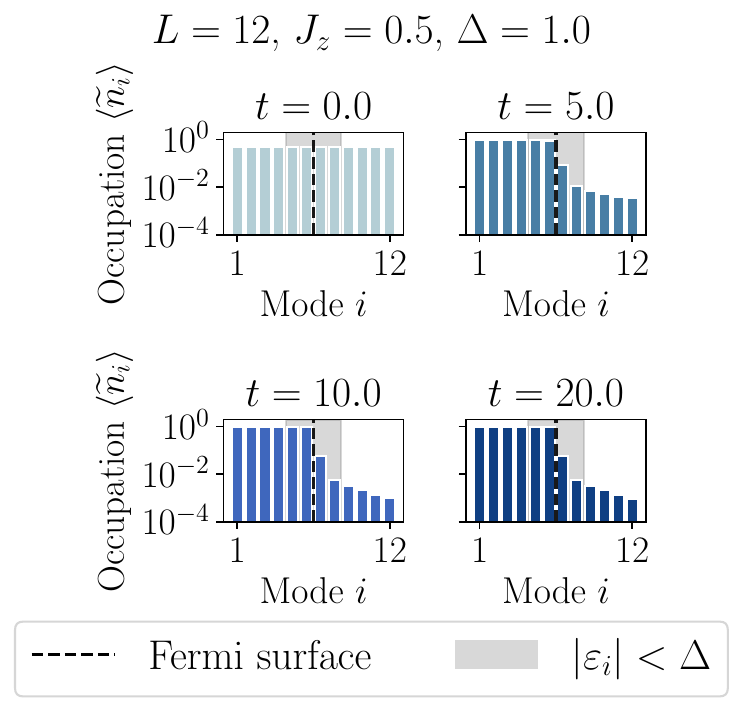}}
    \caption{\label{fig:imperfect_numerical_interacting} Numerical validation of the realistic filter for the interacting XXZ model. The results are shown for a system size $L=12$ with filter resolutions $\Delta = 0.5, 1.0$ and interaction strengths $J_z = 0.2, 0.5$.}
\end{figure*}

\paragraph{Free bosons.} 

We present additional numerical results for the single-mode bosonic system in \cref{fig:imperfect_numerical_boson}, using the same setup as in \cref{sec:free_bosons} while varying the leakage rate $s$. In this single-mode squeezed-vacuum setting, leakage {populates arbitrarily high even Fock levels while odd populations vanish in the steady state}; however, the resulting occupation exhibits an exponentially decaying profile across Fock space. Consequently, the leakage rate primarily controls the exponential tail of the steady state under realistic dissipative dynamics.

More specifically, as the leakage rate $s$ increases, the exponential decay becomes slower, leading to a broader distribution over Fock states. In contrast, in the limit of an ideal filter ($s=0$), the system converges to the pure vacuum state, with only $\ket{0}_a$ occupied. These numerical results are in excellent agreement with the theoretical prediction in \cref{eq:squeezed_vacuum_state_alpha_zero}, as discussed in Appendix~\ref{appendix:boson}.

\paragraph{Interacting XXZ model.}  As shown in \cref{fig:imperfect_numerical_interacting}, the realistic cooling dynamics of the interacting XXZ model qualitatively exhibits a combination of the free-fermion and free-boson behaviors. Recall that the canonical occupation numbers are defined with respect to the non-interacting part of the Hamiltonian. Within the resolution of the realistic filter, the energy modes are ``soft'' and not perfectly resolved, resembling the free-fermion case with a realistic filter. Beyond the filter resolution, however, the occupation profile does not display a sharp cutoff; instead, leakage effects produce an exponentially decaying distribution across energy modes, analogous to the free-boson case. We further observe that increasing the interaction strength $J_z$ enhances the leakage effect, leading to a heavier exponential tail in the canonical occupation-number distribution. In this sense, the trends observed in \cref{fig:imperfect_numerical_interacting} are consistent with the intuition developed from the free-fermion picture in \cref{fig:imperfect_numerical_FF} and the free-boson picture in \cref{fig:imperfect_numerical_boson}.

\section{Details in diagonalizing the Kitaev honeycomb model}\label{app:kitaev_diagonalization}

In this section, we provide details on how to diagonalize and solve the Kitaev honeycomb model. Most of the relevant calculations can be found in the literature, for instance, in \cite{Kitaev2006,CuiCaoFan2010}. We only include part of the calculations for completeness.

The Kitaev honeycomb model can be described by introducing four Majorana fermions on each site, {normalized by $\{\eta_i^a,\eta_j^b\}=\delta_{ij}\delta_{ab}$,}
\begin{equation}
    {S_i^\alpha = \I \eta_i^\alpha \eta_i,\quad \alpha = x,y,z,}
\end{equation}
together with the constraint $S_i^xS_i^yS_i^z = \frac\I2 \eta_i^x \eta_i^y \eta_i^z \eta_i= \frac\I8$. The gauge field is defined on the bonds as 
\begin{equation}\label{eq:gauge_field}
\mf u_{ij} :=  2\I \eta_i^\alpha \eta_j^\alpha
\end{equation}
if $i,j$ are connected by an $\alpha$-type bond. It is easy to check that for any bonds $\langle i,j \rangle_{\alpha(i,j)}$ and $\langle k,l \rangle_{\alpha(k,l)}$, the gauge field satisfies the constraint
\begin{equation}
    [\mf u_{ij}, \mf u_{kl}] = 0, \quad [\mf u_{ij}, H] = 0.
\end{equation}
This means that each operator $\mf u_{ij}$, together with the Hamiltonian can be simultaneously diagonalized in some choice of gauge. In particular, the eigenvalues of $\mf u_{ij}$ are $\pm 1 $ since $ \mf u_{ij} ^2=1$ (and which is where the name $\mathbb Z_2$ gauge field comes from). Therefore, we can treat $\mf u_{ij}$ as a classical variable and the Hamiltonian becomes quadratic in Majorana fermions which can be solved exactly in this background gauge field:
\begin{equation}
    H = 2\sum_{\langle i,j \rangle} \I J_{\alpha(i,j)} \mf u_{ij} \eta_i  \eta_j  = \sum_{i,j}\I J_{\alpha(i,j)} \mf u_{ij} \eta_i \eta_j .
\end{equation}
{In the second expression the sum over $i,j$ is over the two directed orientations of each nearest-neighbor bond.}
From \cite{Lieb1994,Kitaev2006,CuiCaoFan2010} we know that the ground state of the system is in the vortex-free sector which means we can work within the sector where $\mf u_{ij}=1=-\mf u_{ji}$ for all bonds, where $i$ is the site on the $A$ sublattice and $j$ is the site on the $B$ sublattice (this lattice as seen in \cref{fig:honeycomb} is bipartite and is actually bicolorable). The unit cell of the honeycomb lattice contains an $A$-color site and a $B$-color site. We denote the position of the unit cell by $\bv s$ and the two sites in the unit cell by $(\bv s,\mu)$ where $\mu = A,B$. The Hamiltonian can be written as
\begin{equation}\label{eq:kitaev_majorana}
    H = \I \sum_{\bv s, \mu, \bv t, \nu} J_{\alpha(\bv s, \mu; \bv t, \nu)}  \eta_{\bv s\mu} \eta_{\bv t\nu} .
\end{equation}

In particular, if $(\bv s, \mu)$ and $(\bv t, \nu)$ are not nearest neighbors, then $J_{\alpha(\bv s, \mu; \bv t, \nu)} = 0$. We see that because of the translation invariance, $\alpha(\bv s, \mu; \bv t, \nu)$ (and thus $J$) only depends on $\mu$, $\nu$ and $\bv t-\bv s$, therefore we can perform the Bloch--Floquet decomposition 
\begin{equation}
    J_{\bv s,\mu; \bv t, \nu} = J_{\mu\nu}(\bv t-\bv s) = \frac{1}{L^2} \sum_{\bv q} e^{-\I \bv q\cdot(\bv t-\bv s)} \wh J_{\mu\nu}(\bv q),\quad \eta_{\bv s, \mu}  = {\frac1{L}} \sum_{\bv q} e^{\I \bv q\cdot \bv s} a_{\bv q,\mu}.
\end{equation}
Here $\bv q = (q_x, q_y)$ goes through the discretized first Brillouin zone of the reciprocal lattice. Here $q_x$ and $q_y$ are the components of $\bv q$ in the basis of the reciprocal lattice vectors in the direction of $\bv n_1$ and $\bv n_2$ respectively. $q_x,q_y = \frac{2\pi n}{L}, n =-\frac{L-1}{2},\cdots, \frac{L-1}{2}$. We take $L$ to be odd to make it convenient for looking up the $-\bv q$ mode in the following calculations. 
The inverse transforms are given by
\begin{equation}
    \wh J_{\mu\nu}(\bv q) = \sum_{\bv t} e^{ \I \bv q\cdot \bv t} J_{\mu\nu}(\bv t),  \quad a_{\bv q, \mu} = {\frac1{L}  \sum_{\bv s} e^{-\I \bv q\cdot \bv s} \eta_{\bv s, \mu}}.
\end{equation}
The Fourier modes of the Majorana fermions satisfy the  anti-commutation relation, and the momentum-conserving $\bv q$-mode operators are related by a Hermitian conjugation:
\begin{equation}\label{eq:anticommutation}
    a_{\bv q, \mu}^\dagger = a_{-\bv q, \mu}, \quad \{a_{\bv q, \mu}, a_{\bv q', \nu}^\dagger\} = \delta_{\bv q, \bv q'}\delta_{\mu,\nu}, \quad \{a_{\bv q, \mu}, a_{\bv q', \nu}\} = \{a_{\bv q, \mu}^\dagger, a_{\bv q', \nu}^\dagger\} = \delta_{\mathbf q,-\mathbf q'}\delta_{\mu\nu},\text{~in particular,~} a_{\bv q, \mu}^2 = 0~\text{for}~\bv q\ne \bv 0.
\end{equation}
The transformed Hamiltonian can be written as
\begin{equation}
    H = \I \sum_{\bv q} \sum_{\mu,\nu \in \{A,B\}} \wh J_{\mu\nu}(\bv q) a_{\bv q, \mu}^\dagger a_{\bv q, \nu} = \I \sum_{\bv q} \sum_{\mu,\nu \in \{A,B\}} \wh J_{\mu\nu}(\bv q) a_{-\bv q, \mu} a_{\bv q, \nu}.
\end{equation}
Because $J_{0A;\bv tA}=J_{0B;\bv t B} =0$ (there is no bond there), we have $\wh J_{AA}(\bv q) = \wh J_{BB}(\bv q) = 0$. For $\wh J_{AB}(\bv q)$, only three values of $\bv t$ contribute: $\bv t = 0$, $\bv t = \bv n_1$ and $\bv t = \bv n_2$ where $\bv n_1$ and $\bv n_2$ are the two primitive lattice vectors of the honeycomb lattice. We have 
\begin{equation}\label{eq:J_AB_q}
    \wh J_{AB}(\bv q) = J_xe^{ \I \bv q\cdot \bv n_1}  + J_y e^{ \I \bv q\cdot \bv n_2}  +  J_z=: \varepsilon(\bv q) + \I \delta(\bv q)=: g(\bv q).
\end{equation}
\begin{equation}
    \varepsilon(\bv q) = J_x\cos q_x + J_y \cos q_y + J_z, \quad \delta(\bv q) = J_x\sin q_x + J_y \sin q_y,\quad 
    \wh J_{BA}(\bv q) = - \wh J_{AB}^*(\bv q) = -\varepsilon(\bv q) + \I \delta(\bv q).
\end{equation}
Correspondingly, we take $\bv s = (s_x,s_y) = s_x \bv n_1 + s_y \bv n_2$ where $s_x, s_y =-\frac{L-1}{2},\cdots, \frac{L-1}{2}$. 

We see that $\delta(\bv q)$ is odd and $\varepsilon(\bv q)$ is even in $\bv q$. Therefore $
    g(-\bv q) = g^\ast(\bv q).$ The Hamiltonian can be written as
\begin{equation}
    H = \I \sum_{\bv q} \left( g(\bv q) a_{ \bv q, A} ^\dag a_{\bv q, B} - g^*(\bv q) a_{ \bv q, B}^\dag a_{\bv q, A}\right).
\end{equation}
Applying the Bogoliubov transformation to decouple the $\bv q$ and $-\bv q$ modes in \cref{eq:anticommutation}, we obtain
\begin{equation}
    \mqty(c_{\bv q A} \\ c_{\bv q B}) = \mqty( u_{\bv q} & v_{\bv q} \\ v_{\bv q}^* & -u_{\bv q}^*) \mqty(a_{\bv q, A} \\ a_{\bv q, B}),\quad u_{\bv q} = \frac1{\sqrt{2}} ,\quad v_{\bv q} = \frac{\I}{\sqrt{2} } e^{\I \arg g(\bv q)},
\end{equation}
\begin{equation}
    H = \sum_{\bv q} \abs{g(\bv q)} (c_{\bv q A}^\dag c_{\bv q A} - c_{\bv q B}^\dag c_{\bv q B})  .
\end{equation}
The new operators also satisfy the canonical anti-commutation relations. The modes on two sublattices $A$ and $B$ are related by  
 $c_{-\bv q,A} =  2u_{\bv q}^\ast v_{\bv q}^\ast c_{\bv q,B}^\dag $ and {$c_{\bv qA}^\dag c_{\bv qA} = 1-c_{-\bv q,B}^\dag c_{-\bv q,B}$}. The first relation is because
\begin{equation}
    c_{-\bv q,A}  = u_{-\bv q} a_{\bv q,A}^\dag +v_{-\bv q} a_{\bv q,B}^\dag = u _{\bv q}^\ast a_{\bv q,A}^\dag - v_{\bv q}^\ast  a_{\bv q,B}^\dag
\end{equation}
and
\begin{equation}
    2u_{\bv q}^\ast v_{\bv q}^\ast c_{\bv q,B}^\dag = 2u_{\bv q}^\ast v_{\bv q}^\ast (v_{\bv q}  a_{\bv qA}^\dag - u_{\bv q} a_{\bv qB}^\dag) =2u_{\bv q}^\ast \left(\frac12 a_{\bv qA}^\dag - u_{\bv q} v_{\bv q}^\ast a_{\bv qB}^\dag \right) = u _{\bv q}^\ast a_{\bv q,A}^\dag - v_{\bv q}^\ast  a_{\bv q,B}^\dag.
\end{equation}
This implies
\begin{equation}\label{eq:CA}
    c_{\bv q, A} = 2u_{-\bv q}^\ast v_{-\bv q}^\ast c_{-\bv q,B}^\dag =  -2u_{\bv q} v_{\bv q}  c_{-\bv q,B}^\dag.
\end{equation}
{After relabeling $\bv q\mapsto -\bv q$ in the first occupation term,} we can express the Hamiltonian as
\begin{equation}\label{eq:kitaev_diagonalized}
    H = \sum_{\bv q}\abs{g(\bv q)}  (1-2c_{\bv qB}^\dag c_{\bv qB}) .
\end{equation}
Note that this Hamiltonian only depends on $c_{\bv q B}$ operators which should be understood as quasiparticle degrees of freedom. The ground state is given by the fully filled state of the $c_{\bv qB}$ modes, namely 
\begin{equation}
    \ket{\text{GS}} = \prod_{\bv q} c_{\bv qB}^\dag \ket{0}.
\end{equation}
{The system is gapless if and only if there exists some $\bv q \in \text{BZ}$ such that ${g}(\bv q) = 0$ in the continuum limit. 
This is equivalent to the inequalities $|J_x| \leq |J_y| + |J_z|$, $|J_y| \leq |J_x| + |J_z|$ and $|J_z| \leq |J_x| + |J_y|$. Outside this parameter regime, the system is gapped.}

\section{Details of  topological phase transitions in the  Kitaev honeycomb model}

\subsection{Chern number and purity gap}\label{appendix:topological}
We could probe the topological properties of the states along the dissipative dynamics by evaluating the Chern number of the covariance matrix $\Gamma$ in the real space, defined as 
\begin{equation}
  \Gamma_{\bv s \mu, \bv t \nu} = \frac{\mathrm{i}}{2}\Tr\left(\rho[\eta_{\bv s \mu},\eta_{\bv t \nu}]\right)  = \I \langle \eta_{\bv s \mu} \eta_{\bv t \nu} \rangle - \frac{\mathrm{i}}{2}\delta_{\bv s \mu, \bv t \nu},\quad \mu, \nu = A, B,
\end{equation}
which only involves the two-point correlation functions in the real space and can be efficiently measured in experiments such as quantum simulations on  neutral atom platform  \cite{Evered2025Honeycomb}. We will discuss the details of the simulation of the dissipative dynamics in \cref{appendix:kitaev_quasifree}.

 Because of translation invariance, we could write $\Gamma_{\bv s \mu, \bv t \nu} = \Gamma_{\bv 0\mu, \bv t-\bv s, \nu} =: \Gamma_{\mu\nu}(\bv t-\bv s)$. Thus we can define the $2\times 2$ Bloch covariance matrix $\Gamma(\bv q ) $ for each $\bv q$ in the first Brillouin zone, obtained from the Fourier transform
\begin{equation}\label{eq:bloch_covariance_matrix}
  \Gamma_{\mu\nu}(\bv q) = \sum_{\bv s} e^{\I \bv q\cdot \bv s} \Gamma_{\bv 0\mu,  \bv s\nu}.
\end{equation}
Here, with a slight abuse of notation, we use the same symbol $\Gamma$ to denote both the real-space covariance matrix and the Bloch representation in the momentum space. 

We define the parent Hamiltonian $H_{\rm p}$ associated with the covariance matrix $\Gamma$ as
\begin{equation}
  \begin{aligned}
  H_{\rm p}& = \sum_{\bv s, \bv t ; \mu,\nu \in\{A,B\}} \I \Gamma_{\mu\nu}(\bv t-\bv s) \eta_{\bv s \mu} \eta_{\bv t \nu} = \sum_{\bv q} \sum_{\mu,\nu\in\{A,B\}}2 \I \Gamma_{\mu\nu}(\bv q) a_{-\bv q\mu} a_{\bv q\nu} \\
  &=  \sum_{\bv q} \sum_{\mu,\nu\in\{A,B\}}2 \I \Gamma_{\mu\nu}(\bv q)  a_{\bv q\mu}^\dag a_{\bv q\nu} = \sum_{\bv q} \bv a_{\bv q}^\dag \mc H(\bv q) \bv a_{\bv q},
  \end{aligned}
\end{equation}
where the Bloch Hamiltonian $\mc H(\bv q)$ is a $2\times 2$ Hermitian matrix defined as
\begin{equation}
  \mc H (\bv q) := 2\I \Gamma(\bv q) = \mqty(2\I \Gamma_{AA}(\bv q) & 2\I \Gamma_{AB}(\bv q) \\ 2\I \Gamma_{BA}(\bv q) & 2\I \Gamma_{BB}(\bv q)),\quad \bv a_{\bv q} := \mqty(a_{\bv qA} \\ a_{\bv qB}).
\end{equation}
We consider the single occupied band of $\mc H(\bv q)$. Let $\ket{n(\bv q)}$ denote the normalized negative energy eigenstate of $\mc H(\bv q)$, then the Berry potential or connection $\mathscr{A}$ is defined as
\begin{equation}
  \mathscr A_\alpha(\bv q) =  \ip{n(\bv q)}{\partial_{q_\alpha}n(\bv q)},\quad \alpha = x, y,
\end{equation}
and the gauge-invariant Berry curvature, or the flux density is given by
\begin{equation}
  \mathscr F_{xy}(\bv q) = \partial_{q_x} \mathscr A_y(\bv q) - \partial_{q_y} \mathscr A_x(\bv q).
\end{equation}
The Chern number is then given by the integral of the Berry curvature over the first Brillouin zone:
\begin{equation}\label{eq:chern_berry}
  \mc C = \frac{1}{2\pi \I} \iint_{\rm BZ} \mathscr F_{xy}(\bv q) \dd q_x \dd q_y.
\end{equation}
The Chern number could also be formulated using the projection operator $\mc P(\bv q) = \ket{n(\bv q)}\bra{n(\bv q)}$ onto the single-occupied band 
\cite{Bardyn2013TopoDiss}. We note that for any pure state, the covariance matrix $\Gamma$ and thus the Bloch covariance matrix $\Gamma(\bv q)$ satisfies $\Gamma ^2 = -\frac14 I$. This implies that $\mc H(\bv q) =  2\I \Gamma(\bv q) $  satisfies $\mc H^2(\bv q) = I$. Therefore, $\mc H(\bv q) $ is a reflection and $\mc P(\bv q)$ could be directly defined as $\mc P(\bv q) = \frac{1}{2}(I - \mc H(\bv q))$ without explicitly computing the eigenstates. For general mixed states, we need to ``flatten'' the energy spectrum of $\mc H(\bv q)$ by applying the sign transform
\begin{equation}
  { \mc H}^\flat(\bv q) := \mathrm{sgn}(\mc H(\bv q))
\end{equation}
where $\mathrm{sgn}(\mc H(\bv q))$ is understood as a functional calculus with the $\mathrm{sgn}$ function applied to the eigenvalues of $\mc H(\bv q)$. The flattened $ { \mc H}^\flat(\bv q)$ is now a reflection and the projection operator $\mc P(\bv q)$ can be defined as $\mc P(\bv q) = \frac{1}{2}(I -  {\mc H}^\flat(\bv q))$.   
 The Chern number can then be computed as
\begin{equation}\label{eq:chern_projection}
  \mc C = \frac{1}{2\pi \I} \iint_{\rm BZ} \Tr\left(\mc P(\bv q) [\partial_{q_x} \mc P(\bv q), \partial_{q_y} \mc P(\bv q)]\right) \dd q_x \dd q_y.
\end{equation}
The formula \cref{eq:chern_projection} can be generalized to the multi-occupied band case, where $\mc P(\bv q)$ is the projection operator onto the occupied bands.

In practice, we are always working with a finite lattice, or equivalently, a discretized Brillouin zone. Due to the highly concentrated nature of the Berry curvature, a naive finite difference scheme for discretizing the Berry curvature may not be able to capture the spikes at the Dirac points, which could completely miss the topological properties. Instead, we adopt the more stable specialized Fukui--Hatsugai--Suzuki (FHS) 
scheme for computing the Chern number on a discretized Brillouin zone, as described in \cite{FukuiHatsugaiSuzuki2005}. For simplicity, we only consider the single-occupied case. The FHS scheme defines the link variable as
\begin{equation}
  \mathscr U_x(\bv q) = \frac{\ip{n(\bv q)}{n(\bv q + \bv x)}}{|\ip{n(\bv q)}{n(\bv q + \bv x)}|},\quad \mathscr U_y(\bv q) = \frac{\ip{n(\bv q)}{n(\bv q + \bv y)}}{|\ip{n(\bv q)}{n(\bv q + \bv y)}|},
\end{equation}  
where $\bv x$ and $\bv y$ denote {one mesh step in the $q_x$ and $q_y$ directions, respectively}. The discretized Berry curvature is then given by
\begin{equation}
  \mathscr F_{xy}(\bv q) = \log\left(\mathscr U_x(\bv q) \mathscr U_y(\bv q + \bv x) \mathscr U_x^{-1}(\bv q + \bv y) \mathscr U_y^{-1}(\bv q)\right),
\end{equation}
where the principal branch of the logarithm is taken to ensure that $\frac{1}{\I}\mathscr F_{xy}(\bv q)$ is in the range $(-\pi, \pi]$. The Chern number is then given by $\mc C=(2\pi\I)^{-1}\sum_{\bv q}\mathscr F_{xy}(\bv q)$. 
The FHS scheme is stable and gives an integer-valued lattice Chern number whenever the link variables are well defined. This allows us to reliably extract the topological properties even with a coarse discretization of the Brillouin zone.

It is also informative to analyze the purity spectrum of a state with covariance matrix $\Gamma$. We define the purity gap as the smallest eigenvalue of the positive semidefinite matrix $(2\I \Gamma)^2$
\begin{equation}
  \Delta_{P} (\Gamma)= \min \{\lambda \ge  0 : \lambda \in \mathrm{Spec}[(2\I \Gamma)^2]\}  = \min \{\lambda \ge  0 : \lambda \in \mathrm{Spec}(-4 \Gamma^2)\}.  
\end{equation}
In particular, for a pure state, $\Delta_P (\Gamma ) = 1$ and for a maximally mixed state, $\Delta_P(\Gamma) = 0$. 
As discussed in \cite{DiehlRicoBaranovEtAl2011,Bardyn2013TopoDiss}, two Gaussian states with correlation matrices $\Gamma$ and $\Gamma'$ are \emph{topologically equivalent} if they can be connected continuously without closing the bulk purity gap. Unlike the earlier work \cite{Bardyn2013TopoDiss}, we focus here on dynamical evolution rather than only the non-equilibrium steady states. A transient closing of the purity gap during the dynamics signals a change in topological class and thus marks a possible dynamical topological transition. Tracking the purity spectrum therefore provides complementary information to the Chern number and helps distinguish true topological changes from finite-size or numerical artifacts.

 \subsection{A simple mechanism for dissipative preparation of Chern insulators}\label{app:CI}

We explain how dissipative protocols can evade the unitary topological obstruction to preparing a Chern insulator. Consider a general two-band Chern insulator with Bloch Hamiltonian \cite{QiZhang2011}
\begin{equation}
   H(\bv q) = \mathbf{d}(\bv q)\cdot \boldsymbol{\sigma}
   = \mqty(d_z(\bv q) & d_x(\bv q) - \I d_y(\bv q) \\
   d_x(\bv q) + \I d_y(\bv q) & -d_z(\bv q)),\quad \boldsymbol{\sigma} = (\sigma_x, \sigma_y, \sigma_z),
\end{equation}
The Qi--Wu--Zhang model \cite{QiWuZhang2006}, the Haldane model \cite{Haldane1988}, and the Kitaev honeycomb model with explicitly broken time-reversal symmetry \cite{Kitaev2006} can all be written in this form. If the two bands are separated, then $\abs{\mathbf d(\bv q)}\neq 0$ throughout the Brillouin zone, and the topology is encoded by the normalized map
$\widehat{\mathbf d}(\bv q)=\mathbf d(\bv q)/\abs{\mathbf d(\bv q)}:\mathbb T^2\to\mathbb S^2$. {Here $\mathbb T^2: = (\RR/2\pi\ZZ)^2$ denotes the two-dimensional torus and $\mathbb S^2 := \{\abs{x}=1:x\in\RR^3\}$ denotes the two-dimensional sphere.} 
Equivalently, the occupied-band projection is
$\mc P_-(\bv q)=\frac12(I-\widehat{\mathbf d}(\bv q)\cdot\boldsymbol{\sigma})$, whose first Chern number is the degree of $\widehat{\mathbf d}$, up to the occupied-band sign convention.

For a free-fermionic Gaussian state with Bloch covariance matrix $\Gamma(\bv q)$, the associated parent Hamiltonian is $\mc H(\bv q)=2\I\Gamma(\bv q)$, which we write as
$\mc H(\bv q)=\mathbf h(\bv q)\cdot\boldsymbol{\sigma}$.
If the state is pure, then $\mc H(\bv q)^2=I$, equivalently $\abs{\mathbf h(\bv q)}=1$, and $
\mc P(\bv q)=\frac12(I-\mc H(\bv q))$ 
is a rank-one projection. Its Chern number is
\begin{equation}\label{eq:chern_projection_pure}
\mc C = \frac{1}{2\pi \I}\iint_{\rm BZ}
\Tr\left(\mc P(\bv q)[\partial_{q_x}\mc P(\bv q),\partial_{q_y}\mc P(\bv q)]\right)
\dd q_x\dd q_y .
\end{equation}
In the two-band case this equals the winding number, or degree, of
$\mathbf h:\mathbb T^2\to\mathbb S^2$. Thus the projection formula and the Bloch-sphere winding description give the same integer invariant for pure states. 
Let $\mc P_L(\bv q)$ denote the finite-size data on an $L\times L$ torus. By adding twists, $\mc P_L(\bv q)$ can be defined for continuous $\bv q\in\mathbb T^2$. When the covariance matrix is pure, $\mc P_L(\bv q)$ is a projection, and the finite-size Chern number $\mc C_L$ is defined by the same formula as in \cref{eq:chern_projection_pure}.

Under quasi-local unitary dynamics generated by a translation-invariant free-fermionic Hamiltonian, purity is preserved, so $\mc P_L(\bv q;t)$ remains a smooth family of projections. Quasi-locality gives uniform bounds, independent of $L$, on the relevant momentum derivatives over any finite time window $t\in[0,T]$ with $T$ fixed. Together with convergence of the finite-size covariance data and Chern integrands to the thermodynamic limit, this allows the use of dominated convergence and gives
$\mc C(t)=\lim_{L\to\infty}\mc C_L(t)$.
For each fixed $L$, the map $\bv q\mapsto\mc P_L(\bv q;t)$ is a smooth homotopy from $\mathbb T^2$ to the space of rank-one projections, naturally identified with $\mathbb S^2$. Since the Chern number is invariant under such homotopies,
$\mc C_L(t)=\mc C_L(0)$, and therefore $
\mc C(t)=\lim_{L\to\infty}\mc C_L(t)
=\lim_{L\to\infty}\mc C_L(0)
=\mc C(0).$ 
Thus quasi-local unitary evolution cannot transform a trivial pure state into a pure Chern insulator with a different Chern number in finite time.

Under quasi-local Lindblad dynamics, by contrast, the covariance matrix is generally mixed. Then $
\mc P_L(\bv q;t)=\frac12(I-\mc H_L(\bv q;t))$ 
need not be a projection. In the two-band parameterization, the Bloch vector
$\mathbf h_L(\bv q;t)$ is no longer confined to $\mathbb S^2$, but may move inside the Bloch ball
$\mathbb B^3=\{\abs{x}\le 1:x\in\RR^3\}$.
The eigenvalues of $\mc P_L(\bv q;t)$ are
$\frac12(1\pm\abs{\mathbf h_L(\bv q;t)})$, so $\mc P_L$ is a projection only when
$\abs{\mathbf h_L(\bv q;t)}=1$. 
Nevertheless, if $\mathbf h_L(\bv q;t)\neq 0$ for all $\bv q$, one can still define the topological invariant by normalizing
$\mathbf h_L\mapsto\mathbf h_L/\abs{\mathbf h_L}$, or equivalently by using the corresponding flattened projection, as discussed in Appendix \ref{appendix:topological}. This requires a nonzero purity gap $\Delta_P(t)>0$. As long as this gap remains open, the normalized projection gives a smooth homotopy
$\mathbb T^2\to\mathbb S^2$, with quasi-locality again ensuring bounded momentum derivatives, and its Chern number cannot change. A change is possible only when the purity gap closes, namely when
$\mathbf h_L(\bv q_*;t_*)=0$ at some momentum $\bv q_*$ and time $t_*$. At this point the normalized vector, and hence the flattened projection, is undefined. 
This is the topological mechanism by which quasi-local Lindblad dynamics can evade the unitary Chern-number obstruction \cite{Goldstein2019}: the mixed-state Chern number is not conserved because the state may leave the sphere of pure covariance matrices, pass through the interior of the Bloch ball, and become undefined at a purity-gap closing.

 \subsection{Probing the topological properties with realistic filters}\label{app:steady_state_topological}

 In \cref{sec:Kitaev}, we demonstrate that the ground-state topological phase of the Kitaev honeycomb model can be probed by computing the Chern number directly from finite-filter-resolution dissipative dynamics. Following the approach in \cite{Bardyn2013TopoDiss}, computing the Chern number only requires distinguishing the momentum-space fermionic modes, and the excited and ground states of each mode from the real-space covariance matrix, or the parent Hamiltonian. This allows one to construct the flattened single-particle Hamiltonian and subsequently evaluate the topological invariant.

Recall that we choose the set of coupling operators to be $\{\sqrt{2}\eta_{\bv s\lambda}\}_{\bv s,\lambda\in\{A,B\}}$, which implements fermionic bulk dissipation. Therefore, the mixing time of the dissipative dynamics toward the steady state is upper bounded by $\varO(\log L)$ (see Appendix \ref{appendix:free_fermion} and \cite[Theorem 1]{ZhanDingHuhnEtAl2025}).

For the steady-state characterization, the results of Appendix \ref{appendix:free_fermion} imply that the density matrix in the $c_{\bv qB}$-basis factorizes as
\begin{equation}\label{eq:steady_state_topo}
  \rho_{\rm ss} = \bigotimes_{ {2\abs{g(\bv q)}>\Delta}} \dyad{1}_{\bv qB}\;\bigotimes_{ {2\abs{g(\bv q)}\le\Delta}} \left((1-p_{\bv q})\dyad{0}_{\bv qB} + p_{\bv q}\dyad{1}_{\bv qB}\right).
\end{equation}
The occupation probability is
\begin{equation}
  p_{\bv q} = \frac{\wh f(-2\abs{g(\bv q)})^2}{\wh f(-2\abs{g(\bv q)})^2 + \wh f(2\abs{g(\bv q)})^2} = \langle c_{\bv qB}^\dag c_{\bv qB}\rangle_{\rm ss}.
\end{equation}
We introduce the Majorana basis via
\begin{equation}\label{eq:majorana_momentum_first}
  \mqty(c_{\bv qB}\\ c_{\bv qB}^\dag) = \frac{1}{\sqrt{2}}\mqty(1&-\I\\1&\I)\mqty(w_{\bv q,+}\\ w_{\bv q,-}),\qquad
  \mqty(w_{\bv q,+}\\ w_{\bv q,-}) = \frac{1}{\sqrt{2}}\mqty(1&1\\\I&-\I)\mqty(c_{\bv qB}\\ c_{\bv qB}^\dag).
\end{equation}
The covariance matrix in this Majorana basis is
\begin{equation}\label{eq:majorana_momentum_cov_first}
  \Gamma^w_{(\bv p,\bullet),(\bv q,\circ)} = \frac{\I}{2}\mathrm{Tr}\big(\rho\,[w_{\bv p,\bullet},w_{\bv q,\circ}]\big) = \I\langle w_{\bv p,\bullet}w_{\bv q,\circ}\rangle - \frac{\I}{2}\delta_{(\bv p,\bullet),(\bv q,\circ)},\quad \bullet,\circ=+,-.
\end{equation}
Hence the steady-state Majorana covariance matrix is block-diagonal in momentum:  
\begin{equation}
  \Gamma^w_{\rm ss} = \bigoplus_{ {2\abs{g(\bv q)}>\Delta}}\begin{pmatrix}0 & -1/2\\1/2 & 0\end{pmatrix}\;\bigoplus_{ {2\abs{g(\bv q)}\le\Delta}}\begin{pmatrix}0 & 1/2 - p_{\bv q}\\-1/2 + p_{\bv q} & 0\end{pmatrix} =: \bigoplus_{\bv q} \Gamma^w_{\rm ss}(\bv q).
\end{equation}
Because $\Gamma^w_{\rm ss}(\bv q)$ is unitarily related to the Bloch covariance matrix $\Gamma_{\rm ss}(\bv q)$ in \cref{eq:bloch_covariance_matrix}, the matrices $2\I\Gamma^w_{\rm ss}(\bv q)$ and $\mc H_{\rm ss}(\bv q)=2\I\Gamma_{\rm ss}(\bv q)$ share the same spectrum. Hence
\begin{equation}
  \lambda_\pm(\bv q) = \pm\left(2p_{\bv q}-1\right),
\end{equation}
with $p_{\bv q}=1$ for fully occupied modes. In the ideal-filter limit $\lambda_\pm(\bv q)=\pm1 $ for all $\bv q$, so the steady state is pure and coincides with the ground state. For a realistic filter we must ensure the bands are correctly resolved, i.e. {the positive and negative eigenvalues of $2\I\Gamma^w_{\rm ss}(\bv q)$ stay separated from zero}, equivalently $p_{\bv q}>1/2$ for all $\bv q$. 

For the gapped topologically trivial phase ($J_z>0.5, J_x = J_y = \frac12 (1-J_z)$), the minimal  {excitation energy $\min_{\bv q}2\abs{g(\bv q)}$} is lower bounded by a positive constant independent of the system size. Therefore, if we choose $\Delta$ to be smaller than this constant, the steady state {for an exact finite-resolution cutoff} is pure and coincides with the ground state. However, in the gapless topologically nontrivial phase ($J_z<0.5, J_x = J_y = \frac12 (1-J_z)$), the minimal {excitation energy $\min_{\bv q}2\abs{g(\bv q)}$} scales as $\varO(1/L)$ according to \cref{eq:J_AB_q}, therefore we cannot prepare the pure ground state even in the infinite-time limit of the Lindblad dynamics with any finite $\Delta$ for sufficiently large system size. This suggests that our 
 topological diagnosis essentially relies on the extended definition of the Chern number for mixed states.

Assuming the filter construction of Appendix \ref{appendix:filter} and accounting for measurement error in real-space correlations, we additionally require
\begin{equation}
  p_{\bv q}=\langle c_{\bv qB}^\dag c_{\bv qB}\rangle_{\rm ss} > \tfrac12 + \epsilon,
\end{equation}
where $\epsilon$ denotes the additive measurement error (including Fourier and Bogoliubov transformations). Equivalently,
\begin{equation}
  \frac{\gamma_{h,\bv q}^2}{\gamma_{c,\bv q}^2} := \frac{\wh f(2\abs{g(\bv q)})^2}{\wh f(-2\abs{g(\bv q)})^2} < \frac{1-2\epsilon}{1+2\epsilon}.
\end{equation}
It suffices to consider the smallest mode energy $\min_{\bv q}2\abs{g(\bv q)}=\varO(1/L)$ in the gapless topological regime. Then one can approximate
\begin{equation}
  \wh f(\pm2\abs{g(\bv q)}) \approx \wh f(0) \pm 2\wh f'(0)\abs{g(\bv q)} = \tfrac12 \mp \mc O\big(1/(L\Delta)\big).
\end{equation}
Hence the margin $p_{\bv q}-1/2$ at the lowest mode is only $\varO(1/(L\Delta))$, which implies the covariance matrix elements must be measured to $\varO(1/(L\Delta))$ precision, or $\varO(1/L)$ precision when $\Delta$ is fixed.

\section{Solving fermionic quasi-free Lindblad dynamics}
\label{appendix:quasifree}
In this section, 
we provide details on how to solve the Lindblad dynamics for fermionic quasi-free systems. In Appendix \ref{appendix:xx_bulk} we will discuss the XX model (i.e. the XXZ model in \cref{sec:XXZ} with $J_z=0$) with fermionic bulk cooling, and in Appendix \ref{appendix:kitaev_quasifree} we will discuss the Kitaev honeycomb model with local Majorana cooling.

A system is called quasi-free if the Hamiltonian is quadratic and the jump operators are linear in the fermionic creation and annihilation operators. For comprehensive and detailed reviews on quasi-free systems, we refer readers to Refs. \cite{Prosen2008,  ProsenZunkovic2010, BarthelZhang2022}. A key property of a quasi-free system is that the dynamics of the system can be fully characterized by the two-point correlation functions. In particular, if the system is particle-number preserving, it can be fully characterized only by the one-particle reduced density matrix (1-RDM).
\subsection{The XX model with bulk cooling}\label{appendix:xx_bulk}
\subsubsection{Solving the 1-RDM dynamics}
We recall that the XX model can be fermionized into a free fermionic model by the Jordan--Wigner transformation, and the resulting Hamiltonian is given by \cref{eq:XX_fermionized}. 
We consider the bulk cooling, where the coupling operators are chosen as the local fermionic creation and annihilation operators $\{c_i, c_i^\dag\}_{i=1}^L$. The dynamics of the system can be fully characterized by the 1-RDM $P(t)$,
 which consists of two-point correlation functions of the following form \cite{LiZhanLin2025}:
\begin{equation}
    P_{ij}(t) = \Tr(\rho(t) c_j^\dagger c_i),\quad i,j = 1,\cdots, L.
\end{equation}
{We consider both open boundary conditions (OBCs) and periodic boundary conditions (PBCs). In the PBC case, the wraparound term introduces 
\begin{equation}
    S_L^x S_1^x + S_L^y S_1^y = -\frac12 (c_L^\dag c_1 + c_1^\dag c_L) \exp(\I \pi \sum_j c_j^\dag c_j),
\end{equation}
and since we are working in the half-filling sector, we have $\exp(\I \pi \sum_j c_j^\dag c_j) = \exp(\I \pi L/2)   = (-1)^{L/2}  $. Thus the coefficient matrix of the Hamiltonian is given by}
\begin{equation}
    F ={\frac{J_{xy}}{2}} \begin{pmatrix}0 & 1 & 0 & \cdots & p\\
1 & 0 & 1 & \cdots & 0 \\
0 & 1 & 0 & \cdots & 0 \\
\vdots & \vdots & \vdots & \ddots & 1 \\
p & 0 & 0 & 1 & 0 \end{pmatrix}
\text{\quad where\quad }
    p = \begin{cases}
        1 & \text{{PBCs}, and $L= 2 \mod 4 $;}\\
        -1 & \text{{PBCs}, and $L = 0 \mod 4$;}\\
        0 & \text{{OBCs}.}
    \end{cases}
\end{equation}
For simplicity, we only discuss the case where $L$ is a multiple of $4$ and take $J_{xy} = 1$. Then we have the following explicit form of the diagonalization of $F = U\Lambda U^\dagger$. For OBC,
\begin{equation}
 \Lambda = \mathrm{diag}(\varepsilon_i) = \mathrm{diag}\mqty(-  \cos\frac{\pi}{L+1},\cdots - \cos \frac{L\pi}{L+1}),\quad U = \mqty(
    \sqrt{\frac{2}{L+1}} \sin \frac{j(L+1-k)\pi}{L+1}
)_{j,k=1}^L,
\label{eq:fourier_tightbinding}
\end{equation}
For PBC with $L = 0 \mod 4$,
\begin{equation}
  \Lambda = \diag(\varepsilon_i) = \diag
\mqty(- \cos\frac{\pi}{L},- \cos \frac{\pi}{L},- \cos \frac{3\pi}{L}, - \cos \frac{3\pi}{L},\cdots, - \cos \frac{(L-1)\pi}{L}, - \cos \frac{(L-1)\pi}{L} ),
\end{equation}
\begin{equation}
    U =   \begin{pmatrix} \vert & \vert & & \vert \\ u_1 & u_2 & \cdots & {u_L} \\ \vert & \vert & & \vert \end{pmatrix} ,\quad u_{2m-1}=\mqty(\sqrt{\frac2L}\cos \frac{(L-j)(L-2m+1)}{L}\pi)_{j=1}^L,\quad u_{2m} = \mqty(\sqrt{\frac2L}\sin \frac{(L-j)(L-2m+1)}{L}\pi)_{j=1}^L.
\end{equation}
We define the transformed fermionic modes as 
\begin{equation}
    \tilde{c}_k = \sum_{j=1}^L U_{jk}^*  c_j,\quad \tilde{c}_k^\dagger = \sum_{j=1}^L U_{jk} c_j^\dagger,\quad H = \sum_{k=1}^L \varepsilon_k \tilde{c}_k^\dagger \tilde{c}_k,
\end{equation}
and  assume that $L$ is even for simplicity, then 
\begin{equation}
   \varepsilon_1\le \cdots \le \varepsilon_{L/2} < 0 < \varepsilon_{L/2+1} \le \cdots \le \varepsilon_L.
\end{equation}
Thus the ground state of the system is given by the half-filling state
\begin{equation}\label{eq:halffilled}
    \ket{\psi_0} = \prod_{k=1}^{L/2} \tilde{c}_k^\dagger \ket{0}
\end{equation}
and, {for OBC,} the spectral gap of the Hamiltonian  
is given by 
\begin{equation}
    \Delta_H = \varepsilon_{L/2+1} - \varepsilon_{L/2 } = 2\sin\frac{\pi}{2(L+1)}\sim \frac1L \quad (L\to \infty).
\end{equation}

Using \cref{eq:thouless_fermion}, the jump operators can be expressed as 
\begin{equation}
    \begin{aligned}
    K_i^+ &:= \int_{\RR} f(s) c_i^\dagger(s) \dd s = \sum_{j=1}^L \int_{\RR} f(s) (e^{\I Fs } )_{ji} c_j^\dagger \dd s = \sum_{j=1}^L \wh{f}( F)_{ji} c_j^\dagger,\\
    K_i^- &:= \int_{\RR} f(s) c_i(s) \dd s = \sum_{j=1}^L \int_{\RR} f(s) (e^{-\I Fs } )_{ij} c_j \dd s = \sum_{j=1}^L \wh{f}(-F)_{ij} c_j.
    \end{aligned}
\end{equation}
Therefore, the jump operators are also linear in the fermionic creation and annihilation operators, and the resulting Lindblad dynamics is quasi-free. In the Heisenberg picture, the Hilbert--Schmidt adjoint of $\mc L$ is given by
\begin{equation}
    \mc L^\dagger [O] = \I [H,O] + \sum_{i=1}^L \left( K_i^{+\dagger} O K_i^+ - \frac{1}{2} \{K_i^{+\dagger} K_i^+, O\} + K_i^{-\dagger} O K_i^- - \frac{1}{2} \{K_i^{-\dagger} K_i^-, O\} \right).
\end{equation}
Plugging in $O = c_j^\dagger c_i$, using the canonical anticommutation relations (CAR) to simplify, and then taking the expectation value with respect to the state $\rho(t)$, we get a closed-form equation for the 1-RDM $P(t)$ \cite{LiZhanLin2025}
\begin{equation}\label{eq:quasifree_dynamics}
    \frac{\dd}{\dd t} P(t) = -\I [F, P(t) ] + M^+ -\frac12\{P(t), M^++M^-\}
\end{equation}
where the matrices $M^+$ and $M^-$ are given by
\begin{equation}
    (M^+)_{pq} = \sum_{k} \wh f(F)_{pk} {\wh f(F)^*_{qk}},\quad (M^-)_{pq} = \sum_{k} \wh f(-F)_{pk} {\wh f(-F)^*_{qk}},\quad \text{i.e.} \quad M^{\pm} = \wh f(\pm F)  \wh f(\pm F) ^\dag.
\end{equation}
We choose the ``pure N\'eel state'' as the initial state. 
In the fermionic language, it can be expressed as $\ket{\psi_{\text{ini}}} = \prod_{j=1}^{L/2} c_{2j-1}^\dagger \ket{0}_{c}$, and the corresponding 1-RDM is given by
\begin{equation}\label{eq:neel_1RDM}
    P_{ij}(0) = \Tr(\ket{\psi_{\text{ini}}}\bra{\psi_{\text{ini}}} c_j^\dagger c_i) = \delta_{ij} \cdot \begin{cases}1 & i \text{ is odd} \\ 0 & i \text{ is even}\end{cases}.
\end{equation}
Therefore $\Tr P(0) = L/2$ and the initial state is half-filled.

In this case, the initial state is a Slater determinant state and in particular, a Gaussian pure state. Since the dynamics is quasi-free, the state at any time $t$ is also Gaussian and the observables of interest can be computed from the 1-RDM $P(t)$ using Wick's theorem. For the spin correlation function $\langle S_i^z S_j^z \rangle$, when $i\ne j$,
\begin{equation}
    \langle S_i^z S_j^z \rangle = \left\langle \left(c_i^\dagger c_i - \frac12\right)\left(c_j^\dagger c_j - \frac12\right)\right \rangle = \langle n_i \rangle \langle n_j\rangle -\langle c_i^\dag c_j\rangle\langle c_j^\dag c_i\rangle -\frac12(\langle n_i\rangle+\langle n_j\rangle)+\frac14 = P_{ii} P_{jj} - |P_{ij}|^2 -\frac12(P_{ii}+P_{jj})+\frac14,
\end{equation}
and when $i=j$, using $n_i^2 = n_i$,
\begin{equation}
    \langle S_i^z S_i^z \rangle = \left\langle \left(c_i^\dagger c_i - \frac12\right)^2 \right\rangle = \langle n_i\rangle -\langle n_i\rangle +\frac14 = \frac14.
\end{equation}
Together with $\langle S_i^z\rangle = \langle n_i\rangle -\frac12 = P_{ii} -\frac12$, we can obtain the connected spin correlation function $\langle S_i^z S_j^z \rangle - \langle S_i^z\rangle \langle S_j^z\rangle$ which is needed to compute the subsystem bipartite spin fluctuation $\mathcal{F}(\ell)$.

We are also interested in the fidelity of the evolved state $\rho(t)$ with respect to the true ground state \cref{eq:halffilled}, whose 1-RDM is given by
\begin{equation}
    P_{ij}^{\text{GS}} = \bra{\psi_{\rm GS}} c_j^\dagger c_i \ket{\psi_{\rm GS}} = \sum_{k,\ell=1}^L U_{jk}^* U_{i\ell} \bra{\psi_{\rm GS}} \tilde{c}_k^\dagger \tilde{c}_\ell \ket{\psi_{\rm GS}} = \sum_{k=1}^{L/2} U_{jk}^* U_{ik}.
\end{equation}
We note that in the limit of ideal filtering, we have $\wh f(-F) = U \mqty(0&0\\0&I_{\frac L2})U^\dagger$ and $\wh f(F) = U\mqty(I_{\frac L2}&0\\0&0) U^\dagger$, thus the steady state $P_{\rm ss}$ of \cref{eq:quasifree_dynamics} satisfies 
\begin{equation}
  -  \I [F, P_{\rm ss}] + U\mqty(I_{\frac L2}&0\\0&0) U^\dagger-\frac12\{P_{\rm ss}, I_L\} = 0.
\end{equation}
The solution of the Sylvester-type equation is given by
\begin{equation}
    P_{\rm ss} = U\mqty(I_{\frac L2}&0\\0&0) U^\dagger= P^{\text{GS}}.
\end{equation}
Note that for free-fermionic mixed states, the {trace distance to the ground state can be controlled by the trace-norm difference of the corresponding covariance data, see} \cite[Theorem 1]{BittelMeleEisertLeone2025}. 
Therefore we can study the convergence of the evolved state to the true ground state by looking at the convergence of the 1-RDM $P(t)$ to $P^{\text{GS}}$ in the trace norm, as is done in \cref{fig:dissipative_XX_trace}.

\subsubsection{Derivation of error in 1-RDM}\label{appendix:1rdm_error}
In this subsection we compute the difference between dissipative steady state and ground state expectation values of the 1-RDM $\braket{c_i^{\dagger} c_j}$.

Using \cref{eq:fourier_tightbinding} for OBC, we can write
\begin{equation}
    \begin{aligned}
        \tr \left( c_i^{\dagger} c_j \rho \right) = & \sum_{k,l} U_{ik}^{*}U_{jl} \tr \left( \tilde{c}_k^{\dagger}\tilde{c}_l  \rho \right) = \sum_k U_{ik}^{*}U_{jk} \braket{\wt n_k}\\
        = & \sum_k \frac{2}{L+1} \sin \left( \frac{i(L+1 - k)}{L+1}\pi \right) \sin \left( \frac{j(L+1 - k)}{L+1}\pi \right) \braket{ \wt n_k}\\
        = & \frac{1}{L+1} \sum_k \left[ \cos \left( \frac{(i-j)(L+1 - k)}{L+1}\pi \right) - \cos \left( \frac{(i+j)(L+1 - k)}{L+1}\pi \right) \right] \braket{\wt n_k}.
    \end{aligned}
\end{equation}
Since $\braket{\wt n_k}_{\rm ss} - \braket{\wt n_k}_{\rm GS} = 0$ for all $|\varepsilon_k| > \Delta$, we have
\begin{equation}
    \begin{aligned}
        \braket{c_i^{\dagger} c_j}_{\rm ss} - \braket{c_i^{\dagger} c_j}_{\rm GS} & = \frac{1}{L+1} \sum_{|\varepsilon_k| < \Delta} \left[ \cos \left( \frac{(i-j)(L+1 - k)}{L+1}\pi \right) - \cos \left( \frac{(i+j)(L+1 - k)}{L+1}\pi \right) \right] \left(\braket{\wt n_k}_{\rm ss} - \braket{\wt n_k}_{\rm GS}\right)\\
        & \sim \frac{1}{\pi} \int_{- \arcsin(\Delta  )}^{\arcsin(\Delta )}  \left[\cos\left( (i-j)(\pi/2-x) \right) - \cos\left( (i+j)(\pi/2-x) \right)\right] \varphi(x) \dInt x
    \end{aligned}
\end{equation}
where 
$\varphi(x) := \widehat f(x)^2 / \left( \widehat f(x)^2 + \widehat f(-x)^2\right) - \mathbf{1}_{(-\infty,0]}(x)$ accounts for the difference between the realistic and ideal filters. In particular, for  
$q_k := k\pi / (L+1) - \pi / 2$, we have $\varphi (q_k)=\left(\braket{\wt n_{k}}_{\rm ss} - \braket{\wt n_{k}}_{\rm GS}\right)$ and {$\varphi$ is supported in $(-\Delta,\Delta)$ for an exact finite-resolution cutoff.}

According to \cref{eq:steady_state_free_fermion}, we have $\varphi(x) = -\varphi(-x)$ and $|\varphi(x)| \le 1$. When we focus on the bulk of the system, i.e., {when $i+j$ is large}, the second term with $i+j$ will vanish due to the Riemann--Lebesgue lemma. For the first term, when $i - j$ is odd, we have
\begin{equation}
    \begin{aligned}
        &\left|\int_{- \arcsin(\Delta )}^{\arcsin(\Delta )}  \cos\left( (i-j)(\pi/2-x) \right) \varphi(x)  \dInt x \right| \\
        = &\left|\int_{- \arcsin(\Delta )}^{\arcsin(\Delta )}  \sin\left( (i-j)x\right)\varphi(x)\dInt x \right| = \varO(\Delta^2 |i-j|),
    \end{aligned}
\end{equation}
for small $\Delta$. When $i-j$ is even, then the integral vanishes since $\varphi(x) = -\varphi(-x)$. Therefore we obtain that $|\braket{c_i^{\dagger} c_j}_{\rm ss} - \braket{c_i^{\dagger} c_j}_{\rm GS}| = \varO(\Delta^2|i-j|)$ for small $\Delta$.

\subsection{The Kitaev honeycomb model with local Majorana cooling}\label{appendix:kitaev_quasifree}
\subsubsection{Solving the covariance matrix dynamics}\label{appendix:kitaev_quasifree_dynamics}

We next discuss the Kitaev honeycomb model, where the jump operators are chosen as the local Majorana operators $\{\eta_{\bv sA}
,\eta_{\bv sB}\}_{\bv s}$. We recall that, after fixing the $\ZZ_2$ gauge and diagonalizing the Hamiltonian as discussed in Appendix \ref{app:kitaev_diagonalization}, the Kitaev honeycomb model can be written as a quadratic fermionic form \cref{eq:kitaev_majorana}. They are related to the fermionic operators in the momentum space as
\begin{equation}\label{eq:jump_majorana}
   { \eta_{\bv s A} =  \frac 1L  \sum_{\bv q} e^{\I \bv q\cdot \bv s} v_{\bv q} (c_{\bv qB}-c_{-\bv q B}^\dagger),\quad \eta_{\bv sB} = - \frac 1L  \sum_{\bv q} e^{\I \bv q\cdot\bv s} u_{\bv q} (c_{\bv qB}+c_{-\bv q B}^\dagger)},\quad u_{\bv q} = \frac{1}{\sqrt{2}},  \quad {v_{\bv q} = \frac{\I e^{\I \arg g(\bv q)}}{\sqrt{2}}}.
\end{equation}
The inverse transformation is given by
\begin{equation}\label{eq:CB}
    c_{\bv qB} = v_{\bv q}^\ast a_{\bv qA} - u_{\bv q}^\ast a_{\bv qB} = \frac1L \sum_{\bv s} e^{-\I \bv q\cdot \bv s} (v_{\bv q}^\ast \eta_{\bv sA} - u_{\bv q}^\ast \eta_{\bv sB}) = -\frac1L\sum_{\bv s} e^{-\I \bv q\cdot \bv s}  \left( v_{-\bv q}  \eta_{\bv sA} + u_{-\bv q}  \eta_{\bv sB}\right).
\end{equation}
\begin{equation}\label{eq:CBdag}
    c_{\bv qB}^\dag = v_{\bv q} a_{-\bv q A} - u_{\bv q} a_{-\bv q B} =    \frac1L\sum_{\bv s} e^{\I \bv q\cdot \bv s}  \left( v_{\bv q}  \eta_{\bv sA} - u_{\bv q}  \eta_{\bv sB}\right).
\end{equation}
Here the jump operators are hybrid linear combinations of the fermionic creation and annihilation operators, thus the resulting Lindblad dynamics is still quasi-free. However, the dynamics of the 1-RDM is not closed anymore since the {cross} terms that contain the ``anomalous'' correlation functions $\langle c_{\bv qB}c_{\bv q'B}\rangle$ cannot be omitted. Therefore, we need to consider the covariance matrix that contains more information than the 1-RDM \cite{BarthelZhang2022,ZhanDingHuhnEtAl2025}. For this, we define the Majorana operators and the corresponding Majorana covariance matrix in the momentum space as follows (see also \cref{eq:majorana_momentum_first})
\begin{equation}\label{eq:majorana_momentum}
    \mqty(c_{\bv qB}\\ c_{\bv qB}^\dag) = \frac1{\sqrt{2} } \mqty(1&-\I \\1&\I) \mqty(w_{\bv q,+}\\ w_{\bv q,-}),\quad \mqty(w_{\bv q,+}\\ w_{\bv q,-}) = \frac1{\sqrt{2}} \mqty(1&1\\\I&-\I)\mqty(c_{\bv qB}\\ c_{\bv qB}^\dag),
\end{equation}
\begin{equation}\label{eq:gamma_w_from_w}
    \Gamma^w_{(\bv p, \bullet), (\bv q,\circ)} = \frac{\I}{2} \mathrm{Tr}(\rho [w_{\bv p,\bullet}, w_{\bv q,\circ}]) = \I\langle w_{\bv p,\bullet} w_{\bv q,\circ}\rangle -\frac{\I}{2}\delta_{(\bv p,\bullet), (\bv q,\circ)},\quad  \bullet,\circ = +,-.
\end{equation}
With a slight abuse of notation, we drop the superscript $w$ and denote the momentum-space covariance matrix by $\Gamma = \left(\frac{\I}{2 } \langle[ w_{\bv p,\bullet}, w_{\bv q,\circ}]\rangle\right)_{(\bv p,\bullet), (\bv q,\circ)}$ throughout the following discussion for brevity. This notation should be distinguished from the real-space covariance matrix $\Gamma = \left(\frac{\I}{2 }\langle [\eta_{\bv s\mu},\eta_{\bv t\nu}]\rangle\right)_{\bv s\mu, \bv t\nu}$ discussed in \cref{appendix:topological}.

The Hamiltonian can be rewritten as 
\begin{equation}
    {H = \sum_{\bv q}\left({\I} \abs{g(\bv q)} w_{\bv q,+} w_{\bv q,-} {-} \I\abs{g(\bv q)} w_{\bv q,-} w_{\bv q,+}\right)}
\end{equation}
Accordingly, the coefficient matrix here is given by 
\begin{equation}
    {G_{(\bv q,+), (\bv q,-)} =  {\I} \abs{g(\bv q)},\quad G_{(\bv q,-), (\bv q,+)} = {-}\I\abs{g(\bv q)},\quad G_{(\bv p,\bullet), (\bv q,\circ)} = 0 \text{ for $\bv p \ne \bv q$}.}
\end{equation}
The matrix $G$ is Hermitian and anti-symmetric (in particular, it is purely imaginary). Now using the Heisenberg evolution of the Majorana operators
\begin{equation}
  w_i(t) = e^{\I Ht} w_i e^{-\I Ht} = \sum_j (e^{-2\I G t})_{ij} w_j,\quad\text{for} \quad H = \sum_{i,j } G_{ij} w_{i} w_j,
\end{equation}
together with \cref{eq:jump_majorana,eq:majorana_momentum}, we can express the jump operators as a linear combination of Majorana operators in the momentum space as
\begin{equation}
    \begin{aligned}
K_{\bv s A}  &= \int_{\RR} \sqrt 2\eta_{\bv sA}(s) f(s)\dd s \\&
= \frac{1}{L} \sum_{\bv q, \bv p,\bullet\in \{\pm\}} e^{\I \bv q\cdot  \bv s} v_{\bv q}  \left( -\wh f(-2G)_{(-\bv q,+), (\bv p, \bullet)} -\I \wh f(-2G)_{(-\bv q, -), (\bv p, \bullet)} +  \wh f(-2G)_{(\bv q, +), (\bv p, \bullet)} -\I  \wh f(-2G)_{(\bv q, -), (\bv p, \bullet)}\right) {w_{(\bv p,\bullet)}}
    \end{aligned}
\end{equation}
\begin{equation}
    \begin{aligned}
    K_{\bv sB} &= \int_{\RR} \sqrt 2\eta_{\bv sB}(s) f(s)\dd s \\& =\frac{1}{L} \sum_{\bv q, \bv p,\bullet\in \{\pm\}} e^{\I \bv q\cdot  \bv s}u_{\bv q}  \left( -\wh f(-2G)_{(-\bv q,+), (\bv p, \bullet)} -\I \wh f(-2G)_{(-\bv q, -), (\bv p, \bullet)} -  \wh f(-2G)_{(\bv q, +), (\bv p, \bullet)} +\I  \wh f(-2G)_{(\bv q, -), (\bv p, \bullet)}\right) {w_{(\bv p,\bullet)}}.
    \end{aligned}
\end{equation}
We define 
\begin{equation}
    \begin{aligned}
E_{(\bv p,\bullet),\bv sA} &:= \frac1L \sum_{\bv q} e^{\I \bv q\cdot  \bv s}v_{\bv q}  \left( -\wh f(-2G)_{(-\bv q,+), (\bv p, \bullet)} -\I \wh f(-2G)_{(-\bv q, -), (\bv p, \bullet)} +  \wh f(-2G)_{(\bv q, +), (\bv p, \bullet)} -\I  \wh f(-2G)_{(\bv q, -), (\bv p, \bullet)}\right),\\
    E_{(\bv p,\bullet),\bv sB} &:= \frac1L \sum_{\bv q} e^{\I \bv q\cdot  \bv s} u_{\bv q} \left( -\wh f(-2G)_{(-\bv q,+), (\bv p, \bullet)} -\I \wh f(-2G)_{(-\bv q, -), (\bv p, \bullet)} -  \wh f(-2G)_{(\bv q, +), (\bv p, \bullet)} +\I  \wh f(-2G)_{(\bv q, -), (\bv p, \bullet)}\right),
    \end{aligned}
\end{equation}
and
\begin{equation}
    D_{(\bv p, \bullet), (\bv p', \bullet')} =  \sum_{\bv s } E_{(\bv p,\bullet),\bv sA} E_{(\bv p',\bullet'),\bv sA}^\ast + E_{(\bv p,\bullet),\bv sB} E_{(\bv p',\bullet'),\bv sB}^\ast.
\end{equation}
Then the dynamics of the covariance matrix is given by \cite[Proposition 1]{BarthelZhang2022}
\begin{equation}
    \dv{\Gamma}{t} = X\Gamma + \Gamma X^T + Y,\quad X = -2\I G - D^R, \quad Y =   D^I.
\end{equation}
Here, $D^R$ and $D^I$ are the \emph{element-wise} real and imaginary parts of $D$, respectively. 
The full-filled state corresponds to 
\begin{equation}\label{eq:cov_fullfilled}
    {\Gamma^{\rm GS} = \mqty(0& - I_{L^2\times L^2}/2\\    I_{L^2\times L^2}/2& 0).}
\end{equation}
Again, according to \cite[Theorem 1]{BittelMeleEisertLeone2025}, the distance to the ground state in the trace norm of the many-body density matrix {is controlled by the covariance-matrix discrepancy}, which justifies the study of the convergence of the covariance matrix to $\Gamma^{\rm GS}$ in \cref{fig:dissipative_2Dc}.
\subsubsection{The $z$-bond correlation function}\label{appendix:kitaev_zbond}
Next, we discuss how to compute the $z$-bond correlation function $\langle S_{\bv sA}^z S_{\bv sB}^z \rangle$ in terms of the covariance matrix. We first express the $z$-bond correlation function in terms of the fermionic operators in the momentum space as

\begin{equation}
    \begin{aligned}
    \langle  S_{\bv sA}^z S_{\bv sB}^z \rangle  &= \langle \eta_{\bv s A}^z \eta_{\bv s B}^z \frac{{4}}{L^2} \sum_{\bv q,\bv q'} e^{\I (\bv q+\bv q')\cdot \bv s} a_{\bv q A} a_{\bv q' B} \rangle \note{\cref{eq:gauge_field}}= - \frac\I{L^2} \sum_{\bv q,\bv q'} e^{\I(\bv q+\bv q')\cdot \bv s}\langle a_{\bv {q} A}a_{\bv q'B}\rangle 
    \end{aligned}
\end{equation}
where, by \cref{eq:CA}, we have
\begin{equation}
    \begin{aligned}
    \langle a_{\bv qA} a_{\bv q'B}\rangle%
    & =  
\frac{\I}{2}(e^{\I \arg g(\bv q)}\langle c_{-\bv qB}^\dag c_{-\bv q' B}^\dag \rangle +e^{\I \arg g(\bv q)} \langle c_{-\bv qB}^\dag c_{\bv q'B}\rangle  - e^{\I (2\arg g(\bv q) -\arg g(\bv q'))}\langle c_{\bv qB} c_{-\bv q' B}^\dag \rangle- {e^{\I \arg g(\bv q)}}  \langle c_{\bv qB} c_{\bv q'B}\rangle ).
    \end{aligned}
\end{equation}
In particular, if the state is the full-filled ground state, then only the $\langle c_{-\bv qB}^\dag c_{\bv q'B}\rangle$ term does not vanish and we have
\begin{equation}
    \langle S_{\bv sA}^z S_{\bv sB}^z \rangle =  \frac{1}{2L^2} \sum_{\bv q,\bv q'} \exp(\I [(\bv q+\bv q')\cdot \bv s + \arg g(\bv q)])\langle c_{-\bv qB}^\dagger  c_{\bv q'B}\rangle.
\end{equation}
Using the transformation between complex fermions and Majorana fermions \cref{eq:majorana_momentum}, we could express the correlation functions as follows
\begin{equation}
    \begin{aligned}
    \langle S_{\bv sA}^z S_{\bv sB}^z \rangle
    & =    \frac{1}{4L^2} \sum_{\mathbf{q},\mathbf{q}'} e^{\I(\mathbf{q}+\mathbf{q}')\cdot \mathbf{s} + {\I \arg g(\mathbf{q})}} \Bigg(  \mathsf{W}_{(-\mathbf{q},+),(-\mathbf{q}',+)} + \I\mathsf{W}_{(-\mathbf{q},+),(-\mathbf{q}',-)} + \I\mathsf{W}_{(-\mathbf{q},-),(-\mathbf{q}',+)} - \mathsf{W}_{(-\mathbf{q},-),(-\mathbf{q}',-)} \\
&\quad + \mathsf{W}_{(-\mathbf{q},+),(\mathbf{q}',+)} - \I\mathsf{W}_{(-\mathbf{q},+),(\mathbf{q}',-)} + \I\mathsf{W}_{(-\mathbf{q},-),(\mathbf{q}',+)} + \mathsf{W}_{(-\mathbf{q},-),(\mathbf{q}',-)} \\
&\quad - \mathsf{W}_{(\mathbf{q},+),(-\mathbf{q}',+)} - \I\mathsf{W}_{(\mathbf{q},+),(-\mathbf{q}',-)} + \I\mathsf{W}_{(\mathbf{q},-),(-\mathbf{q}',+)} - \mathsf{W}_{(\mathbf{q},-),(-\mathbf{q}',-)} \\
&\quad - \mathsf{W}_{(\mathbf{q},+),(\mathbf{q}',+)} + \I\mathsf{W}_{(\mathbf{q},+),(\mathbf{q}',-)} + \I\mathsf{W}_{(\mathbf{q},-),(\mathbf{q}',+)} + \mathsf{W}_{(\mathbf{q},-),(\mathbf{q}',-)} \Bigg) 
    \end{aligned}
\end{equation}
Here we introduce the notation $\mathsf W_{(\bv p,\bullet), (\bv q,\circ)} = \langle w_{(\bv p,\bullet)} w_{(\bv q,\circ)}\rangle$ for $\bullet,\circ \in \{+,-\}$ for brevity.
 In particular, if the state is the full-filled ground state, it can be simplified to
\begin{equation}
    \langle S_{\bv sA}^z S_{\bv sB}^z \rangle =  \frac{1}{4L^2} \sum_{\bv q,\bv q'} e^{\I [(\bv q+\bv q')\cdot \bv s + \arg g(\bv q)]}   \left(\mathsf{W}_{(-\mathbf{q},+),(\mathbf{q}',+)} - \I\mathsf{W}_{(-\mathbf{q},+),(\mathbf{q}',-)} + \I\mathsf{W}_{(-\mathbf{q},-),(\mathbf{q}',+)} + \mathsf{W}_{(-\mathbf{q},-),(\mathbf{q}',-)} \right).
\end{equation}

If we want to simulate the cooling dynamics starting from the all-spins-up state, i.e. the ground state of the system in the limit $J_x=J_y=0$, $J_z=1$, we could express a ground state of Hamiltonian in one parameter point $J_{z,1}$ in the basis of $c_{\bv qB}^{(2)}$ operators which correspond to the ground state in another parameter point $J_{z,2}$. To be specific, these two ground states are given by
\begin{equation}
    \ket{\text{GS}_1} = \prod_{\bv q} c_{\bv qB}^{(1)\dag} \ket{0},\quad \ket{\text{GS}_2} = \prod_{\bv q} c_{\bv qB}^{(2)\dag} \ket{0}.
\end{equation}
And we compute the covariance matrix of $\ket{\text{GS}_1}$ in the basis of $c_{\bv qB}^{(2)}$ operators. For simplicity, we denote $c_{\bv qB} = c_{\bv qB}^{(2)}$ and $c_{\bv qB}^\dag = c_{\bv qB}^{(2)\dag}$ and $\wt c_{\bv qB} = c_{\bv qB}^{(1)}$ and $\wt c_{\bv qB}^\dag = c_{\bv qB}^{(1)\dag}$. Similarly, all of the $J_z$-dependent variables in $J_{z,1}$ case are denoted with tilde.

We first express $w$ in terms of $\wt w$. We do this following the chain of transformations $w\to c\to \eta \to \wt c \to \wt w$, using \cref{eq:CB,eq:CBdag,eq:jump_majorana,eq:majorana_momentum}. Since performing the sum over $\bv s$ gives
$
\frac{1}{L^2}\sum_{\bv s} e^{\I(\bv q+\bv k)\cdot\bv s}
=\delta_{\bv k,-\bv q}$ and $
\frac{1}{L^2}\sum_{\bv s} e^{\I(\bv k-\bv q)\cdot\bv s}
=\delta_{\bv k,\bv q},
$ we have

\begin{equation}
\begin{aligned}
w_{(\bv q,+)}
&=\frac12\Big[
v_{\bv q} \wt v_{-\bv q}(-\wt w_{(\bv q,+)}-\I\wt w_{(\bv q,-)}+\wt w_{(-\bv q,+)}-\I\wt w_{(-\bv q,-)})
-v_{-\bv q}\wt v_{\bv q}(-\wt w_{(-\bv q,+)}-\I\wt w_{(-\bv q,-)}+\wt w_{(\bv q,+)}-\I\wt w_{(\bv q,-)})
\\
&\qquad
-u_{-\bv q}\wt u_{\bv q}(-\wt w_{(-\bv q,+)}-\I\wt w_{(-\bv q,-)}-\wt w_{(\bv q,+)}+\I\wt w_{(\bv q,-)})
-u_{\bv q} \wt u_{-\bv q}(-\wt w_{(\bv q,+)}-\I\wt w_{(\bv q,-)}-\wt w_{(-\bv q,+)}+\I\wt w_{(-\bv q,-)})
\Big].
\end{aligned}
\end{equation}
Grouping the terms according to $\wt w_{(\pm \bv q,\pm)}$ we introduce the coefficients
\begin{equation}
\begin{aligned}
\mathsf A_{\bv q} &=\frac12\Big(-
v_{\bv q}\wt v_{-\bv q}-v_{-\bv q}\wt v_{\bv q}
+u_{-\bv q}\wt u_{\bv q}+u_{\bv q}\wt u_{-\bv q}\Big),\quad \mathsf  B_{\bv q}  =\frac12\Big(
v_{\bv q}\wt v_{-\bv q}+v_{-\bv q}\wt v_{\bv q}
+u_{-\bv q}\wt u_{\bv q}+u_{\bv q}\wt u_{-\bv q}\Big),\\
\mathsf C_{\bv q} &=\frac12\Big(-
v_{\bv q}\wt v_{-\bv q}+v_{-\bv q}\wt v_{\bv q}
-u_{-\bv q}\wt u_{\bv q}+u_{\bv q}\wt u_{-\bv q}\Big),\quad \mathsf D_{\bv q}  =\frac12\Big(
v_{\bv q}\wt v_{-\bv q}-v_{-\bv q}\wt v_{\bv q}
-u_{-\bv q}\wt u_{\bv q}+u_{\bv q}\wt u_{-\bv q}\Big).
\end{aligned}
\end{equation}
Thus we have the final linear mapping from $\wt w$ to $ w$:
\begin{equation}
    \mqty(w_{ \bv q,+ } \\ w_{ \bv q,- }) =  \mqty(\mathsf A_{\bv q} & \I \mathsf C_{\bv q} \\
-\I \mathsf C_{\bv q} & \mathsf A_{\bv q} )\mqty(\wt w_{(\bv q,+)} \\ \wt w_{(\bv q,-)} ) + \mqty(\mathsf B_{\bv q} & -\I \mathsf D_{\bv q} \\-\I\mathsf D_{\bv q} & -  \mathsf B_{\bv q} )\mqty(\wt w_{(-\bv q,+)} \\ \wt w_{(-\bv q,-)} )=: \mathsf T_{\bv q} \mqty(\wt w_{(\bv q,+)} \\ \wt w_{(\bv q,-)} ) + \mathsf S_{\bv q} \mqty(\wt w_{(-\bv q,+)} \\ \wt w_{(-\bv q,-)} ).
\end{equation}
Consequently, the two-point correlation function transforms as follows:
\begin{equation}\label{eq:WtoW}
\mathsf{W}_{\bv q,\bv q'} = \mathsf X_{\bv q} \wt {\mathsf{W}}_{\bv q,\bv q'} \mathsf X_{\bv q'}^\dag ,\quad \mathsf X_{\bv q} = \begin{pmatrix} \mathsf T_{-\mathbf{q}} &  \mathsf S_{-\mathbf{q}} \\ \mathsf S_{\mathbf{q}} &  \mathsf T_{\mathbf{q}} \end{pmatrix}.
\end{equation}
Here we introduce
\begin{equation}\label{eq:Wqqp}
    \quad \mathsf{W}_{\mathbf{q}, \mathbf{q}'} = 
\begin{pmatrix}
\mathsf{W}_{(-\mathbf{q},+),(-\mathbf{q}',+)} & \mathsf{W}_{(-\mathbf{q},+),(-\mathbf{q}',-)} & \mathsf{W}_{(-\mathbf{q},+),(\mathbf{q}',+)} & \mathsf{W}_{(-\mathbf{q},+),(\mathbf{q}',-)} \\
\mathsf{W}_{(-\mathbf{q},-),(-\mathbf{q}',+)} & \mathsf{W}_{(-\mathbf{q},-),(-\mathbf{q}',-)} & \mathsf{W}_{(-\mathbf{q},-),(\mathbf{q}',+)} & \mathsf{W}_{(-\mathbf{q},-),(\mathbf{q}',-)} \\
\mathsf{W}_{(\mathbf{q},+),(-\mathbf{q}',+)} & \mathsf{W}_{(\mathbf{q},+),(-\mathbf{q}',-)} & \mathsf{W}_{(\mathbf{q},+),(\mathbf{q}',+)} & \mathsf{W}_{(\mathbf{q},+),(\mathbf{q}',-)} \\
\mathsf{W}_{(\mathbf{q},-),(-\mathbf{q}',+)} & \mathsf{W}_{(\mathbf{q},-),(-\mathbf{q}',-)} & \mathsf{W}_{(\mathbf{q},-),(\mathbf{q}',+)} & \mathsf{W}_{(\mathbf{q},-),(\mathbf{q}',-)} 
\end{pmatrix}.
\end{equation}
The expectation values for the ground state $\ket{\wt{\text{GS}}}$ evaluate to
\begin{equation}\label{eq:GS_W_W_GS}
 \wt{\mathsf W}_{(\bv q,\circ),(\bv q',\circ')}=    \mel{\wt{\text{GS}}  }{ \wt{w}_{(\bv q,\circ)} \wt{w}_{(\bv q',\circ')}  }{\wt{\text{GS}}}= - \I \wt{\Gamma}_{(\bv q,\circ),(\bv q',\circ')}^{\rm GS} + \frac12 \delta_{(\bv q,\circ),(\bv q',\circ')},
\end{equation}
where the ground-state covariance matrix $\wt{\Gamma}^{\rm GS}$ in this basis takes the block form
\begin{equation}
    {\wt{\Gamma}^{\rm GS} =  \mqty(0& - I_{L^2\times L^2}/2\\    I_{L^2\times L^2}/2& 0),}
\end{equation}
which characterizes a fully occupied Fermi sea for the $\wt c$ fermions (refer to \cref{eq:cov_fullfilled}). By combining \cref{eq:WtoW} and \cref{eq:GS_W_W_GS}, we can thus explicitly construct the covariance matrix of $\ket{\wt{\rm GS}}$ in the {$w$-Majorana basis associated with the untilded Hamiltonian}.

\subsubsection{The covariance matrix in the original real space basis}\label{app:cov_real_space}

Note by \cref{eq:majorana_momentum,eq:CB,eq:CBdag} that we can obtain the transformation from the $w$-Majorana basis (where we perform the Lindblad simulation as discussed in \cref{appendix:kitaev_quasifree_dynamics}) to the $\eta$-Majorana basis in the real space (where we propose to carry out the construction of the parent Hamiltonian and computation of the topological diagnostics, etc.) as follows: 
 \begin{equation}\label{eq:w_to_eta}
    \begin{aligned}
    \eta_{\bv sA}   = 
    \frac1{\sqrt{2L^2}} \sum_{\bv q} e^{\I\bv q \cdot \bv s}  v_{\bv q} \left(-  w_{-\bv q+} -\I   w_{-\bv q-} + w_{\bv q+} -\I w_{\bv q-} \right)
\\
          \eta_{\bv sB}  
         = \frac1{\sqrt{2L^2}} \sum_{\bv q} e^{\I\bv q \cdot \bv s}  u_{\bv q} \left(- w_{-\bv q+} -\I w_{-\bv q-} - w_{\bv q+} +\I w_{\bv q-} \right)
    \end{aligned}
 \end{equation}
 Since the Majorana operators satisfy $
  \{\eta_{\bv s,\mu},\eta_{\bv s',\mu'}\}  =\delta_{\bv s,\bv s'} \delta_{\mu,\mu'} $ 
 we can compute the covariance matrix in the $\eta$-Majorana basis, i.e. the real space covariance matrix as follows:
 \begin{equation}
  \Gamma_{\bv s\mu,\bv s'\mu' } = \frac{\I}{2}  {\langle[\eta_{\bv s,\mu}, \eta_{\bv s',\mu'}]\rangle} = \I \langle \eta_{\bv s,\mu} \eta_{\bv s',\mu'} \rangle - \frac{\I}{2} \delta_{\bv s,\bv s'} \delta_{\mu,\mu'}.
 \end{equation}
Substituting the transformations into the definition of the covariance matrix, we can express the components of $\Gamma$ directly in terms of the expectation values of the $w$-Majoranas in momentum space. By expanding the product $\langle \eta_{\mathbf{s}\lambda} \eta_{\mathbf{s}'\lambda'} \rangle$ using \cref{eq:w_to_eta}, we obtain the four block components:
\begin{equation}
   \Gamma _{\mathbf{s}, \mathbf{s}'} = 
\begin{pmatrix} 
\Gamma_{\mathbf{s} A, \mathbf{s}' A} & \Gamma_{\mathbf{s} A, \mathbf{s}' B} \\ 
\Gamma_{\mathbf{s} B, \mathbf{s}' A} & \Gamma_{\mathbf{s} B, \mathbf{s}' B} 
\end{pmatrix},
\end{equation}
We then have
\begin{equation}
  \Gamma_{\mathbf{s}, \mathbf{s}'} = \frac{\I}{L^2} \sum_{\mathbf{q}, \mathbf{q}'} e^{\I(\mathbf{q} \cdot \mathbf{s}  - \mathbf{q}' \cdot \mathbf{s}')} \mathsf{R}_{\mathbf{q}} \mathsf{W}_{\mathbf{q}, \mathbf{q}'} \mathsf{R}_{\mathbf{q}'}^\dag - \frac{\I}{2} \delta_{\mathbf{s}, \mathbf{s}'}I,\quad \text{where}\quad  \mathsf{R}_{\mathbf{q}} = {\frac{1}{\sqrt{2}}}
\begin{pmatrix} 
v_{\mathbf{q}} & \I v_{\mathbf{q}} & -v_{\mathbf{q}} & \I v_{\mathbf{q}} \\ 
u_{\mathbf{q}} & \I u_{\mathbf{q}} & u_{\mathbf{q}} & -\I u_{\mathbf{q}} 
\end{pmatrix}.
\end{equation}
Here $\mathsf W_{\bv q,\bv q'}$ is defined in \cref{eq:Wqqp} and is directly related to the $w$-Majorana covariance matrix by \cref{eq:gamma_w_from_w}.

\end{document}